\documentclass[11pt,
 superscriptaddress,
 nofootinbib,
 notitlepage, floatfix, tightenlines,
  amsmath,amssymb,
 prx,
 ]{revtex4-2}

\usepackage{graphicx}% Include figure files
\usepackage{dcolumn}% Align table columns on decimal point
\usepackage{bm}% bold math
\usepackage{amssymb,dsfont}

\usepackage[ruled, lined, linesnumbered, commentsnumbered, longend]{algorithm2e}

\usepackage{amsmath}   % For mathematical environments and symbols
\usepackage{amsthm}
\usepackage[compat=0.6]{yquant}  % For drawing quantum circuits (yquant 0.7.3+)
\usepackage{braket}

\usepackage[utf8]{inputenc} % allow utf-8 input
\usepackage[T1]{fontenc}    % use 8-bit T1 fonts
\usepackage[hidelinks]{hyperref}      % hyperlinks
\usepackage{url}            % simple URL typesetting
\usepackage{booktabs}       % professional-quality tables
\usepackage{amsfonts}       % blackboard math symbols
\usepackage{nicefrac}       % compact symbols for 1/2, etc.
\usepackage{microtype}      % microtypography
\usepackage{physics}
\usepackage{placeins}
\usepackage{bbold}

\usepackage{tikz}
\usepackage{forest}
\usepackage{multirow}
\usetikzlibrary{calc,arrows.meta,positioning, shapes.geometric, arrows, shapes.symbols,shapes.callouts,patterns,quotes,decorations.pathreplacing,fit,backgrounds}

\tikzset{
    every node/.style={font=\sffamily\small},
    main node/.style={thick,circle ,draw},
    visible node/.style={thick,rectangle ,draw}
}
\usepackage{cleveref}
\usepackage{xcolor}
\usepackage{enumitem}
\usepackage{subcaption}
\usepackage{floatrow}
\usepackage{mathtools}
\usepackage{verbatim}
\usepackage{upgreek}
\usepackage{textgreek}
\usepackage{thmtools}
\usepackage{thm-restate}
\usepackage{fancyhdr}
\usepackage{fontawesome}

\newtheorem{theorem}{Theorem}
\newtheorem{corollary}[theorem]{Corollary}
\newtheorem{proposition}[theorem]{Proposition}
\newtheorem{definition}[theorem]{Definition}
\newtheorem{lemma}[theorem]{Lemma}
\theoremstyle{remark} % Sets the style for subsequent theorem-like environments
\newtheorem{remark}{Remark} % Defines the 'remark' environment
\DeclareMathOperator{\EX}{\mathbb{E}}

\begin{document}

%\preprint{APS/123-QED}  % Replace with arXiv ID before submission
\title{Encoded Quantum Signal Processing for Heisenberg-Limited Metrology}
\author{Carlos Ortiz Marrero}
\affiliation{Physical Detection Systems and Deployment Division, Pacific Northwest National Laboratory,  Richland, WA 99354}
\affiliation{Department of Computer Science, Colorado State University,  Fort Collins, CO 80521}
\email{carlos.ortiz.marrero@colostate.edu}
\author{Rui Jie Tang}
\email{ruijie.tang@mail.utoronto.ca}
\affiliation{Department of Physics, University of Toronto, ON M5S 1A1, Canada}
\author{Nathan Wiebe}
\affiliation{ Department of Computer Science, University of Toronto, ON M5S 1A1, Canada}
\affiliation{Advanced Computing, Mathematics, and Data Division, Pacific Northwest National Laboratory, Richland, WA 99354}
\email{nawiebe@cs.toronto.edu}



\date{\today}% It is always \today, today,

\begin{abstract}
Entangled quantum probes can achieve Heisenberg-limited measurement precision, but this advantage is typically destroyed by noise. We address this issue by introducing a framework that we call encoded quantum signal processing, which unifies quantum error detection and quantum signal processing into an effective single-qubit framework, and provides a paradigm for constructing logical sensors that are robust to noise while remaining sensitive to the signal of interest. We show that encoding sensor qubits into a repetition code and using syndrome measurements as a signal-processing primitive restores Heisenberg scaling under realistic noise, without applying recovery operations. We prove that product-state sensing with syndrome post-processing is fundamentally limited to standard quantum limit (SQL) scaling, and develop four protocols that overcome this barrier through entanglement or sequential signal amplification, achieving Heisenberg-limited precision with exponential error suppression in code distance. For spatially inhomogeneous fields, Bayesian marginalization preserves Heisenberg scaling provided noise decreases sufficiently with system size. The underlying mechanism, which we formalize as encoded quantum signal processing, reduces multi-qubit metrology to an effective single-qubit problem where syndrome measurement implements nonlinear signal transformations. Numerical simulations validate the theoretical predictions: syndrome-based inference achieves near-Heisenberg scaling at noise levels where bare probes approach the SQL, and a concatenated protocol maintains this scaling under joint transverse noise and longitudinal inhomogeneities.
\end{abstract}

\maketitle

%\tableofcontents

\section{Introduction}
Three quantum technologies have ignited the interest of the quantum community: quantum error correction~\cite{gottesman1997stabilizer}, quantum signal processing~\cite{low2017optimal,gilyen2019quantum}, and quantum sensing \cite{giovannetti2004quantum, giovannetti2006quantum, kitagawa1993squeezed, gerry2000heisenberg}. At the core of these technologies is the desire to protect, control and extract information from quantum systems as efficiently as possible. However, the connections between these technologies have only begun to be explored. 
Specifically, quantum error correction provides ways of encoding logical qubits inside a code space that is specifically designed to be resilient to errors. Quantum error correction is an indispensable technique for addressing noise in quantum systems, but is less well suited to manipulating transformations within the protected space of the code. Quantum Signal Processing on the other hand, provides a technique for modifying transformations acting on a logical space but is incapable of addressing noise. These differences raise the question of whether there exists a larger framework that is capable of encompassing both quantum signal processing and quantum error correction.
Our work suggests a framework that unifies aspects of quantum signal processing and quantum error correction which in turn we show can be used to design experiments in metrology that allow us to trade off robustness and sensitivity of experiments to optimize learning in noisy environments.

It is often observed that a poor qubit makes a good sensor, and in some sense this intuition is true. A poor qubit is often coupled strongly to the environment and cannot keep its information for long before decohering; however, this also means that such a device can be highly sensitive to its surroundings. In order to build a perfect sensor for a given signal, one wishes to build a qubit that is highly sensitive to that signal while insensitive to other confounding information in the environment (such as orthogonal fields, time-dependent drift and other phenomena). A challenge with such a design is that finding an optimal physical qubit for the effect that one wishes to measure can be challenging. Rather than trying to build an optimal physical qubit, our approach is to construct a logical qubit that provides the desired robustness and then use quantum signal processing to manipulate the resultant signal.

This intuition that quantum error correction may be useful for quantum metrology has been explored in a number of different works.
Several important results have established the scope and limits of quantum error correction (QEC)-enhanced metrology~\cite{zhou2021error, kapourniotis2019fault}. Zhou et al.~\cite{zhou2018achieving} and Zhou and Jiang~\cite{zhou2018optimal} identified the Hamiltonian-not-in-Lindblad-span (HNLS) condition as the necessary and sufficient criterion for achieving Heisenberg scaling with full quantum error correction under Markovian noise. Sekatski et al.~\cite{sekatski2017quantum} proved that noiseless ancillae are required for QEC-enhanced metrology to surpass the SQL, while Layden et al.~\cite{layden2019ancilla} showed that approximate error correction can restore Heisenberg scaling even without noiseless ancillae when the noise rate is sufficiently small. Recent work by Sahu et al.~\cite{sahu2026heisenberg} achieves Heisenberg-limited sensing with fault-tolerant syndrome extraction using the repetition code, complementing the error-detection approach developed here. Recent work on quantum computational sensing~\cite{allen2025quantum,khan2025quantum} has further demonstrated that interleaving sensing with computation can provide metrological speedups. Our approach differs from these works in two key respects: (i) we use error \emph{detection} rather than full error correction, extracting the signal directly from syndrome-conditioned parity measurements without applying recovery operations, and (ii) we address spatially inhomogeneous noise through Bayesian marginalization, a regime not treated by the HNLS framework, which assumes identical Markovian channels on each qubit.

Our main contributions in this work involve the proposal of a framework that unifies quantum signal processing and error correction that we call Encoded Quantum Signal Processing (EQSP) and discuss issues surrounding this in Section~\ref{sec:problem}. We further show the following applications of these ideas for metrology: (i) bit-flip repetition codes with syndrome-dependent phase extraction for transverse noise (Section~\ref{sec:fixed-omega}), (ii) an adaptive binary search with GHZ states for spatially inhomogeneous fields (Section~\ref{sec:binary-search}), (iii) a concatenated protocol combining both error types (Section~\ref{sec:combined}), and (iv) a product-state sequential binary search that achieves Heisenberg scaling through sequential signal applications without entangled state preparation (Section~\ref{subsec:sequential-binary-search}). We first establish in Theorem~\ref{thm:sql-barrier} that the simplest approach (product-state input, single signal application with syndrome post-processing) is fundamentally SQL-limited, regardless of classical post-processing. The four protocols overcome this barrier through different resources: entanglement, sequential signal amplification, or their combination. We demonstrate that these protocols achieve the Heisenberg limit ($\alpha=1$) via quantum Fisher information bounds (Proposition~\ref{prop:cramer-rao}), subject to a device quality condition $N\sigma_\epsilon^2 = o(1)$ when field inhomogeneities are present; here $\sigma_\epsilon^2$ denotes the variance of qubit-to-qubit field variations. This condition, which requires noise to decrease with system size, distinguishes our results from no-go theorems under fixed-rate Markovian decoherence~\cite{demkowicz2012elusive}. Numerical simulations characterize finite-size behavior and practical noise thresholds (Section~\ref{sec:numerical}).

%\section{Background}

\section{Encoded QSP}
\label{sec:problem}

The central idea of this work is to construct a single logical qubit from physical sensor qubits such that a multi-qubit metrology problem reduces to an effective single-qubit signal processing problem. We call this framework \emph{encoded quantum signal processing} (Encoded QSP). In standard QSP~\cite{low2017optimal,gilyen2019quantum} and its generalization GQSP~\cite{motlagh2023generalized}, polynomial transformations of a signal unitary are implemented through alternating signal and processing rotations on a single qubit. Encoded QSP extends this paradigm: the signal unitary acts individually on $N$ physical qubits forming a single encoded logical qubit, and syndrome measurement replaces explicit processing rotations, implementing nonlinear signal transformations on the logical qubit. This perspective unifies quantum error detection with signal processing and reveals metrology as a natural application of the QSP framework.

\begin{definition}[Encoded Quantum Signal Processing]\label{def:encoded-qsp}
    Let $\mathcal{C}$ be an $[[n,1,d]]$ quantum code with $\mathcal{H}_L$ being its logical codespace and let $Z_L: \mathcal{H}_L \rightarrow \mathcal{H}_L$ be the logical Pauli-$Z$ operation within. Let $\mathcal{D}$ be a decoding channel  that maps from density matrices on the $n$-qubit space to the logical qubit space (i.e., $L(\mathbb{C}^{2^n}) \rightarrow L(\mathcal{H}_L)$) and let $\mathcal{E}$ be an encoding (or error) channel that maps from $L(\mathcal{H}_L)\rightarrow L(\mathbb{C}^{2^n})$. Next let us define $H\in L(\mathbb{C}^{2^n\times 2^{n}})$ to be the signal Hamiltonian. An encoded quantum signal processing experiment in the $Z$-basis then takes the form, for $\theta_i \in \mathbb{R}$,
    $$\Lambda_{EQSP}(\theta_i) :\rho \mapsto 
    %e^{-i Z_L \theta_i}\mathcal{D}\left(  \left(  \bigotimes_{j=0}^{n-1} e^{-i H_j} \right)  \mathcal{E}(\rho)  \left(\bigotimes_{j=0}^{n-1} e^{i H_j} \right)\right)e^{i Z_L \theta_i}
    e^{-i Z_L \theta_i}\mathcal{D}\left(   e^{-iH}   \mathcal{E}(\rho)  e^{iH} \right)e^{i Z_L \theta_i}
    $$
    and define for $\theta\in \mathbb{R}^{D+1}$
    $$
    \Lambda_{EQSP}(\theta) : \rho \mapsto  \Lambda_{EQSP}(\theta_D) \circ \Lambda_{EQSP}(\theta_{D-1}) \circ \cdots \circ \Lambda_{EQSP}(\theta_1)(e^{-i \theta_0 Z_L} \rho e^{i \theta_0 Z_L})   
    $$
\end{definition}
This definition clearly reduces to the conventional quantum signal processing problem when $\mathcal{C}$ is a $[[1,1,1]]$ code and $\mathcal{E}$ is taken to be the identity mapping, or in other words when there is no distinction between logical and physical qubits and no errors. An important note to make is that in the majority of our discussion here, we will not necessarily assume that the quantum channels are in fact noisy. This may seem like an odd choice because typically quantum error correcting codes are only used to correct imperfect quantum devices. In this context we will primarily be using them to control quantum systems rather than to correct errors.

Further, encoded quantum signal processing generalizes quantum error correction. Indeed, for the case where $\theta_i = 0$, the protocol corresponds to simply applying a quantum error correcting channel $D$ times. As such encoded quantum signal processing is a proper generalization of ordinary quantum signal processing as well as single-qubit quantum error correction.

One important deficiency that EQSP resolves in ordinary quantum signal processing is that ordinary quantum signal processing cannot change the von Neumann entropy of a state fed into a channel. This means that QSP cannot ever increase the purity of a state once it becomes mixed. By incorporating error correction into the QSP sequence we can achieve this. We explicitly demonstrate this below.

\begin{proposition}\label{prop:ECQSP}
    Let $\mathcal{C}$ be the $[[2L+1,1,2L+1]]$ phase flip code and let $\mathcal{E}: \rho \mapsto (1-p)\rho + \frac{p}{2L+1}\sum_{j=0}^{2L} Z_j \rho Z_j$ be a single-error phase flip channel (each term applies at most one $Z_j$ error) and let $H = \arccos(x) X_L$ be the signal Hamiltonian where $X_L = X^{\otimes 2L+1}$ is the logical not operation for this code. We then have that there exists an EQSP protocol that uses $\theta \in \mathbb{R}^D$ that can implement a transformation between density operators in $\mathcal{H}_L$ and density operators in $\mathcal{H}_L$ that is of the form $U_L : \rho \mapsto U(x) \rho U^\dagger(x)$ where
    $$
    U(x) = \begin{bmatrix}
        P(x) & -i\sqrt{1-x^2} Q(x) \\
        -i\sqrt{1-x^2} Q^*(x) & P^*(x)
    \end{bmatrix}
    $$
    for any polynomial $P,Q$ that satisfy
    \begin{enumerate}
        \item ${\rm deg}(P) \le D\quad {\rm deg}(Q) \le D-1.$
        \item The polynomial $P$ is even and $Q$ is odd or vice versa.
        \item for all $x$ $|P(x)|^2 +(1-x^2)|Q(x)|^2 =1$
\end{enumerate}
further the procedure uses at most $3D$ queries to the signal unitary $e^{i H}$ to achieve this transformation.
\end{proposition}
\begin{proof}
    The proof of this proposition is straightforward and mostly involves choosing a decoding channel $\mathcal{D}$. In this setting, let us consider first a decoding channel $\mathcal{D}'$ consisting of measurement of the syndromes of the phase flip code. As the error channel $\mathcal{E}$ only makes a single error and the distance of the code is $2L+1 > 1$, this ensures that the maximum likelihood estimate of the error will guarantee successful correction of the error. However, correcting the error does not guarantee that the desired phase is applied.

    To see this, let us assume that a single iterate of EQSP is applied to a state and an error is detected on the first qubit. This results in a transformation of the form
    \begin{align}
        \rho &\mapsto e^{-i Z_L \theta_i}\mathcal{D}'\left(  e^{-i \arccos(x) X_L}  (Z\otimes I)\rho (Z\otimes I)  e^{i \arccos(x) X_L}\right)e^{i Z_L \theta_i}\nonumber\\
        &= e^{-i Z_L \theta_i} \left(e^{iX_L \arccos(x)} \rho e^{-iX_L \arccos(x)}\right)e^{i Z_L \theta_i}
    \end{align}
where the last line follows from the fact that $\{X_L,Z_0\otimes I\}=0$ and hence the phase flip reverses the direction of the logical rotation. This can be corrected by changing our decoding channel to yield
\begin{equation}
    \mathcal{D} : \rho \mapsto \begin{cases}
        \mathcal{D}(\rho) & \text{if no error detected}\\
        e^{-2i H}\mathcal{D}(\rho)e^{2i H} & \text{otherwise}
    \end{cases}
\end{equation}
Here we have constrained our channel to have only a single error so these are the only two cases that we need to consider. The worst case scenario involves three queries to $e^{iH}$ in the event that an error is detected. Thus we can implement the transformation
\begin{equation}
    \rho \mapsto e^{-iZ_L \theta} e^{-i \arccos(x) X_L} \rho e^{i \arccos(x) X_L} e^{iZ_L \theta} 
\end{equation}
using three queries to the signal $e^{-iH}$.

The transformation above is simply a single iteration of quantum signal processing in the $Z-X$ basis. Thus the achievability of transformations within this logical subspace coincides with that of ordinary quantum signal processing within the same logical space. The conditions for this then follow from standard results on quantum signal processing~\cite{gilyen2019quantum,martyn2021grand}. The number of queries being $3D$ then follows from the fact that implementing a degree $D$ polynomial $P(x)$ that satisfies the three conditions outlined above requires $D$ queries to the quantum signal processing iteration. We have from above that $3$ queries are needed to the signal Hamiltonian to implement each and thus the transformation requires $3D$ queries to the signal Hamiltonian as claimed.
\end{proof}

The above transformation shows how we can generalize quantum signal processing to deal with an encoding channel that does not preserve von Neumann entropy. A more significant application though involves building transformations that may not be easy to implement within the quantum signal processing formalism by using ideas from quantum error correction. The following proposition provides such a family, which provides the primary motivation for our applications in metrology.

\begin{proposition}\label{prop:arctan}
Let $\mathcal{C}$ be the $[2L+1,1,2L+1]$ repetition code, $\mathcal{E}=I$, $\mathcal{D}$ be the standard error correction channel for the code and $\delta>0$ be a failure probability.
Further let $P,Q$ be polynomials such that for any $y\in [-1,1]$
\begin{enumerate}
    \item ${\rm deg}(P) \le D$ and ${\rm deg}(Q) \le D-1$.
    \item The parity of $P$ is even and $Q$ is odd or vice-versa\\
    \item $|P(y)|^2 + \sqrt{1-y^2}|Q(y)|^2=1$ for all $y\in [-1,1]$.
\end{enumerate}
Finally let $e^{-iH}=\prod_j e^{-i X_j \arccos(x)}$ be a signal unitary with $H = \sum_j \arccos(x) X_j$ for $|x|\in [0,\sqrt{1-(1-\delta)^{1/(2L+1)}}]\cup [(1-\delta)^{1/(4L+2)},1]$. The logical unitary
\begin{equation}
    {\tiny
        W_L:=\begin{bmatrix}
            {P}\left( \frac{x^{2L}}{\sqrt{x^{2(2L+1)}+(1-x^2)^{2L+1}}}\right) & -i\sqrt{1-\frac{x^{4L}}{x^{2(2L+1)}+(1-x^2)^{2L+1}} } Q\left(\frac{x^{2L}}{\sqrt{x^{2(2L+1)}+(1-x^2)^{2L+1}}} \right)\\
            -i\sqrt{1-\frac{x^{4L}}{x^{2(2L+1)}+(1-x^2)^{2L+1}} } Q^*\left(\frac{x^{2L}}{\sqrt{x^{2(2L+1)}+(1-x^2)^{2L+1}}}\right) &{P}^*\left( \frac{x^{2L}}{\sqrt{x^{2(2L+1)}+(1-x^2)^{2L+1}}}\right). 
        \end{bmatrix}}
    \end{equation}
can then be implemented using $(2L+1)D$ queries to $e^{-iX_j \arccos(x)}$ with probability at least $1-\delta$.
\end{proposition}
\begin{proof}
    Under the assumption that $\mathcal{C}$ is the $[2L+1,1,2L+1]$ bit flip code and our signal Hamiltonian is $H= \sum_j \arccos(x)X_j$ and $\mathcal{E}=I$ prior to decoding we only need to reason about the action of $e^{-iH}$. Further, because the entire input state is always pure, it suffices to use the state vector formalism rather than the density matrix formalism. The action of $e^{-iH}$ on the codewords of the bit flip code are
    \begin{align}
        e^{-iH} \ket{0}^{\otimes 2L+1} &= \sum_{\mathbf{b} \in \{0,1\}^{2L+1}} (-i)^{|\mathbf{b}|}\, x^{2L+1-|\mathbf{b}|}(\sqrt{1-x^2})^{|\mathbf{b}|}\ket{\mathbf{b}} \nonumber\\
        e^{-iH} \ket{1}^{\otimes 2L+1} &= \sum_{\mathbf{b} \in \{0,1\}^{2L+1}} (-i)^{2L+1-|\mathbf{b}|}\, x^{|\mathbf{b}|}(\sqrt{1-x^2})^{2L+1-|\mathbf{b}|}\ket{\mathbf{b}}
    \end{align}
    Here $|\mathbf{b}|$ denotes the Hamming weight of the bit string $\mathbf{b}$.
    
    As $\mathcal{D}$ is the decoder for the bit flip code, let us assume for the moment that no error is observed. In this case we have that $\mathbf{b}=0\cdots 0$ or $\mathbf{b}=1\cdots 1$. The probability of this occurring is the same for both codewords and if we define $\Pi_L$ to be the logical subspace corresponding to these codewords we have that the probability of projecting into this codespace is
    \begin{equation}
        \bra{0}^{\otimes 2L+1} e^{iH} \Pi_L e^{-iH} \ket{0}^{\otimes 2L+1} = x^{2(2L+1)} + (1-x^2)^{2L+1}=\bra{1}^{\otimes 2L+1} e^{iH} \Pi_L e^{-iH} \ket{1}^{\otimes 2L+1}
    \end{equation}
    Thus we can guarantee that the probability of failing to project is at most $\delta$ if
    \begin{equation}
         |x|\in [0,\sqrt{1-(1-\delta)^{1/(2L+1)}}]\cup [(1-\delta)^{1/(4L+2)},1]  
    \end{equation}
    In the event that the projection is successful then either $I$ or $X_L$ is applied to the state which leads to
    \begin{align}
        \ket{0}_L &\rightarrow \frac{x^{2L+1} \ket{0}_L + (-1)^L(-i) ({1-x^2})^{(2L+1)/2} \ket{1}_L}{\sqrt{ x^{2(2L+1)} + (1-x^2)^{2L+1}}}\nonumber\\
        \ket{1}_L &\rightarrow \frac{x^{2L+1} \ket{1}_L + (-1)^L(-i) ({1-x^2})^{(2L+1)/2} \ket{0}_L}{\sqrt{ x^{2(2L+1)} + (1-x^2)^{2L+1}}}
    \end{align}
    This transformation is logically equivalent to $R_L(x)$ where
    \begin{equation}
        R_L(x) = e^{-i(-1)^LX_L \arctan\left( \frac{x^{2L+1}}{(1-x^2)^{(2L+1)/2}}\right)}.
    \end{equation}
    Thus if no errors are detected in the entire QSP protocol then a logical rotation $R_L$ is enacted.

    Let us assume that we wish to carry out at most $D$ such rotations then the probability of succeeding can be bounded above by $\delta$ from the union bound if the probability of each failure is at most $\delta/D$. This is equivalent to assuming that
    \begin{equation}
        |x|\in [0,\sqrt{1-(1-\delta/D)^{1/(2L+1)}}]\cup [(1-\delta/D)^{1/(4L+2)},1] 
    \end{equation}

    Given that $P,Q$ satisfy the standard assumptions for QSP in the Z-X basis (stated in Proposition \ref{prop:ECQSP}) we then have that we can implement a QSP sequence of the form
    \begin{equation}
        e^{-iZ_L\theta_D} R_L(x) e^{-iZ_L \theta_{D-1}} R_L(x) \cdots R_L(x) e^{-iZ_L \theta_0}
    \end{equation}
    to enact a transformation of the form
    \begin{equation}
    {\tiny
        \begin{bmatrix}
            {P}\left( \frac{x^{2L}}{\sqrt{x^{2(2L+1)}+(1-x^2)^{2L+1}}}\right) & -i\sqrt{1-\frac{x^{4L}}{x^{2(2L+1)}+(1-x^2)^{2L+1}} } Q\left(\frac{x^{2L}}{\sqrt{x^{2(2L+1)}+(1-x^2)^{2L+1}}} \right)\\
            -i\sqrt{1-\frac{x^{4L}}{x^{2(2L+1)}+(1-x^2)^{2L+1}} } Q^*\left(\frac{x^{2L}}{\sqrt{x^{2(2L+1)}+(1-x^2)^{2L+1}}}\right) &{P}^*\left( \frac{x^{2L}}{\sqrt{x^{2(2L+1)}+(1-x^2)^{2L+1}}}\right). 
        \end{bmatrix}}
    \end{equation}
    As this sequence requires $D$ queries to $R_L$ the probability of failure is at most $\delta$. Each query to $R_L$ can be implemented using $2L+1$ queries to $e^{-i X_j \arccos(x)}$ and so the total number of queries needed is at most $(2L+1)D$.
\end{proof}
The transformation in the above proposition is of a form that is not directly attainable using conventional quantum signal processing methods. This can be seen readily from $P$ because the polynomial attained for any $D>0$ is not a finite degree polynomial of $x$. This means that the above transformation can at best be approximated using conventional quantum signal processing, but not precisely attained. Failure probability is present in the above result however, but this can be mitigated as discussed below.
\begin{corollary}
    Under the assumptions of Proposition~\ref{prop:arctan}, encoded quantum signal processing can be used to implement a transformation $\tilde{W}_L$  such that $\|W_L - \tilde{W}_L\|\le \epsilon$ using a number of queries to $e^{-i X_j \arccos(x)}$ that are in
    $$
    \widetilde{{O}}\left(\frac{LD \log(1/\epsilon)}{\sqrt{\delta}}\right)
    $$
    where $\delta$ is the failure probability of the code-space projection.
\end{corollary}
\begin{proof}
    Let $R_g$ be a reflection operator such that $R_g\ket{p}=-\ket{p}$ if $\ket{p}\in \mathcal{H}_L$ for $\mathcal{C}$. Then if we construct a unitary that reflects about the ancillary registers used for the measurement $R_0$, we can construct $S=-R_0R_g$ as a walk operator. The operator $S$ by the results of Proposition~\ref{prop:arctan} requires $(2L+1)D$ queries to $e^{-iX_j \arccos(x)}$ to implement. Upon successful projection, we have from Proposition~\ref{prop:arctan} that a unitary is applied. Thus if the success probability is known a priori then oblivious amplitude amplification can be applied on $S$ to achieve $W_L'$ using $O(1/\sqrt{\delta})$ applications~\cite{berry2014exponential}.

    In the event that the success probability is not known, the quantum singular value transformation can be used to construct a unitary $\tilde{W}_L$ such that $\|\tilde{W}_L - W_L\|\le \epsilon$ by implementing a polynomial of degree $\widetilde{O}(\log(1/\epsilon)/\sqrt{\delta})$~\cite{yoder2014fixed,martyn2021grand,gilyen2019quantum}. This implies from the quantum singular value transformation that $\widetilde{O}(\log(1/\epsilon)/\sqrt{\delta})$ queries to $S$ are needed in this process. The claim then follows by multiplying this number of queries to $e^{-i X_j \arccos(x)}$ per query to $S$ which is $(2L+1)D$.
\end{proof}
This corollary shows that, in principle, EQSP can implement families of transformations that are not directly implementable using conventional quantum signal processing and further the success probability can be boosted to $1$ using oblivious amplitude amplification (at the price of small but controllable error). The resulting function can be seen to be a sigmoid, as illustrated in Fig.~\ref{fig:arctan}. For the case where $D=1$, we see that such a sigmoid can be constructed using  $2L+1= 11$ queries, which (after post-selection) is substantially better than the $19$ queries seen for the threshold function in~\cite{martyn2021grand} that is constructed using traditional QSP (but does not require any post-selection).

The fact that EQSP can easily implement sigmoid functions is, in its essence, the core property that we use in the following discussion where we show how to use these properties to solve metrological problems including Heisenberg limited sensing in the presence of stray $X$ fields and noisy phase estimation problems. 

\begin{figure}[t]
    \centering
    \includegraphics[width=0.45\linewidth]{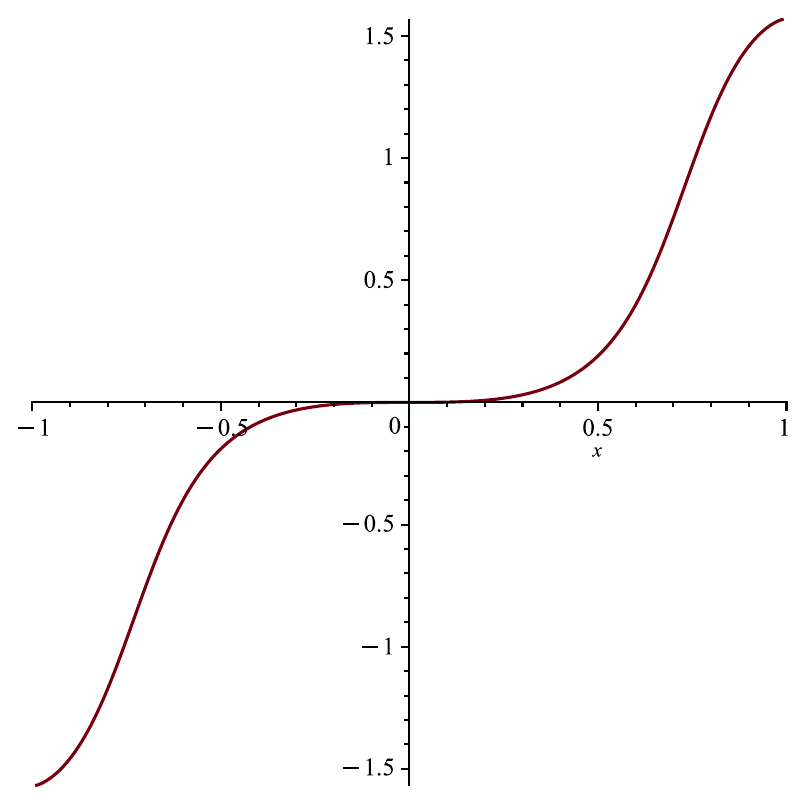}
    \includegraphics[width=0.45\linewidth]{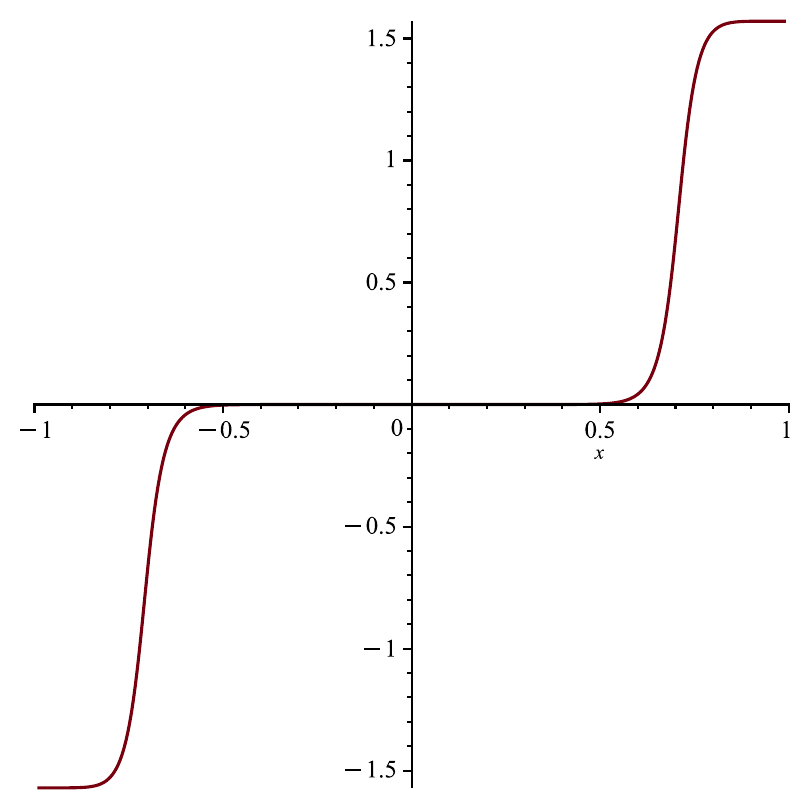}
    \caption{Plot of $\arctan\left( \frac{x^{2L+1}}{(1-x^2)^{(2L+1)/2}}\right)$ for $L=1$ ($N=3$, left) and $L=2$ ($N=5$, right). The resultant phase angle approaches a sigmoid function very rapidly and exactly upon success.}
    \label{fig:arctan}
\end{figure}

\section{Metrology as a Single-Qubit EQSP Problem}
We now apply the Encoded QSP framework of Section~\ref{sec:problem} to quantum metrology. The signal corresponds to a magnetic field applied to each qubit, and the error-detecting code provides robustness to transverse noise while preserving sensitivity to the signal.

\subsection{Problem Statement}
We relate EQSP to metrology by considering the following setting. We wish to combine an  entangled system of qubits in a $[2L+1,1,2L+1]$ repetition code analogous to the setting considered in Proposition~\ref{prop:arctan}. The signal in this setting can be thought of as a Hamiltonian that corresponds to a magnetic field applied to each of the $2L+1$ qubits. Specifically, we take $H= \sum_{k=0}^{2L} H_k$ where
\begin{equation}
H_k = \omega_k Z_k + \gamma_k X_k + \chi_k Y_k,
\end{equation}
where $Z_k$, $X_k$, and $Y_k$ are the Pauli operators on qubit $k$, $\omega_k = \omega + \epsilon_k$ can be interpreted as a longitudinal field strength comprising the signal $\omega$ and the inhomogeneity $\epsilon_k$, $\gamma_k$ is the transverse field component in the $X$ direction, and $\chi_k$ is the transverse field component in the $Y$ direction.

%We combine a system of $n=(2L+1)$ spin-1/2 particles (qubits) coupled to external magnetic fields, where $L$ is a positive integer controlling the code size. Each qubit $k$ evolves under a local Hamiltonian:

Our goal is to estimate the signal parameter $\omega$ using a number of queries to $e^{-i H_k}$ within error $\delta \omega$ that approaches the Heisenberg limit, meaning that the estimation error $\delta\omega$ scales as $1/N$ with $N$ being the total number of resources (qubits or measurements).  In this setting, one can interpret this as a system of entangled spin $1/2$ particles sent through a magnetic field and here the number of queries can be interpreted as the number of particles exposed to the field (as is common in entanglement based metrology schemes).

Traditional entanglement-based metrology schemes fail to learn $\omega$ at the Heisenberg limit because the remaining $X$ and $Y$ fields confound the inference.  Our aim here is to consider the use of error correction within the EQSP protocol to mitigate the impact of these terms.  Our proposal then corresponds to $\mathcal{E}=I$ and $\mathcal{D}$ the decoding channel for the bit flip code.

The unitary evolution of each qubit is:
\begin{equation}
U_k = e^{-i(\omega_k Z_k + \gamma_k X_k + \chi_k Y_k)t},
\end{equation}
where we set $t=1$ for simplicity. Using the fact that Pauli matrices anti-commute, we can expand:
\begin{align}
U_k = \cos(\Omega_k) I -i\sin(\Omega_k)\left(\frac{\omega_k Z_k + \gamma_k X_k + \chi_k Y_k}{\Omega_k} \right),
\end{align}
where $\Omega_k = \sqrt{\omega_k^2 + \gamma_k^2 + \chi_k^2}$ is the total field strength.
%\subsection{Decomposition into Signal and Error Components}

Notice that we can decompose the evolution into a desired signal component (pure $Z$-rotation) and unwanted error terms. By grouping the $Z$ and identity terms, we find that there exist $\beta_k \geq 0$ and $\phi_k \in [0,2\pi)$ such that:  
\begin{align}
    U_k &= \beta_k e^{-iZ_k\phi_k}- i\frac{\sin(\Omega_k)\sqrt{\gamma_k^2 + \chi_k^2}}{\Omega_k}\left(\frac{\gamma_k X_k + \chi_k Y_k}{\sqrt{\gamma_k^2 + \chi_k^2}} \right)\nonumber\\
    &= \beta_k e^{-iZ_k\phi_k} - i \sqrt{1- \beta_k^2}\left(\frac{\gamma_k X_k + \chi_k Y_k}{\sqrt{\gamma_k^2 + \chi_k^2}} \right)\nonumber\\
    &=\beta_k e^{-iZ_k\phi_k} - i \sqrt{1- \beta_k^2} e^{-i \Upsilon_k Z_k}X_k e^{i \Upsilon_k Z_k}\nonumber\\
    &=\beta_k e^{-iZ_k\phi_k} - i \sqrt{1- \beta_k^2} X_k e^{2i \Upsilon_k Z_k},\label{eq:ErrorExpansion}
\end{align}
where $\Upsilon_k = \frac{1}{2}\arctan(\chi_k/\gamma_k)$ characterizes the relative strength of $Y$ versus $X$ errors.

Using Euler's identity and elementary trigonometry, we obtain:
\begin{align}
    \phi_k&= \arctan\left(\frac{\tan(\Omega_k)\omega_k}{\Omega_k} \right),\\
    \beta_k&= \sqrt{\cos^2(\Omega_k) + \sin^2(\Omega_k)\left( \frac{\omega_k^2}{\Omega_k^2}\right)}.
\end{align}

In the small noise regime where $\max_k |\gamma_k/\omega_k| = o(1)$ and $\chi_k = 0$ (the general case with $\chi_k \neq 0$ is treated in Lemma~\ref{lem:mapping}), we can expand:
\begin{equation}
\Omega_k = \sqrt{\omega_k^2 + \gamma_k^2} = \omega_k\sqrt{1 + \gamma_k^2/\omega_k^2} = \omega_k\left(1 + \frac{\gamma_k^2}{2\omega_k^2} + O(\gamma_k^4/\omega_k^4)\right),
\end{equation}
which gives $\tan(\Omega_k) = \tan(\omega_k)(1 + O(\gamma_k^2/\omega_k^2))$ and $\tan(\Omega_k)\omega_k/\Omega_k = \tan(\omega_k) + O(\gamma_k^2/\omega_k)$. Therefore:
\begin{equation}
\phi_k = \arctan\left(\tan(\omega_k) + O(\gamma_k^2/\omega_k)\right) = \omega_k + O(\gamma_k^2/\omega_k). \label{eq:phi-taylor}
\end{equation}
Similarly, $\beta_k = 1 - O(\gamma_k^2/\omega_k^2)$ in this regime.

Equation~\eqref{eq:ErrorExpansion} reveals the key structure: the evolution decomposes into a desired $Z$-rotation (signal) and an error term involving $X_k$. This motivates our quantum error detection approach. Notice that the $\beta_k e^{-iZ_k\phi_k}$ term contains the signal information about $\omega$ and $X_k e^{2i\Upsilon_k Z_k}$ term represents an effective error. Our goal is to apply quantum error detection techniques to identify these errors while preserving the signal.

A key advantage of our approach is that by using quantum error detection codes, we can suppress undetected errors exponentially as $O(p^{L+1})$ where $p$ is the single-qubit error probability and $L$ is the code parameter \cite{nielsen2002quantum,gottesman1997stabilizer}; we prove this explicitly for the repetition code family in Proposition~\ref{prop:error-suppression} (Section~\ref{sec:fixed-omega}), where the error suppression enables Heisenberg-limited scaling for parameter estimation.

\subsection{Syndrome-Dependent Signal Transformations}

The decomposition in Equation~\eqref{eq:ErrorExpansion} reveals the core mechanism of Encoded QSP: the physical signal unitary, when applied to qubits protected by a repetition code, induces a syndrome-dependent logical rotation. The syndrome measurement itself acts as the signal-processing step, projecting the multi-qubit evolution onto an effective single-logical-qubit transformation. We now formalize this connection.

Consider a signal unitary in the XY-plane:
\begin{equation}
U_s(\phi, \vartheta) = \cos(\phi) I - i\sin(\phi) R(\vartheta), \quad \text{where } R(\vartheta) = \cos(\vartheta) X + \sin(\vartheta) Y.
\label{eq:signal-rotation}
\end{equation}
This corresponds to our physical evolution~\eqref{eq:ErrorExpansion} via the following mapping.

\begin{lemma}[Hamiltonian to Signal Rotation Mapping]
\label{lem:mapping}
The single-qubit evolution $U_k = e^{-i(\omega_k Z_k + \gamma_k X_k + \chi_k Y_k)}$ decomposes as in Equation~\eqref{eq:ErrorExpansion}:
\begin{equation}
U_k = \beta_k e^{-iZ_k\phi_k} - i\sqrt{1-\beta_k^2}\, R(2\Upsilon_k),
\end{equation}
where $R(\vartheta) = \cos(\vartheta)X + \sin(\vartheta)Y$. Comparing the error structure (i.e., the $\beta_k$ and $\sqrt{1-\beta_k^2}$ amplitudes) with the signal rotation form $U_s(\phi, \vartheta) = \cos(\phi)I - i\sin(\phi)R(\vartheta)$ from Equation~\eqref{eq:signal-rotation}, we identify:
\begin{align}
\cos(\phi) &= \beta_k, \quad \sin(\phi) = \sqrt{1-\beta_k^2}, \label{eq:map-amplitudes}\\
\vartheta &= 2\Upsilon_k = \arctan\left(\frac{\chi_k}{\gamma_k}\right).
\end{align}
The Z-rotation $e^{-iZ_k\phi_k}$ with $\phi_k = \arctan(\tan(\Omega_k)\omega_k/\Omega_k)$ is the signal phase we estimate; it commutes with the repetition code stabilizers and contributes directly to the logical phase.

In the small-noise regime where $(\gamma_k^2 + \chi_k^2)/\omega_k^2 = o(1)$:
\begin{align}
\beta_k &= 1 - \frac{\sin^2(\omega_k)(\gamma_k^2 + \chi_k^2)}{2\omega_k^2} + O\left(\frac{(\gamma_k^2 + \chi_k^2)^2}{\omega_k^4}\right),\\
\phi_k &= \omega_k + O\left(\frac{\gamma_k^2 + \chi_k^2}{\omega_k}\right),\\
\sqrt{1-\beta_k^2} &= \frac{|\sin(\omega_k)|\sqrt{\gamma_k^2 + \chi_k^2}}{\omega_k} + O\left(\frac{(\gamma_k^2 + \chi_k^2)^{3/2}}{\omega_k^3}\right).
\end{align}
\end{lemma}

\begin{proof}
The decomposition~\eqref{eq:ErrorExpansion} expresses $U_k$ as a signal term plus an error term. Using the identity $e^{-i\Upsilon Z}X e^{i\Upsilon Z} = X e^{2i\Upsilon Z}$ and the fact that $X e^{2i\theta Z} = \cos(2\theta)X + \sin(2\theta)Y$, the error operator becomes $R(2\Upsilon_k)$. The amplitude $\sqrt{1-\beta_k^2}$ plays the role of $\sin(\phi)$ in the signal rotation form~\eqref{eq:signal-rotation}, while $\beta_k$ plays the role of $\cos(\phi)$. The error bounds follow from Taylor expansion of $\Omega_k = \omega_k\sqrt{1 + (\gamma_k^2+\chi_k^2)/\omega_k^2}$ as detailed in Equation~\eqref{eq:phi-taylor}.
\end{proof}

With this correspondence established, we now characterize the logical operation induced by syndrome measurement.

\begin{theorem}[Syndrome-Dependent Logical Rotation]
\label{thm:syndrome-rotation}
Consider a repetition code of length $N = 2L+1$ with signal operator $U_s(\phi, \vartheta)^{\otimes N}$ applied to the physical qubits. For a syndrome indicating $j$ detected bit flips (where $0 \leq j \leq L$), the effective logical operation is:
\begin{equation}
U_j \propto \exp\left[-i\Theta_j \left(\cos(\vartheta_{\mathrm{eff}}) X_L + \sin(\vartheta_{\mathrm{eff}}) Y_L\right)\right],
\end{equation}
where the effective rotation angle and phase are:
\begin{align}
\Theta_j &= \arctan\left(\tan^{N-2j}(\phi)\right), \label{eq:theta-j}\\
\vartheta_{\mathrm{eff}} &= (N - 2j)\vartheta + (L - j)\pi. \label{eq:vartheta-eff}
\end{align}
The amplitude components satisfy:
\begin{align}
\cos\Theta_j &= \frac{\cos^{N-2j}(\phi)}{\sqrt{\cos^{2(N-2j)}(\phi) + \sin^{2(N-2j)}(\phi)}},\\
\sin\Theta_j &= \frac{\sin^{N-2j}(\phi)}{\sqrt{\cos^{2(N-2j)}(\phi) + \sin^{2(N-2j)}(\phi)}}.
\end{align}
\end{theorem}

\begin{proof}
The proof follows by expanding the tensor product $U_s(\phi, \vartheta)^{\otimes N}$ and grouping terms by the number of bit flips. For each syndrome with $j$ detected errors, two types of terms contribute: (i) terms with exactly $j$ physical rotations $R(\vartheta)$, which are corrected to the logical codespace, and (ii) terms with $N-j$ rotations, which cause a logical flip that the decoder misinterprets.

After applying the correction, the $|0\rangle_L$ component carries amplitude $\cos^{N-j}(\phi)\sin^j(\phi)$ and the $|1\rangle_L$ component carries $\cos^j(\phi)\sin^{N-j}(\phi)$, so their ratio is $\tan^{N-2j}(\phi)$. Normalizing yields the logical rotation $\cos\Theta_j I_L - i\sin\Theta_j R_L(\vartheta_{\mathrm{eff}})$ with $\Theta_j = \arctan(\tan^{N-2j}(\phi))$, which is Equation~\eqref{eq:theta-j}. The full derivation including phase tracking of $\vartheta_{\mathrm{eff}}$ follows from the subset decomposition framework of Theorem~\ref{thm:post-measurement-state-full} (Appendix~\ref{app:binary-search-proofs}), applied with the identification $\omega_k \to \phi_k$ and the error-corrected logical subspace replacing the GHZ subspace. The effective rotation axis $\vartheta_{\mathrm{eff}} = (N-2j)\vartheta + (L-j)\pi$ arises from tracking the $R(\vartheta)$ phases through the correction operation, where the $\pi$ shifts encode the parity of the decoder's correction.
\end{proof}

% \begin{definition}[Encoded Quantum Signal Processing]
% \label{def:encoded-qsp}
% An \emph{Encoded QSP protocol} consists of: (i) $N$ physical qubits encoding a single logical qubit via an error-detecting code $\mathcal{C}$, (ii) a signal unitary $U_s$ applied to the physical qubits, (iii) syndrome measurement that projects onto a syndrome sector $s$, inducing a syndrome-dependent logical transformation $\mathcal{E}_s$, and (iv) classical post-processing of the logical measurement outcome conditioned on $s$. The syndrome measurement replaces the explicit processing rotations of standard QSP~\cite{low2017optimal,gilyen2019quantum,motlagh2023generalized}, using the code structure itself as the signal-processing primitive.
% \end{definition}

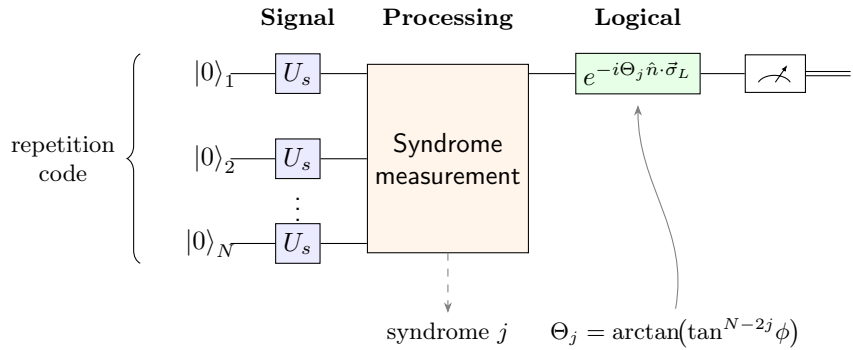
\begin{figure}[t]
\centering
\begin{tikzpicture}
  \begin{yquant}[operator/separation=6mm, register/separation=6mm]
    qubit {} q[3];
    [fill=blue!8, name=sig]
    box {$U_s$} q;
    align q;
    [fill=orange!8, name=syndbox, /yquant/operator/minimum width=22mm]
    box {Syndrome\\measurement} (q);
    discard q[1], q[2];
    [fill=green!10, name=logrot]
    box {$e^{-i\Theta_j \hat{n}\cdot\vec{\sigma}_L}$} q[0];
    measure q[0];
  \end{yquant}
  % Custom qubit labels
  \node[left=4mm of sig-0, anchor=east] {$\ket{0}_1$};
  \node[left=4mm of sig-1, anchor=east] {$\ket{0}_2$};
  \node[left=4mm of sig-2, anchor=east] {$\ket{0}_N$};
  % Vdots between q[1] and q[2]
  \node at ($(sig-1.south)!0.5!(sig-2.north)$) {$\vdots$};
  % Brace on left for ``repetition code''
  \draw[decorate, decoration={brace, amplitude=5pt, mirror}]
    ([xshift=-18mm]sig-0.north west) --
    ([xshift=-18mm]sig-2.south west)
    node[midway, left=6pt, font=\footnotesize, align=center] {repetition\\code};
  % Section labels
  \coordinate (labelY) at ([yshift=6mm]syndbox.north);
  \node[font=\footnotesize\bfseries] at (sig-0 |- labelY) {\strut Signal};
  \node[font=\footnotesize\bfseries] at (syndbox |- labelY) {\strut Processing};
  \node[font=\footnotesize\bfseries] at (logrot |- labelY) {\strut Logical};
  % Syndrome outcome
  \node[below=8mm of syndbox.south, font=\footnotesize] (syndlabel) {syndrome $j$};
  \draw[-{Stealth[length=4pt]}, dashed, gray] (syndbox.south) -- (syndlabel.north);
  % Theta annotation
  \node[right=3mm of syndlabel, font=\footnotesize, anchor=west] (thetalabel)
    {$\Theta_j = \arctan\!\bigl(\tan^{N-2j}\!\phi\bigr)$};
  % Conditioning arrow
  \draw[-{Stealth[length=4pt]}, gray]
    (thetalabel.north) to[out=80, in=-90] ([yshift=-1.5mm]logrot.south);
\end{tikzpicture}
\caption{Schematic of Encoded QSP with the repetition code. $N$ physical qubits, each subject to the signal unitary $U_s$, are encoded in a repetition code. Syndrome measurement extracts the error count $j$ and projects the multi-qubit state onto an effective single-logical-qubit rotation $\Theta_j = \arctan(\tan^{N-2j}(\phi))$ (Theorem~\ref{thm:syndrome-rotation}). The syndrome outcome conditions the classical post-processing, replacing the explicit processing rotations of standard QSP.}
\end{figure}

Theorem~\ref{thm:syndrome-rotation} provides the first concrete instantiation of Encoded QSP using the repetition code: the syndrome-dependent logical transformation takes the specific form $\Theta_j = \arctan(\tan^{N-2j}(\phi))$, a nonlinear function of the signal parameter $\phi$ parameterized by the number of detected errors $j$. The error-detecting code reduces the multi-qubit sensing problem to a family of single-logical-qubit rotations parameterized by the syndrome outcome. For the special case $\vartheta = 0$ (pure X-rotation) and $j = 0$ (no detected errors), we recover the phase amplification function $\Theta_0 = \arctan((-1)^L \tan^N(\phi))$, which drives the Heisenberg-limited binary search protocols of Section~\ref{sec:binary-search}.

The relationship between Encoded QSP and standard QSP merits clarification. In standard QSP and GQSP~\cite{motlagh2023generalized}, alternating signal and processing unitaries implement \emph{polynomial} transformations of the signal parameter (see~\cite{martyn2021grand} for a unifying perspective). Encoded QSP instead implements \emph{non-polynomial} transformations via projective syndrome measurement: for the repetition code, the function $\arctan(\tan^{N-2j}(\phi))$ is not a polynomial in $\phi$. More generally, the class of achievable transformations depends on the choice of code $\mathcal{C}$, the signal structure, and the syndrome processing strategy, opening a rich design space complementary to the polynomial universality of GQSP. Recent work on QSP-based phase estimation~\cite{dong2024optimal} and quantum computational sensing~\cite{khan2025quantum,allen2025quantum} further demonstrates the breadth of this connection. The block-encoding perspective on quantum error correction~\cite{lin2026quantum} provides additional context: the code block-encodes the signal, and syndrome measurement selects a specific block of the logical Hilbert space.

\section{Phase Estimation within Encoded QSP}

Having established the Encoded QSP framework, we now develop the phase estimation machinery needed to extract signal information from the logical qubit. The key idea is to combine information-theoretic phase estimation with syndrome-based error detection, using the syndrome outcome to condition the likelihood function rather than to apply recovery operations.

\SetKwInput{KwData}{Input}
\SetKwInput{KwResult}{Output}
\begin{algorithm}[t]
\KwData{Number of iterations $s\ge 0$, number of discrete phase angles $\tau$, unitary operation $U\in\mathbb{C}^{2^n\times 2^n}$, quantum state $\ket{\psi}\in \mathbb{C}^{2^n}$ }

\KwResult{The maximum likelihood estimate of ${\phi}$ given the experimental data generated}
\For{$i\gets0$ \KwTo $s-1$}{
    $M_i \gets$ sample drawn uniformly from $\{1,\ldots,\tau-1\}$\\
    $\theta_i \gets$ sample drawn uniformly from $[0,2\pi)$\\
    $\ket{\psi} \mapsto (He^{-i Z \theta_i} \otimes I)(\ketbra{0}{0} \otimes I + \ketbra{1}{1}\otimes U^{M_i})\ket{+} \ket{\psi}$\\
    $v_i \gets$ measurement outcome of first qubit in computational basis
    }
\For{$k\gets 0$ \KwTo $\tau-1$}{
    $P[k] \gets \prod_{i=0}^{s-1}\Pr(v_i|M_i,\theta_i,2\pi k/\tau) $\\
}
\Return $2\pi~{\rm argmax}(P)/\tau$
\caption{\label{alg:info}Information-Theoretic Phase Estimation}
\end{algorithm}

The decomposition in Equation~\eqref{eq:ErrorExpansion} shows that errors manifest as both amplitude damping (through $\beta_k$) and stochastic bit-flips. Standard quantum phase estimation, which relies on coherent evolution and interference, fails in the noisy regime we are considering. We need an approach that can extract phase information from noisy, partially-decohered states while maintaining some metrological advantage.

The information-theoretic phase estimation algorithm by Svore, Hastings, and Freedman~\cite{svore2013faster}, building on iterative phase estimation~\cite{kitaev1995quantum} and Bayesian approaches to Hamiltonian learning~\cite{granade2012robust}, provides an adaptive computational framework to extract phase information in the presence of noise (see Algorithm~\ref{alg:info}). Unlike standard approaches, this algorithm uses maximum likelihood estimation on classical measurement data. The key insight is that even when quantum coherence is partially lost, the measurement statistics still contain sufficient information to estimate the phase with Heisenberg-limited precision.

In the context of quantum sensing, $U$ represents the evolution of our probe state under the unknown magnetic field for unit time, and $\ket{\psi}$ is our carefully prepared probe state (e.g., a GHZ state or logical codeword). The algorithm estimates the accumulated phase $\lambda = \omega \tau$ by discretizing the phase space into $\tau$ bins and using maximum likelihood estimation.

More concretely, given an eigenstate $\ket{\lambda}$ such that $U\ket{\lambda} = e^{i\lambda} \ket{\lambda}$, we seek integer $k$ such that $2\pi k/\tau$ is closest to $\lambda = \omega \cdot t_{\text{sense}}$, where $t_{\text{sense}}$ is the sensing duration. The algorithm then performs maximum likelihood estimation: $P[k] = \prod_{i=0}^{s-1}\Pr(v_i|M_i,\theta_i,2\pi k/\tau)$.
The algorithm's robustness for noisy quantum sensing stems from three features: the parameters $M_i$ (evolution time) and $\theta_i$ (measurement basis rotation) are chosen randomly, preventing systematic errors from accumulating; classical maximum likelihood estimation replaces reliance on quantum interference; and the likelihood function naturally incorporates error probabilities from the noise model.

%\subsection{Resource Scaling}
Among the many approaches to phase estimation that achieve Heisenberg-limited scaling~\cite{higgins2007entanglement,wiebe2016efficient,svore2013faster}, Algorithm~\ref{alg:info} is particularly suited to error-detected sensing because its randomized evolution times naturally accommodate experiments where at least one error is detected. The performance of this method is given below for reference.
\begin{theorem}[Resource Scaling~\cite{svore2013faster}]
\label{thm:resource}
Let $U\ket{\psi} = e^{i\phi} \ket{\psi}$ for $U$ a unitary matrix acting on a $2^n$ dimensional Hilbert space.
The information-theoretic phase estimation algorithm (Algorithm~\ref{alg:info}) yields an estimator $\hat{\phi}$ of $\phi$ using 
$N_{\text{total}} = O\left(\frac{\log(1/\delta)\log(1/\epsilon)}{\epsilon}\right)$
total applications of the unitary $U$ to achieve estimation error $|\hat{\phi} - \phi|\le \epsilon$ with probability $1-\delta$.
\end{theorem}

This scaling is fundamental for quantum sensing applications. The $1/\epsilon$ dependence achieves the Heisenberg limit, while logarithmic factors represent only modest overhead as discussed in Table~\ref{tab:babbush}. The logarithmic factors arise from the need to distinguish between $t \sim 1/\epsilon$ possible phase values and ensure high confidence (controlled by $\delta$). Crucially, as we show in subsequent sections, this favorable scaling is preserved even when incorporating syndrome-based error detection: the measurement statistics conditioned on syndrome information retain sufficient information for Heisenberg-limited estimation despite partial decoherence.

\begin{table}[t]
\begin{tabular}{|c|c|}
\hline
     Method & Scaling \\
     \hline
     {Classical sensing (SQL)} \cite{caves1981quantum} & $N = \Theta(1/\epsilon^2)$ \\
     {Ideal quantum sensing (HL)} \cite{giovannetti2004quantum,giovannetti2006quantum} & $N = \Theta(1/\epsilon)$\\
     {Algorithm~\ref{alg:info}} (HL with log overhead)$^*$ \cite{svore2013faster}  & $N = \Theta(\log(1/\epsilon)/\epsilon)$\\
     \hline
\end{tabular}
\caption{Selected phase estimation methods. $^*$Suppressing dependence on the confidence parameter $\delta$; the full bound is $O(\log(1/\delta)\log(1/\epsilon)/\epsilon)$ per Theorem~\ref{thm:resource}.}
\label{tab:babbush}
\end{table}

\subsection{Integration with Syndrome-Based Error Detection}

The error detection protocols leverage the decomposition from Equation~\eqref{eq:ErrorExpansion}, where errors manifest as $X_k e^{2i\Upsilon_k Z_k}$ terms. Stabilizer measurements detect these errors without disturbing the signal component $\beta_k e^{-iZ_k\phi_k}$, providing syndrome patterns $\mathcal{S}$ (bit strings indicating which stabilizer generators detect errors) that indicate which errors occurred. Note that here if we consider the case of non-iterative phase estimation, wherein the state is not fed in recursively into subsequent rounds of error correction, only error detection is needed for phase estimation. For this reason we will primarily examine the case of error detection in the following discussion even though the results can be easily generalized to full error correction.

For each measurement iteration $i$ in Algorithm~\ref{alg:info}, we evolve under $U^{M_i}$ (accumulating phase and errors), then perform stabilizer measurements to extract syndrome $\mathcal{S}_i$ and finally perform a logical qubit measurement yielding outcome $v_i$. The syndrome modifies the likelihood calculation by conditioning on the detected error pattern:
\begin{equation}
P[k|\mathcal{S}] = \prod_{i=0}^{s-1}\Pr(v_i|M_i,\theta_i,2\pi k/t, \mathcal{S}_i).
\end{equation}
The conditional probability $\Pr(v_i|\cdots, \mathcal{S}_i)$ incorporates the syndrome-dependent phase shift:
\begin{equation}
\Pr(v_i|M_i,\theta_i,\phi, \mathcal{S}_i) = \frac{1}{2}\left(1 + \beta_{\mathcal{S}_i} \cos(M_i \phi_{\mathcal{S}_i} - \theta_i)\right),
\end{equation}
where $\phi_{\mathcal{S}_i}$ is the effective phase given syndrome $\mathcal{S}_i$ and $\beta_{\mathcal{S}_i}$ is the reduced visibility due to errors.

The specific form of $\phi_{\mathcal{S}}$ depends on the error detection code and estimation strategy employed. By Theorem~\ref{thm:syndrome-rotation}, when $j$ bit flips are detected, the effective rotation angle becomes:
\begin{equation}
\Theta_j = \arctan\left(\tan^{N-2j}(\phi)\right),
\end{equation}
where $\phi$ relates to the signal $\omega$ via Lemma~\ref{lem:mapping}. We develop the following complementary approaches based on this general framework:
\begin{enumerate}
\item \textbf{SQL barrier (Section~\ref{sec:sql-barrier})}: Shows that Z-basis syndrome measurement on the repetition code state, despite the per-syndrome phase amplification, yields only SQL scaling $F = \Theta(N)$ due to binomial fluctuations in the amplification factor $k = N - 2j$. This motivates the use of collective measurements in subsequent protocols.
\item \textbf{Bit-flip codes (Section~\ref{sec:fixed-omega})}: Specializes Theorem~\ref{thm:syndrome-rotation} to $\vartheta = 0$ (pure transverse noise in X direction). Syndrome measurement counts X-errors, yielding effective phase $\phi_{\mathcal{S}} = (2L+1-d)\omega$ when $d$ errors are detected.
\item \textbf{Adaptive binary search with GHZ states (Sections~\ref{sec:binary-search} and~\ref{sec:implementations})}: For varying signal $\omega_k = \omega + \epsilon_k$ with no transverse noise. Uses adaptive $R_Z$ phase corrections and majority voting to achieve logarithmic-depth estimation.
\item \textbf{Combined protocol with logical GHZ states (Section~\ref{sec:combined})}: For the general case with both transverse noise and field inhomogeneities. Uses repetition code blocks arranged in a logical GHZ configuration with adaptive binary search.
\end{enumerate}
This syndrome-enhanced likelihood estimation is the key innovation: rather than applying quantum recovery operations (which would be resource-intensive and could erase signal information), we use classical post-processing to account for detected errors.

\begin{remark}[Bidirectional QSP-Sensing Interpretation]
The connection between quantum sensing and quantum signal processing is bidirectional. From the QSP perspective, our syndrome-conditioned phase estimation demonstrates how error correction measurements can implement nonlinear phase transformations. From the sensing perspective, the QSP framework (Theorem~\ref{thm:syndrome-rotation}) provides the mathematical structure underlying our error-detection-enhanced protocols. This unified view suggests potential applications of sensing-motivated techniques to quantum algorithm design.
\end{remark}

%%%%%%%%%%%%%%%%%%%%%%%%%%%%%%%%%%%%%%%%%%%%%%%%%%%%%%%%%%%%%%%%%%%%%%%%%%%%%%%%
%%  SQL Barrier for Syndrome-Based Estimation
%%%%%%%%%%%%%%%%%%%%%%%%%%%%%%%%%%%%%%%%%%%%%%%%%%%%%%%%%%%%%%%%%%%%%%%%%%%%%%%%

\section{SQL Barrier: Limits of Product-State Encoded QSP}
\label{sec:sql-barrier}

The Encoded QSP framework of Section~\ref{sec:problem} reduces multi-qubit metrology to syndrome-dependent logical rotations, but the question remains: can the simplest instantiation, with product-state input and a single signal application, achieve Heisenberg scaling? The syndrome-dependent logical rotation of Theorem~\ref{thm:syndrome-rotation} provides an $N$-fold phase amplification $\Phi_{j,N}(\phi) = \arctan(\tan^{N-2j}\phi)$ that, for favorable syndrome outcomes ($|k| = |N - 2j| = \Theta(N)$), yields sensitivity scaling as $k^2$, suggestive of Heisenberg-limited estimation. A natural protocol would combine this amplification with adaptive centering at the steepest point $\phi = \pi/4$ and Z-basis syndrome measurement. We show that this approach cannot exceed SQL scaling: the binomial fluctuations in the syndrome outcome $k$ at the optimal operating point destroy the quadratic advantage. Overcoming this barrier requires either entanglement (Sections~\ref{sec:fixed-omega}--\ref{sec:implementations}) or sequential signal amplification (Section~\ref{subsec:sequential-binary-search}).

\begin{theorem}[SQL Barrier for Syndrome-Based Phase Estimation]
\label{thm:sql-barrier}
Consider the following protocol using $N = 2L+1$ qubits and the signal unitary $U_s(\phi, \vartheta)$ of Theorem~\ref{thm:syndrome-rotation}:
\begin{enumerate}
\item Prepare the computational basis state $|0\rangle_L = |0\rangle^{\otimes N}$.
\item Apply the signal unitary $U_s(\phi, \vartheta)^{\otimes N}$ with the QSP parameter centered at the optimal operating point $\phi = \pi/4$, where the bit-flip probability is $p = \sin^2(\phi) = 1/2$.
\item Perform syndrome measurement in the $Z$-basis, obtaining the syndrome pattern and the effective amplification factor $k = N - 2j$ where $j$ is the error weight of the syndrome coset representative.
\item Optimally measure the logical qubit to extract phase information from the syndrome-dependent rotation $\Phi_{j,N}$.
\end{enumerate}
At the optimal operating point $\phi = \pi/4$, the total quantum Fisher information summed over syndrome sectors satisfies
\begin{equation}
F_{\mathrm{total}} = 4N = \Theta(N).
\end{equation}
This equals the quantum Fisher information of the input state, so the syndrome measurement is optimal at this operating point. By the data-processing inequality, no classical post-processing of the syndrome outcome $k$ can surpass $\Theta(N)$ scaling.
\end{theorem}

\begin{proof}
By Theorem~\ref{thm:syndrome-rotation}, the effective logical operation in syndrome sector $j$ is $M_j = \cos(\Theta_j) I_L - i\sin(\Theta_j) R_L(\vartheta_{\mathrm{eff}})$, where $\Theta_j = \arctan(\tan^k \phi)$ with $k = N - 2j$, $\vartheta_{\mathrm{eff}} = k\vartheta + (L-j)\pi$, and $R_L(\vartheta_{\mathrm{eff}}) = \cos(\vartheta_{\mathrm{eff}}) X_L + \sin(\vartheta_{\mathrm{eff}}) Y_L$. Since $R_L(\vartheta_{\mathrm{eff}})|0\rangle_L = e^{i\vartheta_{\mathrm{eff}}}|1\rangle_L$, the post-syndrome logical state is
\begin{equation}
|\psi_j\rangle = \cos(\Theta_j)|0\rangle_L - i\sin(\Theta_j)\,e^{i\vartheta_{\mathrm{eff}}}|1\rangle_L,
\end{equation}
a pure qubit state whose polar angle on the Bloch sphere depends on $\phi$ through $\Theta_j$. The azimuthal phase $e^{i\vartheta_{\mathrm{eff}}}$ is independent of $\phi$ and does not affect the quantum Fisher information (QFI).

Differentiating the logical rotation angle with respect to $\phi$:
\begin{equation}
\frac{d\Theta_j}{d\phi} = \frac{k \tan^{k-1}(\phi)\sec^2(\phi)}{1+\tan^{2k}(\phi)}.
\end{equation}
At the optimal operating point $\phi = \pi/4$, where $\tan\phi = 1$ and $\sec^2\phi = 2$, this evaluates exactly to
\begin{equation}
\left.\frac{d\Theta_j}{d\phi}\right|_{\phi=\pi/4} = \frac{k \cdot 1 \cdot 2}{1 + 1} = k.
\end{equation}
The quantum Fisher information for estimating $\phi$ from the logical qubit in sector $j$ is $F_j = 4(d\Theta_j/d\phi)^2 = 4k^2$~\cite{braunstein1994statistical}. At $\phi = \pi/4$, we have $\Theta_j = \pi/4$ (since $\tan^k\phi = 1$), so $|\psi_j\rangle$ lies on the equator of the logical Bloch sphere, giving maximal sensitivity.

At $\phi = \pi/4$, the bit-flip probability is $p = \sin^2(\pi/4) = 1/2$, so $j \sim \operatorname{Binomial}(N, 1/2)$. The amplification factor $k = N - 2j$ has
\begin{equation}
\EX[k] = 0, \qquad \EX[k^2] = \operatorname{Var}(k) = 4\operatorname{Var}(j) = N,
\end{equation}
since $\operatorname{Var}(j) = N/4$ at $p = 1/2$. Since syndrome outcomes correspond to orthogonal subspaces, the total Fisher information is the weighted sum over syndrome sectors~\cite{helstrom1976quantum}:
\begin{equation}
F_{\mathrm{total}} = \sum_j P(j)\,F_j = 4\,\EX[k^2] = 4N.
\end{equation}

The mechanism is clear: while favorable syndrome outcomes ($|k| = \Theta(N)$) yield large per-syndrome Fisher information $F_j = 4k^2$, these outcomes have exponentially small probability $P(j) = O(2^{-N})$ for $|k| = \Theta(N)$. The binomial concentration of $k$ around zero, with $\operatorname{Var}(k) = N$, ensures that the weighted average $\EX[k^2]$ is exactly $N$, yielding SQL scaling.

This value $F_{\mathrm{total}} = 4N$ equals the quantum Fisher information of the input state $|0\rangle^{\otimes N}$ under the i.i.d.\ signal unitary (since the generator is $\sum_k R_k(\vartheta)$ with per-qubit variance $\operatorname{Var}(R(\vartheta))_{|0\rangle} = 1$ for all $\vartheta$, using $R(\vartheta)^2 = I$ and $\langle 0|R(\vartheta)|0\rangle = 0$), confirming that the syndrome measurement is optimal at this operating point. The data-processing inequality for Fisher information~\cite{zamir1998proof, ferrie2014data} guarantees that no classical post-processing of the syndrome can surpass the $\Theta(N)$ scaling.
\end{proof}

%\begin{remark}[Entangled Input: GHZ State]
\label{rem:ghz-sql}
For the special case $\vartheta = 0$ (pure X-rotation), a stronger statement holds when the input is the GHZ state $|+\rangle_L = (|0\rangle^{\otimes N} + |1\rangle^{\otimes N})/\sqrt{2}$: the logical qubit carries \emph{zero} Fisher information per syndrome sector. When $\vartheta = 0$, Equation~\eqref{eq:vartheta-eff} gives $\vartheta_{\mathrm{eff}} = (L-j)\pi$, so the effective logical channel reduces to $M_j = \cos(\Theta_j) I_L - i(-1)^{L-j}\sin(\Theta_j) X_L$. Since $|+\rangle_L$ is an eigenvector of $X_L$ with eigenvalue $+1$, the post-syndrome logical state is $M_j|+\rangle_L = e^{-i(-1)^{L-j}\Theta_j}|+\rangle_L$, a global phase. The interference between the two GHZ branches $|0\rangle^{\otimes N}$ and $|1\rangle^{\otimes N}$ produces four amplitude contributions per syndrome sector (two coset representatives $\times$ two logical branches), but these combine to leave the logical qubit invariant.

More generally, the GHZ state is the wrong entangled state for sensing under the i.i.d.\ generator $\sum_k R_k(\vartheta)$. The quantum Fisher information is $F_Q = 4\operatorname{Var}(\sum_k R_k(\vartheta))_{|+\rangle_L} = 4N$ for all $\vartheta$, since $\langle R_j(\vartheta) R_l(\vartheta)\rangle = 0$ for $j \neq l$ and $N \geq 3$. This is SQL, not Heisenberg. By contrast, the same state achieves $F_Q = 4N^2$ for Z-rotations (generator $\sum_k Z_k$), where the two branches accumulate opposite phases. The QSP mapping converts the Z-signal $\omega$ into transverse bit-flips parameterized by $\phi$, and this conversion is where the Heisenberg advantage is lost.
%\end{remark}

\begin{remark}[Connection to the Physical Signal $\omega$]
\label{rem:chain-rule}
Under the combined Hamiltonian $H_k = \omega Z_k + \gamma X_k$, the QSP parameter $\phi$ is related to $\omega$ via Lemma~\ref{lem:mapping}: the bit-flip amplitude satisfies $\sin(\phi) = \gamma|\sin(\Omega)|/\Omega$ with $\Omega = \sqrt{\omega^2 + \gamma^2}$. The centering condition $\phi = \pi/4$ requires $\gamma \geq |\omega|$, so the transverse field must be at least as large as the signal.

By the chain rule, $F_\omega = (d\phi/d\omega)^2 F_\phi$. Since $d\phi/d\omega$ is a nonzero $O(1)$ constant at the operating point (provided the non-degeneracy condition $\tan\Omega \neq \Omega$ holds), the SQL scaling $F_\omega = \Theta(N)$ follows from Theorem~\ref{thm:sql-barrier}.
\end{remark}

\begin{remark}[Measurement Choice Determines Scaling for $\omega$]
For the physical signal $\omega$, the measurement choice \emph{does} determine the scaling. The GHZ state carries $F_Q = 4N^2$ for $\omega$ via the Z-rotation generator $\sum_k Z_k$ (Theorem~\ref{thm:qfi-binary-search}), and X-parity measurement can access this full Heisenberg scaling (Sections~\ref{sec:fixed-omega}--\ref{sec:implementations}). By contrast, the syndrome-based protocol of Theorem~\ref{thm:sql-barrier} converts the Z-signal into X-bit-flips via the QSP mapping, yielding $F_\omega = \Theta(N)$ by the chain rule (Remark~\ref{rem:chain-rule}). Surpassing SQL requires both entanglement \emph{and} a measurement that preserves the coherent $N$-fold phase enhancement rather than decomposing it into per-syndrome sectors.
\end{remark}

%%%%%%%%%%%%%%%%%%%%%%%%%%%%%%%%%%%%%%%%%%%%%%%%%%%%%%%%%%%%%%%%%%%%%%%%%%%%%%%%
%%  Fixed Signal with Transverse Noise: Bit-Flip Codes
%%%%%%%%%%%%%%%%%%%%%%%%%%%%%%%%%%%%%%%%%%%%%%%%%%%%%%%%%%%%%%%%%%%%%%%%%%%%%%%%

\section{Fixed Signal with Transverse Noise: Bit-Flip Codes for Heisenberg Scaling}
\label{sec:fixed-omega}

This section develops the first Heisenberg-scaling instantiation of Encoded QSP, using bit-flip repetition codes to overcome the SQL barrier of Theorem~\ref{thm:sql-barrier}. By encoding the sensor qubits in a repetition code and using entangled (GHZ) logical states with X-parity measurement, we preserve the coherent $N$-fold phase enhancement while exponentially suppressing transverse errors.

\begin{algorithm}[t]
\KwData{Number of iterations $s\ge 0$, number of discrete phase angles $\tau$, number of repeated measurements, a set unitary operations in $\mathbb{C}^{2\times 2}$ of the form of~\eqref{eq:Ukdef} denoted $\{U_k:k=1,\ldots,2L+1\}$}

\KwResult{The maximum likelihood estimate of ${\phi}$ given the experimental data generated}
\For{$i\gets0$ \KwTo $s-1$}{
    $M_i \gets$ sample drawn uniformly from $\{1,\ldots,\tau-1\}$\\
    $\theta_i \gets$ sample drawn uniformly from $[0,2\pi)$\\
    Prepare $2L+1$ qubits in the logical state $\ket{\psi}_L\gets \ket{+}_L$ of the repetition code\\
    $\ket{\psi}_L \gets (\bigotimes_{k=1}^{2L+1}U_k^{M_i})\ket{\psi}_L$ \\
    Measure bit-flip code stabilizers to obtain syndrome $\mathcal{S}_i$ and error count $d_i$\\
    Apply logical rotation $e^{-iZ_L \theta_i}$ to $\ket{\psi}_L$.\\
    $v_i \gets $ measurement outcome from logical measurement of $\ket{\psi}_L$.\\
      }
Input $\{(M_i,\theta_i,v_i),i=0\ldots s-1\}$ into Algorithm~\ref{alg:info} to compute maximum likelihood estimate $\hat{\phi}$.\\
\Return $\hat{\phi}$
\caption{\label{alg:bitflip}Quantum Phase Estimation Using Bit Flip Codes}
\end{algorithm}
%\subsection{Setup and Error Model}
Theorem~\ref{thm:sql-barrier} showed that Z-syndrome measurement on the repetition code cannot extract more than $\Theta(N)$ Fisher information, regardless of classical post-processing. The protocols in this section and those that follow overcome this barrier by using \emph{collective} measurements (specifically X-parity readout and X-error syndrome detection) that preserve the coherent phase enhancement of the entangled state.

We first consider the case where $\chi_k = 0$ (no $Y$-field), $\omega$ is fixed but $\gamma_k$ can vary across the qubits. This scenario demonstrates how bit-flip codes enable Heisenberg-limited scaling by correcting transverse field errors while preserving the signal.
For this case, each qubit evolves under:
\begin{equation}
U_k = e^{-i(\omega Z_k + \gamma_k X_k)},\label{eq:Ukdef}
\end{equation}
where $\omega$ is the known signal we want to amplify and $\gamma_k$ represents unwanted transverse field noise.
Using the decomposition from Equation~\eqref{eq:ErrorExpansion} with $\chi_k = 0$ (no Y-field):
\begin{equation}
U_k = \beta_k e^{-iZ_k\phi_k} - i\sqrt{1-\beta_k^2} X_k,
\label{eq:bitflip-decomp}
\end{equation}
where:
\begin{align}
\phi_k &= \arctan\left(\frac{\tan(\sqrt{\omega^2 + \gamma_k^2})\omega}{\sqrt{\omega^2 + \gamma_k^2}}\right),\\
\beta_k &= \sqrt{\cos^2(\sqrt{\omega^2 + \gamma_k^2}) + \sin^2(\sqrt{\omega^2 + \gamma_k^2})\frac{\omega^2}{\omega^2 + \gamma_k^2}}.
\end{align}

%\subsection{The Bit-Flip Repetition Code}

Our use of error correction here is markedly different from the usual case. Rather than correcting all quantum errors, we wish to construct a code that is sensitive to some types of errors but not others. Specifically, we wish to build a code where physical $e^{-i\theta Z}$ operations correspond to logical $e^{-i \theta Z}$ operations. This means that any phase applied to each of the individual qubits will be applied to the whole logical qubit. This alone will allow the Heisenberg limit to be achieved here. On the other hand, we would like our code to be able to detect and remove $X$ operations. This will allow us to suppress the effect of an $X$-field.
The $[2L+1,1,2L+1]$ repetition code encodes one logical qubit:
\begin{equation}
|0\rangle_L = |0\rangle^{\otimes (2L+1)}, \quad |1\rangle_L = |1\rangle^{\otimes (2L+1)}.
\end{equation}
This code has the following stabilizer generators 
\begin{equation}
g_i = Z_i Z_{i+1}, \quad i = 1, 2, \ldots, 2L.
\end{equation}
which means that the code space can be thought of as being the $+1$ eigenspace of these stabilizer operators.
These $Z$-type stabilizers detect $X$-errors while preserving all $Z$-rotations. This is crucial: syndrome measurements identify transverse noise ($\gamma_k$ terms) without disturbing the signal ($\omega$ terms), allowing us to extract the correct phase through classical post-processing.

The algorithm that we propose for phase estimation using the bit flip code is given in Algorithm~\ref{alg:bitflip}. The following theorem shows how error detection discretizes the accumulated phase. Note that in some cases the correction can be avoided by simply tracing over the erroneous qubits, but we include it here to simplify the discussion.

%Starting with the logical state $|+\rangle_L = (|0\rangle_L + |1\rangle_L)/\sqrt{2}$, evolution under $U^{\otimes (2L+1)}$ yields 

\begin{theorem}[Bit-Flip Code Phase Amplification]
\label{thm:bitflip}
Consider a repetition code with $N = 2L+1$ qubits where $L>0$. If the initial state $\ket{+}_L = (\ket{0}_L + \ket{1}_L)/\sqrt{2}$ evolves under unitaries $\{U_k:k=1,\ldots,2L+1\}$ of the form~\eqref{eq:Ukdef} with fixed $\omega$ and heterogeneous transverse fields $\gamma_k$, and syndrome measurements detect $d$ errors, then the following hold:
\begin{enumerate}
\item The effective phase accumulated by the logical qubit satisfies
\begin{equation}
\phi_{\text{eff}} = \begin{cases} (2L+1-d)\phi & \text{if $\phi_k = \phi~\forall~k$}\nonumber\\
 (2L+1-d)\omega + O(\max_k \gamma_k^2/\omega)& \text{if $\max_k |\gamma_k/\omega| = o(1)$}
\end{cases}
\end{equation}
where $d = |s|$ is the number of detected errors.
\item The probability distribution of the number of errors satisfies
\begin{equation}
P(|s|=d) \le\binom{2L+1}{d}(\max_k \beta_k^2)^{2L+1} \left(\frac{(1-\min_k\beta_k^2)}{\max_k \beta_k^2}\right)^{d}.
\end{equation}
\item When $\max_k |\gamma_k/\omega| = o(1)$, the quantum Fisher information $F_Q$ for $\omega$ satisfies
\begin{equation}
F_Q = 4N^2 \left(1 - O\left(\max_k |\gamma_k/\omega|^2\right)\right),
\end{equation}
which via the quantum Cramér-Rao bound implies that any unbiased estimator of $\omega$ using $n$ independent measurements has variance $\operatorname{Var}(\hat{\omega}) = \Omega(1/(nN^2))$, achieving Heisenberg-limited scaling.
\end{enumerate}
\end{theorem}

\begin{proof}
We analyze the evolution of the logical state $|+\rangle_L$ under the noisy unitary and show how syndrome measurements enable Heisenberg-limited phase estimation.
Starting with $|+\rangle_L = (|0\rangle_L + |1\rangle_L)/\sqrt{2}$, each physical qubit evolves under $U_k = \beta_k e^{-iZ_k\phi_k} - i\sqrt{1-\beta_k^2} X_k$ from Equation~\eqref{eq:bitflip-decomp}. We begin with the case where $\phi_k$ is constant.

The logical basis states evolve as follows. For $|0\rangle_L = |0\rangle^{\otimes (2L+1)}$, since $Z_k|0\rangle = |0\rangle$:
\begin{align}
U_k|0\rangle &= (\beta_k e^{-iZ_k\phi_k} - i\sqrt{1-\beta_k^2} X_k)|0\rangle = \beta_k e^{-i\phi_k}|0\rangle - i\sqrt{1-\beta_k^2}|1\rangle.
\label{eq:single-qubit-0}
\end{align}
Similarly, for $|1\rangle_L = |1\rangle^{\otimes (2L+1)}$, since $Z_k|1\rangle = -|1\rangle$:
\begin{align}
U_k|1\rangle &= (\beta_k e^{-iZ_k\phi_k} - i\sqrt{1-\beta_k^2} X_k)|1\rangle = \beta_k e^{i\phi_k}|1\rangle - i\sqrt{1-\beta_k^2}|0\rangle.
\label{eq:single-qubit-1}
\end{align}

For the homogeneous case where $\gamma_k$ is constant, we have $U_k = U$, $\beta_k = \beta$, and $\phi_k = \phi$ for all $k$. The logical basis states then evolve as:
\begin{align}
    U^{\otimes 2L+1} \ket{0}_L = \sum_{x_1,\ldots,x_{2L+1}=0}^1 \beta^{2L+1-\sum x_i}e^{-i(2L+1 -\sum x_i)\phi}\left(-i\sqrt{1-\beta^2}\right)^{\sum x_i}\ket{x_1\cdots x_{2L+1}}
\label{eq:logical-0-expansion}
\end{align}
and
\begin{align}
    U^{\otimes 2L+1} \ket{1}_L = \sum_{x_1,\ldots,x_{2L+1}=0}^1 \beta^{2L+1-\sum x_i}e^{i(2L+1 -\sum x_i)\phi}\left(-i\sqrt{1-\beta^2}\right)^{\sum x_i}\ket{x_1\cdots x_{2L+1}}
\end{align}

The repetition code employs $2L$ distinct stabilizers of the form $g_j = Z_j Z_{j+1}$ for $j = 1,\ldots,2L$. Each stabilizer detects X-errors through anticommutation:
\begin{equation}
g_j X_k g_j = \begin{cases} -X_k & \text{if } k \in \{j, j+1\} \\ X_k & \text{otherwise} \end{cases}.
\end{equation}
An X-error on qubit $k$ anticommutes with stabilizers $g_{k-1}$ and $g_k$ (where defined), changing their eigenvalues from $+1$ to $-1$. Measuring all stabilizers yields a syndrome bit string ${s} \in \{0,1\}^{2L}$ that uniquely identifies the error pattern.

For error pattern $s = \{x_1,\ldots,x_{2L+1}\}$ where $x_k = 1$ indicates X-error on qubit $k$, let $\Pi_s$ denote the projector onto the subspace consistent with syndrome $s$. To analyze the effective phase, we introduce the correction operation $C_s$, the unitary that would flip the erroneous qubits back to their correct values. In practice, the protocol measures the syndrome to obtain the error count $d = |s|$ and uses this information in classical post-processing of the measurement outcomes, rather than applying the quantum correction $C_s$. However, analyzing the hypothetical corrected state $C_s \Pi_s (\bigotimes_{k=1}^{2L+1}U_k)\ket{+}_L$ reveals the effective phase $(2L+1-d)\phi$ that appears in the classical likelihood function.

To derive the syndrome amplitude, we analyze how $\Pi_s$ selects terms from the expansion in Equation~\eqref{eq:logical-0-expansion}. The syndrome measurement projects onto configurations where errors occur at exactly those locations indicated by $s$. For a syndrome with $|s|$ errors, $\Pi_s$ selects the unique term in the sum where $\sum_i x_i = |s|$. This term has amplitude:
\begin{align}
    \beta^{2L+1-|s|} \cdot \left(-i\sqrt{1-\beta^2}\right)^{|s|} \cdot e^{-i(2L+1-|s|)\phi},
\end{align}
where the first factor arises from the $(2L+1-|s|)$ qubits without errors (contributing $\beta$ each), the second from the $|s|$ qubits with X-errors (contributing $-i\sqrt{1-\beta^2}$ each), and the third from the phase accumulation on non-error qubits. Since $\left(-i\sqrt{1-\beta^2}\right)^{|s|} = (-i)^{|s|}(1-\beta^2)^{|s|/2}$, the amplitude becomes:
\begin{align}
    (-i)^{|s|} \beta^{2L+1-|s|}(1-\beta^2)^{|s|/2} e^{-i(2L+1-|s|)\phi}.
\end{align}
The state before correction is $|x_1\cdots x_{2L+1}\rangle$ with $|s|$ bits flipped from $|0\rangle_L$. Applying the correction operation $C_s$ flips these bits back, returning the state to $|0\rangle_L$. Absorbing the global phase $(-i)^{|s|}$, we obtain for the homogeneous case:
\begin{align}
    C_s \Pi_s U^{\otimes 2L+1} \ket{0}_L = \beta^{2L+1 - |s|}(1-\beta^2)^{|s|/2} e^{-i(2L+1-|s|)\phi} \ket{0}_L \label{eq:syndrome-amplitude-0}
\end{align}
and similarly
\begin{align}
    C_s \Pi_s U^{\otimes 2L+1} \ket{1}_L = \beta^{2L+1 - |s|}(1-\beta^2)^{|s|/2} e^{i(2L+1-|s|)\phi} \ket{1}_L \label{eq:syndrome-amplitude-1}
\end{align}
Thus for any logical qubit state $\ket{\psi}_L = a \ket{0}_L + b\ket{1}_L$ in the homogeneous case, if $|s|$ errors are observed then:
\begin{equation}
    \ket{\psi}_L \mapsto e^{-i(2L+1 - |s|)\phi Z_L}\ket{\psi}_L,
\end{equation}
corresponding to a logical $Z$-rotation through the reduced angle $(2L+1-|s|)\phi$.

Consider a scenario where $m$ errors occur on the first $m$ qubits (without loss of generality). Let $s$ denote the binary string for these error locations with Hamming weight $|s|=m$. Up to exponentially small errors:
\begin{equation}
P(s) = \prod_{k=1}^m (1-\beta_k^2) \prod_{k=m+1}^N \beta_k^2\le (\max_k \beta_k^2)^{N-m} (1-\min_k\beta_k^2)^{m}.
\end{equation}
There are $\binom{N}{m}$ ways to distribute these errors. Since error strings are independent:
\begin{align}
    P(d) = \sum_{s: |s|=d} P(s) &\le \binom{2L+1}{d}(\max_k \beta_k^2)^{2L+1-d} (1-\min_k\beta_k^2)^{d} \nonumber\\
    &=\binom{2L+1}{d}(\max_k \beta_k^2)^{2L+1} \left(\frac{(1-\min_k\beta_k^2)}{\max_k \beta_k^2}\right)^{d}.
\end{align}
This establishes item 2.

We now generalize to the case where $\phi_k$ varies with qubit index $k$ in the regime $\max_k |\gamma_k/\omega| = o(1)$. The key observation from Equations~\eqref{eq:single-qubit-0} and~\eqref{eq:single-qubit-1} is that when applying $\bigotimes_{k=1}^{2L+1} U_k$ to a logical basis state, each qubit $k$ contributes:
\begin{itemize}
\item \textbf{Without X-error}: Phase factor $\beta_k e^{\pm i\phi_k}$ (sign depends on $|0\rangle$ vs $|1\rangle$)
\item \textbf{With X-error}: Amplitude $-i\sqrt{1-\beta_k^2}$ with \emph{no phase dependence}
\end{itemize}
The phase-independence of error terms follows directly from Equations~\eqref{eq:single-qubit-0} and~\eqref{eq:single-qubit-1}, where the X-error contribution $-i\sqrt{1-\beta_k^2}X_k$ appears without any phase factor $e^{\pm i\phi_k}$.

When syndrome measurement detects $d$ errors at locations $\mathcal{E}$ (with $|\mathcal{E}| = d$), the effective phase accumulated by the logical qubit is determined solely by the $(2L+1-d)$ non-erroneous qubits:
\begin{align}
\phi_{\text{eff}} &= \sum_{k \notin \mathcal{E}} \phi_k = \sum_{k=1}^{2L+1} \phi_k - \sum_{j \in \mathcal{E}} \phi_j.
\end{align}
Using Equation~\eqref{eq:phi-taylor}, which states $\phi_k = \omega + O(\gamma_k^2/\omega)$ for each qubit, we obtain:
\begin{align}
\phi_{\text{eff}} &= \sum_{k=1}^{2L+1} \left[\omega + O(\gamma_k^2/\omega)\right] - \sum_{j \in \mathcal{E}} \left[\omega + O(\gamma_j^2/\omega)\right]\\
&= (2L+1)\omega - d\omega + O\left((2L+1)\max_k (\gamma_k^2/\omega)\right) - O\left(d\max_k (\gamma_k^2/\omega)\right)\\
&= (2L+1-d)\omega + O(\max_k \gamma_k^2/\omega),
\end{align}
where the final step uses the fact that the $O(\cdot)$ terms have the same leading order. Combined with the homogeneous case effective phase $(2L+1-|s|)\phi$ from Equations~\eqref{eq:syndrome-amplitude-0} and~\eqref{eq:syndrome-amplitude-1}, this establishes item 1 for both cases.

The quantum Fisher information (QFI) quantifies the maximum information about $\omega$ extractable from our quantum state over all possible measurements~\cite{braunstein1994statistical,helstrom1976quantum}. The quantum Cramér-Rao bound states that for any unbiased estimator $\hat{\omega}$:
\begin{equation}
\operatorname{Var}(\hat{\omega}) \geq \frac{1}{n \cdot F_Q},
\label{eq:cramer-rao}
\end{equation}
where $n$ is the number of independent measurements and $F_Q$ is the QFI.

For a pure state $|\psi(\omega)\rangle$, the QFI is given by:
\begin{equation}
F_Q = 4\left(\langle\partial_\omega\psi|\partial_\omega\psi\rangle - |\langle\psi|\partial_\omega\psi\rangle|^2\right) = 4\operatorname{Var}_\psi(H),
\end{equation}
where $H$ is the generator satisfying $|\partial_\omega\psi\rangle = -iH|\psi(\omega)\rangle$~\cite{braunstein1994statistical}.

After syndrome measurement we will receive a classical mixture of states with different error counts. Since different syndrome outcomes correspond to orthogonal subspaces, the quantum Fisher information for this classical mixture is
\begin{equation}
    F_Q = \sum_{d} P(d) F_Q^{(d)}.
\end{equation}
This shows that we can first turn our attention to the quantum Fisher information given that $d$ errors are observed and then average over the distribution of errors. From Equations~\eqref{eq:syndrome-amplitude-0} and~\eqref{eq:syndrome-amplitude-1}, after syndrome measurement detects $d$ errors, the normalized state on the projected subspace satisfies
\begin{equation}
|\psi_d(\omega)\rangle = \frac{1}{\sqrt{2}}\left(e^{-i(2L+1-d)\omega}|0\rangle_L + e^{i(2L+1-d)\omega}|1\rangle_L\right) + O\left((1-\min_k\beta_k^2)^{(L+1)/2}\right),
\end{equation}
where the amplitude error bound $O((1-\min_k\beta_k^2)^{(L+1)/2})$ is exponentially small in $L$ and follows from the amplitude corrections in~\eqref{eq:syndrome-amplitude-0} and~\eqref{eq:syndrome-amplitude-1}. In the heterogeneous case, there is also a phase correction $\phi_{\text{eff}} = (2L+1-d)\omega + O(\max_k \gamma_k^2/\omega)$ from item 1. The dominant $O(\max_k (\gamma_k/\omega)^2)$ correction to the QFI arises from averaging over the syndrome distribution in Equation~\eqref{eq:qfi-final}, with both amplitude and phase errors contributing at the same order.

The generator is $H = (2L+1-d)Z_L$, and since this state has maximal variance in $Z_L$:
\begin{equation}
F_Q^{(d)} = 4(2L+1-d)^2.
\end{equation}

Averaging over the syndrome distribution:
\begin{equation}
F_Q = \sum_{d} P(d) F_Q^{(d)} = 4\EX[(2L+1-d)^2].
\end{equation}
Expanding the expectation:
\begin{align}
\EX[(2L+1-d)^2] &= (2L+1)^2 - 2(2L+1)\EX[d] + \EX[d^2]\nonumber\\
&= (2L+1)^2 - 2(2L+1)\EX[d] + \operatorname{Var}[d] + \EX[d]^2\nonumber\\
&= (2L+1-\EX[d])^2 + \operatorname{Var}[d].
\end{align}

Define the average error probability $\hat{p} := (2L+1)^{-1}\sum_k p_k$ where $p_k = 1 - \beta_k^2$ is the per-qubit error probability from Equation~\eqref{eq:map-amplitudes}. In the heterogeneous case, $d$ follows a Poisson-binomial distribution with $\EX[d] = \sum_k p_k = (2L+1)\hat{p}$ and $\operatorname{Var}(d) = \sum_k p_k(1-p_k) \leq (2L+1)\hat{p}(1-\hat{p})$. Since $\hat{p} = O(\max_k (\gamma_k/\omega)^2)$ in the small noise regime, the variance contribution scales as:
\begin{equation}
\frac{\operatorname{Var}[d]}{(2L+1)^2} = \frac{\hat{p}(1-\hat{p})}{2L+1} = O\left(\frac{\max_k (\gamma_k/\omega)^2}{2L+1}\right),
\end{equation}
which is negligible compared to the leading linear correction $-2\hat{p} = O(\max_k (\gamma_k/\omega)^2)$ from the mean, with relative suppression $1/(2L+1)$. Therefore:
\begin{equation}
F_Q = 4(2L+1)^2\left(1 - 2\frac{\EX[d]}{2L+1} + \frac{\operatorname{Var}[d]}{(2L+1)^2}\right) = 4(2L+1)^2\left(1 - O(\max_k (\gamma_k/\omega)^2)\right).
\label{eq:qfi-final}
\end{equation}
By the quantum Cramér-Rao bound~\eqref{eq:cramer-rao}, this implies that any unbiased estimator using $n$ independent measurements satisfies
\begin{equation}
\operatorname{Var}(\hat{\omega}) \geq \frac{1}{n \cdot F_Q} = \frac{1}{4n(2L+1)^2(1 - O(\max_k (\gamma_k/\omega)^2))} = \Omega(1/(nN^2)),
\end{equation}
where $N = 2L+1$ is the total number of qubits. This achieves Heisenberg-limited scaling where precision improves quadratically with system size, which proves item 3.
\end{proof}

\begin{remark}[Connection to Syndrome-Dependent Rotations]
Theorem~\ref{thm:bitflip} analyzes the signal phase accumulation $(N-d)\phi_k$ from the Z-rotation component of the evolution. This complements Theorem~\ref{thm:syndrome-rotation}, which describes the effective XY-rotation angle $\Theta_j$ arising from the error structure. Both perspectives share the same underlying mechanism: syndrome counting on the repetition code determines how the signal or error angle is transformed. The key distinction is that Theorem~\ref{thm:syndrome-rotation} applies to pure XY-plane rotations (the QSP framework), while Theorem~\ref{thm:bitflip} includes the Z-rotation signal term relevant for quantum sensing.
\end{remark}

Theorem~\ref{thm:bitflip} establishes performance guarantees for arbitrary deterministic heterogeneous fields $\gamma_k$ using worst-case bounds $\max_k |\gamma_k/\omega|$. In practice, however, noise fields in experimental systems are not fixed but rather drawn from probability distributions due to fabrication imperfections, environmental fluctuations, and calibration uncertainties. The following proposition addresses this realistic scenario by analyzing how the statistical properties of noise heterogeneity affect the average performance, quantifying the degradation through a single heterogeneity parameter $h$ that characterizes the relative spread of noise strengths.

\begin{proposition}[Heterogeneous Noise Correction]
\label{prop:heterogeneous}
Consider the setup of Theorem~\ref{thm:bitflip} with heterogeneous transverse fields $\gamma_k \sim \mathcal{N}(\gamma, (\gamma h)^2)$ where $h \geq 0$ is the relative standard deviation. When $|\gamma/\omega| = o(1)$, the expected quantum Fisher information satisfies
\begin{equation}
\EX[F_Q] = 4(2L+1)^2\left(1 - O((\gamma/\omega)^2(1+h^2))\right),
\end{equation}
where the $(1+h^2)$ correction shows that heterogeneity increases the average error probability by a factor $(1+h^2)$ relative to homogeneous noise with the same mean $\gamma$.
\end{proposition}
\begin{proof}
See Appendix~\ref{app:heterogeneous-noise}.
\end{proof}

\begin{proposition}[X-Error Suppression Bound]
\label{prop:error-suppression}
In a $(2L+1)$-qubit repetition code detecting X-errors with single-qubit error probability $p_X < 1/2$, undetected errors (requiring $L+1$ or more X-errors) satisfy
\begin{equation}
P_{X,\text{undetected}} = O(p_X^{L+1}),
\end{equation}
achieving exponential suppression in the code parameter $L$.
\end{proposition}

\begin{proof}
Recall from Theorem~\ref{thm:bitflip} that the code uses $N = 2L+1$ physical qubits with stabilizer measurements detecting X-errors. An undetected error requires at least $L+1$ X-errors among the $2L+1$ qubits (exceeding the detection threshold). The probability of exactly $d \geq L+1$ errors follows a binomial distribution:
\begin{equation}
P(d \text{ errors}) = \binom{2L+1}{d} p_X^d (1-p_X)^{2L+1-d}.
\end{equation}
Summing over all undetectable patterns:
\begin{align}
P_{X,\text{undetected}} &= \sum_{d=L+1}^{2L+1} \binom{2L+1}{d} p_X^d (1-p_X)^{2L+1-d}\\
&\leq \sum_{d=L+1}^{2L+1} \binom{2L+1}{d} p_X^d = p_X^{L+1} \sum_{d=L+1}^{2L+1} \binom{2L+1}{d} p_X^{d-L-1}.
\end{align}
For $p_X < 1/2$, each factor $p_X^{d-L-1} \leq 1$ and each binomial coefficient satisfies $\binom{2L+1}{d} \leq 2^{2L+1}$, so:
\begin{equation}
P_{X,\text{undetected}} \leq (L+1) \cdot 2^{2L+1} \cdot p_X^{L+1} = O(p_X^{L+1}),
\end{equation}
where the $O(\cdot)$ absorbs the prefactor $(L+1) \cdot 2^{2L+1}$, which is subexponential in $L$ relative to $p_X^{L+1}$ when $p_X < 1/2$.
\end{proof}

The next step in the argument is to show that the steps in Algorithm~\ref{alg:bitflip} will yield the desired estimate of the eigenphase $\phi$ when used inside information theoretic phase estimation.

\begin{lemma}[Maximum Likelihood Sampling Complexity]
\label{lem:mle-sampling}
    Given target precision $\epsilon > 0$ and failure probability $\delta \in (0,1/2]$, the bit-flip code protocol integrated with information-theoretic phase estimation requires $\ell = \Theta(\log(1/(\epsilon\delta)))$ uniformly distributed measurement outcomes. Drawing $s = \Theta(\ell \log(1/\delta))$ raw samples via rejection sampling yields at least $\ell$ accepted samples with probability $1-O(\delta)$, achieving total evolution time:
    \begin{equation}
    T_{\text{total}} = \Theta\left(\frac{\ell \log(1/\delta)}{\epsilon}\right) = \Theta\left(\frac{\log(1/(\epsilon\delta))\log(1/\delta)}{\epsilon}\right).
    \end{equation}
\end{lemma}
\begin{proof}
We prove this by showing how rejection sampling produces uniformly distributed evolution times and deriving concentration bounds for the sample complexity.

The information-theoretic phase estimation algorithm (Algorithm~\ref{alg:info}) requires $\ell=\Theta(\log(1/(\epsilon\delta)))$ evolution times chosen uniformly from $\{1, 2, \ldots, t_{\max}\}$ where $t_{\max} = \Theta(1/\epsilon)$. However, the syndrome-dependent phase $(2L+1-d)\omega$ from Theorem~\ref{thm:bitflip} complicates direct sampling. The natural sampling process is:
\begin{enumerate}
\item Draw intended evolution time $t_i \sim \text{Uniform}(\{1, \ldots, t_{\max}\})$
\item Perform the experiment and observe $d_i$ errors (with $d_i \sim \text{Binomial}(2L+1, \hat{p})$)
\item The effective qubit count is $N_i = N - d_i$ (from the syndrome-dependent phase), so the effective phase is $N_i M_i \phi$ where $M_i$ is the time multiplier
\end{enumerate}

Without intervention, $T_i$ is \emph{not} uniformly distributed because the error count $d_i$ is random. To obtain uniformly distributed effective times $T_i$, we use rejection sampling.

Define the binary random variable $X_i \in \{0,1\}$ where $X_i = 1$ indicates sample $i$ is accepted and $X_i = 0$ indicates rejection. The acceptance rule is:
\begin{equation}
    \Pr(X_i = 1 \mid T_i, d_i) = \begin{cases}
        \frac{P_{\min}}{\Pr(d_i|t_i=T_i+d_i)(d_{\max}+1)} & \text{if $d_i\le d_{\max}$}\\
        0 & \text{otherwise}
    \end{cases}
\end{equation}
where $P_{\min}$ and $d_{\max}$ are parameters chosen to ensure high acceptance probability.

Let $\langle d\rangle = \EX[d]$ be the mean number of errors. Markov's inequality states that for any non-negative random variable $Y$ and $a > 0$:
\begin{equation}
\Pr(Y \ge a) \le \frac{\EX[Y]}{a}.
\end{equation}

Applying this to the number of errors $d$ with $a = \lambda\langle d\rangle$ for $\lambda > 1$:
\begin{equation}
\Pr(d \ge \lambda\langle d\rangle) \le \frac{\EX[d]}{\lambda\langle d\rangle} = \frac{1}{\lambda}.
\end{equation}

Therefore, setting
\begin{equation}
d_{\max} = \lambda\langle d\rangle \quad \text{with} \quad \lambda = \Theta(\log(1/\delta))
\end{equation}
ensures that $\Pr(d \le d_{\max}) \ge 1 - 1/\lambda = 1 - O(\delta)$.

For the binomial error distribution from Theorem~\ref{thm:bitflip} with $\EX[d] = (2L+1)\hat{p}$ where $\hat{p} = O((\gamma/\omega)^2)$, the most likely outcome (mode) is $d=0$ when $(2L+1)\hat{p} < 1$. In this small-noise regime, the minimum acceptance probability is:
\begin{equation}
P_{\min} = \Pr(d=0) = (1-\hat{p})^{2L+1} \ge \exp(-(2L+1)\hat{p}) = \Theta(1)
\label{eq:pmin-bound}
\end{equation}
for fixed noise strength $\gamma = O(1)$.

Since $t_i$ and $d_i$ are independent, and $t_i$ is uniform over $\{1, \ldots, t_{\max}\}$, we compute the probability of obtaining effective time $T_i = k$ and accepting the sample ($X_i=1$):
\begin{align}
    \Pr(T_i = k, X_i=1) &= \sum_{d=0}^{d_{\max}} \Pr(t_i = k+d, d_i = d, X_i=1 \mid d_i=d)\nonumber\\
    &= \sum_{d=0}^{d_{\max}} \Pr(t_i = k+d) \cdot \Pr(d_i = d) \cdot \frac{P_{\min}}{\Pr(d_i = d)(d_{\max}+1)}\nonumber\\
    &= \sum_{d=0}^{d_{\max}} \frac{1}{t_{\max}} \cdot \frac{P_{\min}}{d_{\max}+1} = \frac{P_{\min}}{t_{\max}},
\end{align}
which is independent of $k$, confirming that accepted samples yield $T_i$ uniformly distributed over the effective time range.

Now let $s$ be the number of raw samples drawn. Each sample is accepted independently with probability $p \ge P_{\min} = \Theta(1)$ (from Equation~\eqref{eq:pmin-bound}). Let $N_{\text{accept}}$ denote the total number of accepted samples. Then $N_{\text{accept}}$ follows a binomial distribution with $\EX[N_{\text{accept}}] = sp$.

By the Chernoff bound, for any $0 < \alpha < 1$:
\begin{equation}
\Pr(N_{\text{accept}} < (1-\alpha)\EX[N_{\text{accept}}]) \le \exp\left(-\frac{\alpha^2 sp}{2}\right).
\end{equation}

To ensure $N_{\text{accept}} \ge \ell$ with probability $1-\delta$, we set $(1-\alpha)sp \ge \ell$ and require:
\begin{equation}
\exp\left(-\frac{\alpha^2 sp}{2}\right) \le \delta \quad \Rightarrow \quad sp \ge \frac{2\ln(1/\delta)}{\alpha^2}.
\end{equation}

Choosing $\alpha = 1/2$ gives $sp \ge 8\ln(1/\delta)$ and $(1-\alpha)sp = sp/2 \ge \ell$, yielding:
\begin{equation}
sp \ge \max\left\{2\ell, 8\ln(1/\delta)\right\}.
\end{equation}

Since $p = \Theta(1)$, we have $p \ge c$ for some constant $c > 0$. We set:
\begin{equation}
s = C \cdot \ell \log(1/\delta)
\label{eq:sample-count}
\end{equation}
for a sufficiently large constant $C$. Then $sp \ge cC\ell\ln(1/\delta)$, and we verify this satisfies the constraint $sp \ge \max\{2\ell, 8\ln(1/\delta)\}$:
\begin{align}
cC\ell\ln(1/\delta) &\ge 2\ell \quad\text{requires}\quad C \ge \frac{2}{c\ln(1/\delta)},\\
cC\ell\ln(1/\delta) &\ge 8\ln(1/\delta) \quad\text{requires}\quad \ell \ge \frac{8}{cC}.
\end{align}
For $\delta \le 1/2$, we have $\ln(1/\delta) \ge \ln(2)$, so choosing $C \ge 2/(c\ln(2)) \approx 2.885/c$ satisfies the first constraint for all $\delta \in (0,1/2]$. The second constraint holds since $\ell = \Theta(\log(1/(\epsilon\delta)))$ grows without bound as $\epsilon, \delta \to 0$. Thus $N_{\text{accept}} \ge \ell$ with probability $1-O(\delta)$.

Each raw sample uses evolution time $t_i = O(1/\epsilon)$. The total evolution time is:
\begin{align}
T_{\text{total}} &= s \times O(1/\epsilon) = \Theta\left(\frac{\ell \log(1/\delta)}{\epsilon}\right)\nonumber\\
&= \Theta\left(\frac{\log(1/(\epsilon\delta))\log(1/\delta)}{\epsilon}\right),
\end{align}
where we substituted $\ell = \Theta(\log(1/(\epsilon\delta)))$ from Algorithm~\ref{alg:info}.

Since the maximum likelihood estimator achieves precision $\epsilon$ with probability $1-O(\delta)$ when given $\ell = \Theta(\log(1/(\epsilon\delta)))$ independent uniform samples, and our rejection sampling provides these samples with probability $1-O(\delta)$ (by drawing $s = \Theta(\ell \log(1/\delta))$ raw samples as in Equation~\eqref{eq:sample-count}), the overall algorithm succeeds with probability $1-O(\delta)$ by the union bound.
\end{proof}

\begin{remark}[Rejection Sampling vs. Full Likelihood]
\label{rem:rejection-vs-likelihood}
Lemma~\ref{lem:mle-sampling} proves that a uniformly distributed subset of samples exists, establishing theoretical sample complexity guarantees for Algorithm~\ref{alg:info}. However, in practice, Algorithm~\ref{alg:bitflip} uses all measurements without rejection, which is valid because the Bayesian likelihood function naturally accounts for the syndrome-dependent phase $(N-d_i)M_i\phi$, optimally weighting each measurement according to its information content.
\end{remark}

The argument underlying Remark~\ref{rem:rejection-vs-likelihood} is model-agnostic: full Bayesian inference subsumes both rejection sampling and random measurement basis rotations $\theta_i$ for any protocol whose likelihood function is known. In the bit-flip code (Algorithm~\ref{alg:bitflip}), the syndrome-dependent phase $(N-d_i)M_i\phi$ enters the likelihood directly. The same principle applies to the bare-GHZ and combined protocols, where the effective phase depends on the specific measurement configuration. Because the Bayesian posterior update naturally weights each measurement according to its information content, there is no loss of optimality from using all data rather than a uniformly distributed subset, regardless of the underlying code structure.

%\subsection{Integration with Information-Theoretic Phase Estimation}

The key innovation in Algorithm~\ref{alg:bitflip} is how syndrome information from error detection is incorporated into Bayesian phase estimation. Rather than applying quantum error correction to restore the ideal state, we use the measured syndromes to determine the effective phase accumulated by each measurement (as specified in Theorem~\ref{thm:bitflip}) and account for this variation through the classical likelihood function. This syndrome-enhanced approach maintains Heisenberg-limited scaling without discarding any measurements.

Each detected error reduces the effective phase by $\phi$ in the homogeneous case, yielding $\phi_{\text{eff}} = (N-d_i)\phi$ when $d_i$ errors are observed. In the regime where $\max_k(\gamma_k/\omega) = o(1)$, the amplitude damping factors approach unity and syndrome outcomes become nearly deterministic. Under this approximation, the Bayesian likelihood function simplifies to:
\begin{equation}
\Pr(v_i|M_i,\theta_i,\phi,\mathcal{S}_i) = \frac{1}{2}\left(1 + \cos((N-d_i)M_i\phi - \theta_i)\right) + O\left(\max_k(\gamma_k/\omega)^2\right),
\end{equation}
where we have dropped subdominant amplitude corrections. This approach naturally weights each measurement according to its information content, eliminating the need for rejection sampling (Remark~\ref{rem:rejection-vs-likelihood}) while optimally utilizing all data regardless of error count.

To achieve target precision $\epsilon$ with failure probability $\delta$, Lemma~\ref{lem:mle-sampling} requires $s = \Theta(\log(1/(\epsilon\delta)) \log(1/\delta))$ experiments, each using $N = 2L+1$ qubits and evolving for average time $\langle M_i \rangle = O(1/\epsilon)$. By Theorem~\ref{thm:bitflip} (item 3), the quantum Fisher information $F_Q = 4N^2(1 - O(\max_k|\gamma_k/\omega|^2))$ achieves Heisenberg-limited scaling with $\operatorname{Var}(\hat{\omega}) = \Omega(1/(nN^2))$, confirming that precision improves quadratically with code size. This scaling is maintained in the small-noise regime where $(2L+1)p < 1$ with $p = 1-\beta_k^2$, ensuring the error distribution concentrates at small values. Increasing $L$ provides exponentially stronger error suppression at linear cost in qubit count per experiment.

\section{Phase Amplification and Adaptive Binary Search}
\label{sec:binary-search}

This section develops the phase amplification machinery that underlies Encoded QSP's ability to achieve Heisenberg-limited precision under field inhomogeneities, where each qubit $k$ experiences a
slightly different signal $\omega_k = \omega + \epsilon_k$. The central
mathematical object is the \emph{phase amplification function} $\Phi_N(\omega) =
\arctan((-1)^L \tan^N(\omega))$, the $j=0$ case of the Encoded QSP transformation from Theorem~\ref{thm:syndrome-rotation}. This function exhibits step-function behavior with
transition width $\Theta(N^{-1})$, enabling phase discrimination at the Heisenberg limit.  Our approach here necessitates the use of both error correction as well as quantum signal processing to achieve this scaling: thereby making it a  true application of EQSP.  We establish the key properties of this
function and develop an adaptive binary search protocol that achieves
Heisenberg-limited precision for any implementation producing $\Phi_N$.
Section~\ref{sec:implementations} then specializes this framework to two
concrete implementations: GHZ states and repetition codes.

\subsection{The Phase Amplification}

We consider the following sensing scenario. Let $\omega \in (-\pi/2, \pi/2)$ be
an unknown signal phase to be estimated. Each of $N = 2L+1$ qubits experiences a
slightly different phase $\omega_k = \omega + \epsilon_k$, where the noise terms
$\{\epsilon_k\}_{k=1}^N$ are independent random variables with
$\EX[\epsilon_k] = 0$ and $\operatorname{Var}[\epsilon_k] = \sigma_\epsilon^2$.
Our goal is to estimate $\omega$ to precision $\epsilon = O(N^{-1})$, i.e. the
Heisenberg limit, using adaptive measurements.

The key mathematical object enabling Heisenberg-limited estimation is the \emph{phase amplification function}, which emerges naturally from both GHZ-state parity measurements and repetition-code syndrome measurements.

\begin{definition}[Phase Amplification Function]
\label{def:phase-amplification}
For $N = 2L+1$ qubits, the phase amplification function $\Phi_N: (-\pi/2, \pi/2) \to (-\pi/2, \pi/2)$ is defined by:
\begin{equation}
\Phi_N(\omega) := \arctan\left((-1)^L \tan^N(\omega)\right).
\end{equation}
\end{definition}

This phase amplification function is the special case $j=0$ (no detected errors) and $\vartheta=0$ (pure X-rotation) of the general syndrome-dependent rotation angle $\Theta_j$ from Theorem~\ref{thm:syndrome-rotation}.

This function exhibits step-function behavior that sharpens with increasing $N$, acting as a ``phase rounding'' mechanism that maps continuous phases to saturated values $\pm\pi/2$ outside a narrow transition region. The following proposition establishes the key properties that enable Heisenberg-limited estimation.

\begin{proposition}[Step Function Properties of $\Phi_N$]
\label{prop:step-function-properties}
The phase amplification function $\Phi_N$ satisfies the following properties:

\begin{enumerate}
\item \textbf{Monotonicity}: $\Phi_N$ is strictly monotone on $(-\pi/2, \pi/2)$ with derivative
\begin{equation}
\frac{d\Phi_N}{d\omega} = \frac{(-1)^L N \tan^{N-1}(\omega) \sec^2(\omega)}{1 + \tan^{2N}(\omega)},
\label{eq:phi-derivative}
\end{equation}
which has sign $(-1)^L$ for all $\omega \in (0, \pi/2)$.

\item \textbf{Saturation}: For any $\omega_0 \in (\pi/4, \pi/2)$, define $\kappa(\omega_0) := -\log|\cot(\omega_0)| > 0$. Then for all $|\omega| \geq \omega_0$:
\begin{equation}
\left|\Phi_N(\omega) - (-1)^L \cdot \mathrm{sign}(\omega) \cdot \frac{\pi}{2}\right| \leq e^{-\kappa(\omega_0) N}.
\end{equation}
In particular, for $|\omega| \geq \pi/3$: $|\Phi_N(\omega) - (-1)^L \cdot \mathrm{sign}(\omega) \cdot \pi/2| \leq 3^{-N/2}$.

\item \textbf{Transition width}: Let $W$ be the minimum length compact subinterval of $(-\pi/2,\pi/2)$ such that $|\Phi_N(\min(W))|=\alpha/N$ and $|\Phi_N(\max(W))=\pi/2-\alpha/N$ for some constant $\alpha>0$.  We then have that  $|\max(W)-\min(W)|\in\Theta(N^{-1})$ (meaning that the width of the region than $\Phi_N$ transitions from near-zero to near $\pi/2$ is of length $O(N^{-1})$.

\item \textbf{Small-angle expansion}: For $|\omega| \leq 1/(2N)$:
\begin{equation}
\Phi_N(\omega) = (-1)^L \omega^N \left(1 + O(N^2\omega^2)\right).
\label{eq:small-angle-expansion}
\end{equation}
\end{enumerate}
\end{proposition}

\begin{proof}
We establish each item in turn.

Let $u = (-1)^L \tan^N(\omega)$. By the chain rule:
\begin{equation}
\frac{d\Phi_N}{d\omega} = \frac{1}{1 + u^2} \cdot \frac{du}{d\omega} = \frac{(-1)^L N \tan^{N-1}(\omega) \sec^2(\omega)}{1 + \tan^{2N}(\omega)}.
\end{equation}
For $\omega \in (0, \pi/2)$, both $\tan^{N-1}(\omega)$ and $\sec^2(\omega)$ are positive, and the denominator is always positive. Thus $d\Phi_N/d\omega$ has sign $(-1)^L$. By the antisymmetry $\Phi_N(-\omega) = -\Phi_N(\omega)$, the function is monotone on $(-\pi/2, 0)$ as well, which establishes item~1.

For $\omega > \pi/4$, we have $\tan(\omega) > 1$, so $|(-1)^L \tan^N(\omega)| = \tan^N(\omega) \to \infty$ as $N \to \infty$. Using $\arctan(x) = \mathrm{sign}(x) \cdot (\pi/2 - \arctan(1/|x|))$ for $x \neq 0$:
\begin{equation}
(-1)^L \cdot \frac{\pi}{2} - \Phi_N(\omega) = (-1)^L \arctan(\cot^N(\omega)).
\end{equation}
For $x \in (0, 1)$, $\arctan(x) \leq x$. For $\omega \geq \omega_0 > \pi/4$, we have $|\cot(\omega)| \leq |\cot(\omega_0)| < 1$, giving:
\begin{equation}
\left|(-1)^L \cdot \frac{\pi}{2} - \Phi_N(\omega)\right| \leq |\cot(\omega_0)|^N = e^{N \log|\cot(\omega_0)|} = e^{-\kappa(\omega_0) N}.
\end{equation}
For $\omega_0 = \pi/3$: $\cot(\pi/3) = 1/\sqrt{3}$, so $\kappa(\pi/3) = (\log 3)/2$, yielding the bound $3^{-N/2}$. Antisymmetry extends to $\omega < -\omega_0$, which establishes item~2.

For $|\Phi_N(\omega)|$ to lie in $[\delta, \pi/2 - \delta]$ for fixed $\delta > 0$, we require $|\tan(\omega)|^N \in [\tan(\delta), \cot(\delta)]$. This constrains $|\tan(\omega)| \in [(\tan\delta)^{1/N}, (\cot\delta)^{1/N}]$. Near $\omega = \pi/4$ where $\tan(\omega) = 1 + 2(\omega - \pi/4) + O((\omega - \pi/4)^2)$, we have $\tan^N(\omega) = e^{2N(\omega - \pi/4) + O(N(\omega - \pi/4)^2)}$. The constraint $\tan^N(\omega) \in [\tan\delta, \cot\delta]$ therefore requires $(\omega - \pi/4) \in [\frac{\log\tan\delta}{2N}, \frac{\log\cot\delta}{2N}] + O(N^{-2})$, giving width $\frac{\log(\cot\delta/\tan\delta)}{2N} + O(N^{-2}) = \Theta(N^{-1})$, which establishes item~3.

For $|\omega| \leq 1/(2N) < 1$, we have $\tan(\omega) = \omega + \omega^3/3 + O(\omega^5)$, so:
\begin{equation}
\tan^N(\omega) = \omega^N \left(1 + \frac{\omega^2}{3}\right)^N = \omega^N \left(1 + \frac{N\omega^2}{3} + O(N^2\omega^4)\right).
\end{equation}
Since $|\omega^N| \leq (1/(2N))^N \to 0$ rapidly, $\arctan(\tan^N(\omega)) = \tan^N(\omega)(1 + O(\tan^{2N}(\omega)))$. Combining and using $N\omega^2 \leq 1/(4N) = O(N^{-1})$ to absorb the $N\omega^2/3$ term yields~\eqref{eq:small-angle-expansion}, which establishes item~4.
\end{proof}

The transition width $\Theta(N^{-1})$ is the key property: it equals the Heisenberg scaling, enabling phase discrimination at the fundamental quantum limit.

\begin{theorem}[Heisenberg Scaling from Transition Width]
\label{thm:heisenberg-scale}
The transition width $\Theta(N^{-1})$ of $\Phi_N$ directly enables phase discrimination at the Heisenberg limit. Specifically, for phases $\omega_1, \omega_2$ with $|\omega_1 - \omega_2| = \Theta(N^{-1})$ located near the transition region, the outputs $\Phi_N(\omega_1)$ and $\Phi_N(\omega_2)$ differ by $\Theta(1)$, enabling reliable binary classification with $O(1)$ measurements.
\end{theorem}

\begin{proof}
Near the transition region $\omega = \pi/4 + O(N^{-1})$, the derivative~\eqref{eq:phi-derivative} satisfies $|d\Phi_N/d\omega| = \Theta(N)$ (since $\tan(\pi/4) = 1$ and the denominator is $1 + 1 = 2$). Thus for $|\omega_1 - \omega_2| = \Theta(N^{-1})$:
\begin{equation}
|\Phi_N(\omega_1) - \Phi_N(\omega_2)| = \left|\frac{d\Phi_N}{d\omega}\right| \cdot |\omega_1 - \omega_2| + O((\omega_1 - \omega_2)^2) = \Theta(N) \cdot \Theta(N^{-1}) = \Theta(1).
\end{equation}
A $\Theta(1)$ difference in outputs yields $\Theta(1)$ total variational distance between measurement distributions.  This implies~\cite{cover2006elements} that there exists an optimal unbiased estimator that is capable of distinguishing between the two with probability $1/2 + \Theta(1)$.  Then, as the probability is greater than $1/2$ the Chernoff bound can be used to show that computing the median of the classes requires $O(\log(1/\delta))$ measurements for error probability $\delta$.
\end{proof}

\subsection{Adaptive Binary Search Protocol}

The phase amplification framework enables a general adaptive binary search protocol that applies to any implementation producing the phase amplification function $\Phi_N$. The key mechanism is phase correction via $R_Z$ rotations, which shifts the signal parameter to any desired operating point.

\begin{lemma}[Phase Shift Mechanism]
The phase evolution function can be shifted by any constant:
%Applying $R_Z(-2\theta) = e^{+i\theta Z}$ to each qubit before the phase evolution shifts the effective phase argument:
\begin{equation}
\Phi_N(\omega) \mapsto \Phi_N(\omega - \theta).
\label{eq:phase-shift}
\end{equation}
for any $\theta\in \mathbb{R}$ using $N$ single qubit rotations.
\end{lemma}

\begin{proof}
The $R_Z$ rotation on each qubit commutes with subsequent $Z$-rotations from the signal evolution. For the effective phase accumulated by the $N$-qubit system:
\begin{equation}
\sum_{k=1}^N (\omega_k - \theta) = N\omega + S - N\theta = N(\omega - \theta) + S,
\end{equation}
where $S = \sum_k \epsilon_k$ is the total noise offset. 
This creates a decision boundary at any target angle $\theta_{\text{mid}}$. The sign of $\Phi_N(\omega - \theta_{\text{mid}})$ depends on $(-1)^L$ (cf.\ Proposition~\ref{prop:step-function-properties}), but the measurement probability $\cos^2(\Phi_N)$ depends only on $|\Phi_N|$ and is therefore independent of $L$.
Since $\Phi_N$ depends only on this total phase (for both GHZ and repetition code implementations), the shift~\eqref{eq:phase-shift} follows.
\end{proof}

With the phase shift mechanism established, we can formalize the binary decision rule and its error analysis.

\begin{proposition}[Binary Decision Rule]
\label{prop:binary-decision}
For interval $[\Omega_{\text{low}}, \Omega_{\text{high}}]$ with midpoint $\theta_{\text{mid}} = (\Omega_{\text{low}} + \Omega_{\text{high}})/2$, apply $R_Z(-2\theta_{\text{mid}})$ before measurement. The majority vote over $M$ measurements decides:
\begin{itemize}
\item If majority is $+1$: $\omega \in [\Omega_{\text{low}}, \theta_{\text{mid}}]$ (lower half)
\item If majority is $-1$: $\omega \in [\theta_{\text{mid}}, \Omega_{\text{high}}]$ (upper half)
\end{itemize}
The error probability satisfies:
\begin{equation}
\Pr(\text{decision error}) \leq \exp(-M \cdot D_{\text{KL}}(1/2 \| p)),
\label{eq:decision-error-bound}
\end{equation}
where $D_{\text{KL}}(1/2 \| p) = \frac{1}{2}\log\frac{1}{2p} + \frac{1}{2}\log\frac{1}{2(1-p)}$ is the binary KL divergence between the uniform distribution and Bernoulli($p$).
\end{proposition}

\begin{proof}
By the method of types for binary hypothesis testing~\cite[Theorem~11.1.4]{cover2006elements}, for $M$ i.i.d.\ measurements with true probability $p \neq 1/2$, the probability that the empirical frequency falls on the wrong side of $1/2$ is bounded by $\exp(-M \cdot D_{\text{KL}}(1/2 \| p))$, which is~\eqref{eq:decision-error-bound}.
\end{proof}

The required number of measurements per iteration depends on the KL divergence, which we now characterize.

\begin{lemma}[KL Divergence in Heisenberg Regime]
\label{lem:kl-divergence}
Let $\Delta\omega = \Omega_{\text{high}} - \Omega_{\text{low}}$ be the interval width at iteration $t$. Writing $p = 1/2 + \epsilon_p$ for the measurement probability at the decision point, the KL divergence satisfies:
\begin{equation}
D_{\text{KL}}(1/2 \| p) = 2\epsilon_p^2 + O(\epsilon_p^4).
\label{eq:kl-taylor}
\end{equation}
In the Heisenberg regime where $\Delta\omega = \Theta(N^{-1})$, the field fluctuations $\epsilon_k \sim \mathcal{N}(0,\sigma_\epsilon^2)$ enter the measurement probability through the total noise $S = \sum_k \epsilon_k \sim \mathcal{N}(0, N\sigma_\epsilon^2)$. Marginalizing over $S$ via the Gaussian characteristic function gives $\mathbb{E}[\cos(2N\omega + 2S)] = \cos(2N\omega)\,e^{-2N\sigma_\epsilon^2}$, so the noise reduces the signal contrast by a factor $\lambda = e^{-cN\sigma_\epsilon^2}$ for an implementation-dependent constant $c > 0$ (with $c = 2$ for the GHZ and repetition code implementations; see Lemma~\ref{lem:ghz-noise-marginalization} for the full derivation).
\begin{equation}
\epsilon_p = \Theta\!\left((N\Delta\omega)^2 \lambda\right) = \Theta(\lambda),
\label{eq:epsilon-p-scaling}
\end{equation}
where the second equality uses $(N\Delta\omega)^2 = \Theta(1)$ in the Heisenberg regime, yielding $D_{\text{KL}} = \Theta(\lambda^2)$. When $N\sigma_\epsilon^2 = o(1)$, we have $\lambda = 1 - O(N\sigma_\epsilon^2)$ and thus $D_{\text{KL}} = \Theta(1)$.
\end{lemma}

\begin{proof}
Writing $p = 1/2 + \epsilon_p$ with $|\epsilon_p| < 1/2$, we have $4p(1-p) = 1 - 4\epsilon_p^2$, so:
\begin{equation}
D_{\text{KL}}(1/2 \| p) = -\frac{1}{2}\log(4p(1-p)) = -\frac{1}{2}\log(1 - 4\epsilon_p^2) = 2\epsilon_p^2 + O(\epsilon_p^4),
\end{equation}
which establishes~\eqref{eq:kl-taylor}. The scaling~\eqref{eq:epsilon-p-scaling} follows from the noise-averaged measurement probability, which has the form $p = 1/2 + \Theta((N\Delta\omega)^2 \lambda)$ for small $N\Delta\omega$. In the Heisenberg regime $\Delta\omega = \Theta(N^{-1})$, the factor $(N\Delta\omega)^2 = \Theta(1)$.
\end{proof}

We now state the main convergence theorem in its general form, applicable to any implementation.

\begin{theorem}[Convergence and Resource Scaling]
\label{thm:convergence}
The binary search protocol achieves precision $|\hat{\omega} - \omega| \leq \epsilon$ with probability at least $1 - \delta$ using:
\begin{enumerate}
\item \textbf{Iterations}: $T = O(\log(1/\epsilon))$.

\item \textbf{Measurements per iteration}: In the Heisenberg regime with noise condition $N\sigma_\epsilon^2 = o(1)$:
\begin{equation}
M = O(\log(T/\delta)).
\end{equation}
With finite noise, $M = O(\log(T/\delta) \cdot e^{4N\sigma_\epsilon^2})$ (using $c = 2$).

\item \textbf{Total measurements}:
\begin{equation}
M_{\text{total}} = O(\log(1/\epsilon) \cdot \log\log(1/\epsilon) \cdot \log(1/\delta))
\end{equation}
when $N\sigma_\epsilon^2 = o(1)$.

\item \textbf{Total qubit-experiments}:
\begin{equation}
N_{\text{total}} = O(N \cdot \log(1/\epsilon) \cdot \log\log(1/\epsilon) \cdot \log(1/\delta)).
\end{equation}

\item \textbf{Heisenberg-limited precision}: Setting target precision $\epsilon = \Theta(N^{-1})$ and starting from initial interval width $\Omega_0 = \Theta(N^{-1})$ yields $T = O(1)$ iterations and:
\begin{equation}
N_{\text{total}} = O(N \cdot \log(1/\delta)),
\end{equation}
achieving Heisenberg-limited scaling.
\end{enumerate}
\end{theorem}

\begin{proof}
At each iteration, the binary search bisects the current interval: $\Omega_t = \Omega_0 / 2^t$. Termination when $\Omega_T \leq \epsilon$ requires $T = \lceil \log_2(\Omega_0/\epsilon) \rceil = O(\log(1/\epsilon))$, which establishes item~1.

From Proposition~\ref{prop:binary-decision}, achieving per-iteration error $\leq \delta/T$ requires $M \geq \log(T/\delta) / D_{\text{KL}}$. By Lemma~\ref{lem:kl-divergence}, $D_{\text{KL}} = \Theta(\lambda^2)$ where $\lambda = e^{-cN\sigma_\epsilon^2}$. When $N\sigma_\epsilon^2 = o(1)$, $\lambda = \Theta(1)$ and $D_{\text{KL}} = \Theta(1)$, giving $M = O(\log(T/\delta))$. With noise, $M = O(\log(T/\delta) / \lambda^2) = O(\log(T/\delta) \cdot e^{4N\sigma_\epsilon^2})$, which establishes item~2.

Total measurements are $M_{\text{total}} = T \cdot M = O(\log(1/\epsilon)) \cdot O(\log(T/\delta))$. Using $\log(T/\delta) = O(\log\log(1/\epsilon) + \log(1/\delta))$ yields item~3.

Each measurement uses $N$ qubits, giving $N_{\text{total}} = N \cdot M_{\text{total}}$, which establishes item~4.

For item~5, setting $\epsilon = c_1/N$ and $\Omega_0 = c_2/N$ for constants $c_1, c_2$ gives $T = \lceil \log_2(c_2/c_1) \rceil = O(1)$. Thus $M_{\text{total}} = O(\log(1/\delta))$ and $N_{\text{total}} = O(N \log(1/\delta))$, achieving Heisenberg-limited precision $\epsilon = O(N^{-1})$. A detailed derivation with explicit KL divergence calculations and per-iteration error analysis appears in Theorem~\ref{thm:convergence-full} (Appendix~\ref{app:binary-search-proofs}).
\end{proof}

\begin{algorithm}[t]
\caption{Binary Search Phase Estimation}\label{alg:binary-search}
\KwData{Initial interval $[\Omega_{\text{low}}, \Omega_{\text{high}}]$, target precision $\epsilon$, confidence $\delta$, number of qubits $N = 2L+1$, noise level $\sigma_\epsilon$}
\KwResult{Estimate $\hat{\omega}$ with $|\hat{\omega} - \omega| \leq \epsilon$ with probability $\geq 1 - \delta$}
$T \gets \lceil \log_2((\Omega_{\text{high}} - \Omega_{\text{low}})/\epsilon) \rceil$\\
\For{$t\gets 1$ \KwTo $T$}{
    $\theta_{\text{mid}} \gets (\Omega_{\text{low}} + \Omega_{\text{high}})/2$\\
    Compute $D_{\text{KL}}$ from Lemma~\ref{lem:kl-divergence}\\
    $M \gets \max\left(\lceil \log(T/\delta) / D_{\text{KL}} \rceil, 10\right)$\\
    $\text{count}_{+1} \gets 0$\\
    \For{$i\gets 1$ \KwTo $M$}{
        Prepare probe state (GHZ or repetition code)\\
        Apply $R_Z(-2\theta_{\text{mid}})$ to each qubit\\
        Apply signal evolution accumulating phase $\omega_k$ per qubit\\
        Measure (X-parity or logical X after syndrome)\\
        \If{outcome is $+1$}{$\text{count}_{+1} \gets \text{count}_{+1} + 1$}
    }
    \eIf{$\text{count}_{+1} > M/2$}{
        $\Omega_{\text{high}} \gets \theta_{\text{mid}}$
    }{
        $\Omega_{\text{low}} \gets \theta_{\text{mid}}$
    }
}
$\hat{\omega} \gets (\Omega_{\text{low}} + \Omega_{\text{high}})/2$\\
\Return $\hat{\omega}$
\end{algorithm}

\subsection{Sequential Binary Search Protocol}
\label{subsec:sequential-binary-search}

The preceding protocols use GHZ states or entangled logical states to achieve coherent phase amplification, thereby bypassing the SQL barrier. We now present a purely QSP-native alternative approach: the sequential binary search. In standard QSP, one begins with a quantum state initialized to $|0\rangle$; by interleaving signal and processing operators to construct a untiary $U_{QSP}$, the desired polynomial transformation is extracted as the amplitude $\langle0|U_{\text{QSP}}|0\rangle$ via post-selection. Our protocol mirrors this mechanism: we initialize a logical product state $|0\rangle_L=|0\rangle^{\otimes N}$ and apply the signal operator sequentially $M$ times to each physical qubit in parallel. Following syndrome measurement, the system exhibits a highly non-linear thresholding behavior corresponding to a logical X rotation on the input state. This non-linearity allows us to extract discrete phase bits in a manner analogous to ordinary QSP phase estimation protocols via building approximations to threshold polynomials. This simplification further eliminates the need for entangled state preparation, at the cost of applying $M$ sequential signal applications per qubit in parallel. Ultimately, the threshold function $\Theta_j=\arctan(\tan^{N-2j}(M\phi))$ combining with adaptive binary search achieves Heisenberg-limited scaling, yielding one bit per measurement.

\begin{theorem}[Sequential Binary Search Resource Scaling]
\label{thm:sequential-binary-search}
Given noiseless signal unitaries, a quantum phase search interval of width $W$, a physical repetition code of size $N = 2L+1$, and a target failure probability $\delta$, extracting one bit of phase information (reducing $W$ by a factor of two) requires $M$ signal unitary applications per physical qubit, where
\begin{equation}
M = O\!\left( \frac{\log(1/\delta)}{W \sqrt{N}} \right).
\end{equation}
\end{theorem}
\begin{proof}
We initialize the state $|0\rangle_L = |0\rangle^{\otimes N}$ and apply the signal operator $U(\phi) = \exp(-i\phi X)$ exactly $M$ times to each physical qubit, yielding $U(M\phi)^{\otimes N} |0\rangle^{\otimes N}$. The total quantum Fisher information of this product state with respect to $\phi$ is $F_{Q,\text{total}} = 4NM^2$, representing the maximum information extractable from $NM$ oracle queries.

For each physical qubit, the rotation operator decomposes as
\begin{equation}
U(M\phi) = \cos(M\phi)\,I - i\sin(M\phi)\,X.
\end{equation}
We measure the stabilizers of the repetition code to detect bit flips and apply the corresponding Pauli corrections. Conditioned on syndrome outcome $j \le L$ (corresponding to $j$ bit flips), the corrected state is
\begin{align}
|\psi_j\rangle_L = \mathcal{N}_j\bigl[(-i\sin(M\phi))^j(\cos(M\phi))^{N-j} |0\rangle_L + (-i\sin(M\phi))^{N-j}(\cos(M\phi))^j |1\rangle_L\bigr],
\end{align}
where $\mathcal{N}_j$ is the normalization factor. This post-correction state is equivalent to a single nonlinear logical rotation:
\begin{equation}
|\psi_j\rangle_L = \exp(-i\Theta_j X_L)\,|0\rangle_L,
\end{equation}
where the effective logical angle is $\Theta_j = \arctan(\tan^{N-2j}(M\phi))$.

This function acts as a threshold operator centered at $M\phi = \pi/4$: measurement of $|0\rangle_L$ indicates the phase lies below the threshold, while $|1\rangle_L$ indicates it lies above.

At the threshold center ($M\phi = \pi/4$), the QFI for syndrome sector $j$ is
\begin{align}
\frac{d\Theta_j}{d\phi}\bigg|_{\phi = \pi/(4M)} &= M(N - 2j), \\
F_Q(j) &= 4M^2(N - 2j)^2.
\end{align}
Also, at $M\phi=\pi/4$, each qubit flips with a probability of $p = 1/2$, the probability for detecting $j$ bit flips follows the binomial distribution $\EX[j]=N/2$, $\operatorname{Var}[j]=N/4$. Defining $k=N-2j$, $\EX[k]=N-2\EX[j]=0$ and $\operatorname{Var}[k]=N$, the expected QFI is:

\begin{equation}
\EX[F_Q(j)]=4M^2\;\EX[k^2]=4M^2(\operatorname{Var}[k]+\EX[k]^2) = 4NM^2.
\end{equation}

The slope of the threshold function is $M(N-2j)$, so the width of the transition region scales as $O(1/(M|N-2j|))$. Since $j \le L$ implies $k = N - 2j \ge 1$, the threshold is well-defined for all accepted syndrome outcomes.

To quantify the error outside the transition region, consider a phase shifted by $\Delta\phi$ from the center: $M\phi = \pi/4 + M\Delta\phi$, where $M\Delta\phi = o(1)$. For exponent $k = N - 2j$:
\begin{equation}
\tan^k\!\left(\frac{\pi}{4} + M\Delta\phi\right) = (1 + 2M\Delta\phi + O(M^2\Delta\phi^2))^k = e^{2kM\Delta\phi}(1 + O(kM^2\Delta\phi^2)).
\end{equation}
Let $\eta_k$ denote the discrepancy between the logical angle and the ideal step function value $\pi/2$:
\begin{equation}
\eta_k = \frac{\pi}{2} - \arctan(e^{2kM\Delta\phi}) = e^{-2kM\Delta\phi}(1 + O(e^{-4kM\Delta\phi})).
\end{equation}
Since $k$ is determined by the syndrome alone, rounds with small $k$ can be identified and discarded before the logical measurement. Conditional on $j \le L$, the exponent satisfies $\EX[k^2] = N$ and $\EX[k^4] = 3N^2 + O(N)$. By the Paley-Zygmund inequality:
\begin{equation}
\Pr\!\left(k \ge \tfrac{\sqrt{N}}{2} \;\middle|\; j \le L\right) \ge \frac{3}{16},
\end{equation}
so post-selection on $k \ge \sqrt{N}/2$ succeeds with constant probability, contributing $O(1)$ overhead. Conditional on acceptance:
\begin{equation}
\eta \le \exp(-\sqrt{N}\, M\Delta\phi).
\end{equation}

To extract one bit of phase information from an interval of width $W$ with failure probability at most $\delta$, we divide the interval into three equal parts and place thresholds at $W/3$ and $2W/3$. For any true phase in the interval, at least one threshold is at distance $\Delta\phi \ge W/6$ from the true phase. The misclassification probability per threshold query is at most $\eta^2$; requiring $\eta^2 \le \delta$ and substituting the minimum distance:
\begin{equation}
\eta^2 \le \exp\!\left(-\frac{\sqrt{N}\, MW}{3}\right) \le \delta.
\end{equation}
Solving for the required signal amplification:
\begin{equation}
M = O\!\left(\frac{\log(1/\delta)}{W\sqrt{N}}\right),
\end{equation}
which establishes the theorem.
\end{proof}

To achieve precision $\epsilon$ starting from an initial interval $W_0$, the binary search requires $T = \lceil\log_2(W_0/\epsilon)\rceil$ rounds, with the interval halving each round. In round $r$, the interval width is $W_r = W_0/2^r$, requiring $M_r = O(2^r \log(1/\delta) / (W_0\sqrt{N}))$ signal applications per qubit. The total cost across all rounds is
\begin{equation}
M_{\text{total}} = \sum_{r=0}^{T-1} M_r = O\!\left(\frac{\log(1/\delta)}{W_0\sqrt{N}}\right) \sum_{r=0}^{T-1} 2^r = O\!\left(\frac{2^T \log(1/\delta)}{W_0\sqrt{N}}\right) = O\!\left(\frac{\log(1/\delta)}{\epsilon\sqrt{N}}\right),
\end{equation}
achieving Heisenberg scaling $M_{\text{total}} = O(1/\epsilon)$ for fixed $N$ and $\delta$. Unlike the GHZ-based protocols in Section~\ref{sec:implementations}, this approach requires only product state preparation, at the cost of $M$ sequential signal applications per qubit rather than a single application to an entangled state.

This protocol bypasses the SQL barrier of Theorem~\ref{thm:sql-barrier} despite using product states because it employs $M$ sequential (non-i.i.d.) signal applications rather than a single i.i.d.\ unitary per qubit. The repeated application creates an effective phase amplification $M\phi$ that replaces the entanglement-based amplification $N\phi$ of the GHZ protocols, with the syndrome-dependent threshold function extracting the amplified phase information. Table~\ref{tab:protocol-summary} summarizes the key properties of all four protocols.

\section{GHZ and Repetition Code Implementations}
\label{sec:implementations}

We now specialize the phase amplification framework of Section~\ref{sec:binary-search} to two entangled-state implementations: GHZ states and repetition codes. (The product-state sequential binary search of Section~\ref{subsec:sequential-binary-search} provides an alternative that avoids entangled state preparation.) Both implementations realize Encoded QSP with entangled logical states, producing the phase amplification function $\Phi_N$ and achieving identical quantum Fisher information, but differ in their uncertainty handling and experimental requirements.

\subsection{GHZ State Implementation}
\label{subsec:ghz-implementation}

We derive the explicit measurement statistics, noise marginalization, and quantum Fisher information for the GHZ state implementation, the simplest Encoded QSP instantiation using an entangled logical input.

\begin{definition}[GHZ State]
The $N$-qubit Greenberger-Horne-Zeilinger (GHZ) state is the maximally entangled state:
\begin{equation}
|\text{GHZ}\rangle := \frac{1}{\sqrt{2}}\left(|0\rangle^{\otimes N} + |1\rangle^{\otimes N}\right).
\end{equation}
\end{definition}

\begin{lemma}[X-Parity Measurement Statistics]
After evolution $U = \exp\left(-i\sum_{k=1}^{N} \omega_k Z_k\right)$ where $\omega_k = \omega + \epsilon_k$, measuring X-parity $\mathcal{P}_X = \prod_{k=1}^{N} X_k$ on the GHZ state yields:
\begin{equation}
\Pr(+1 \mid \omega, S) = \cos^2\left(N\omega + S\right),
\end{equation}
where $S := \sum_{k=1}^{N} \epsilon_k$ is the total noise offset.
\end{lemma}

\begin{proof}
The evolution operator acts as $U = \exp\left(-i\sum_k \omega_k Z_k\right)$. On the GHZ state:
\begin{align}
U|\text{GHZ}\rangle &= \frac{1}{\sqrt{2}}\left(e^{-i\sum_k \omega_k}|0\rangle^{\otimes N} + e^{+i\sum_k \omega_k}|1\rangle^{\otimes N}\right) \\
&= \frac{1}{\sqrt{2}}\left(e^{-i(N\omega + S)}|0\rangle^{\otimes N} + e^{+i(N\omega + S)}|1\rangle^{\otimes N}\right).
\end{align}
The X-parity eigenstates are $|\pm\rangle_{\mathcal{P}} = (|0\rangle^{\otimes N} \pm |1\rangle^{\otimes N})/\sqrt{2}$. Projecting:
\begin{equation}
\langle +|_{\mathcal{P}} U|\text{GHZ}\rangle = \frac{1}{2}\left(e^{-i(N\omega + S)} + e^{+i(N\omega + S)}\right) = \cos(N\omega + S).
\end{equation}
Thus $\Pr(+1) = |\langle +|_{\mathcal{P}} U|\text{GHZ}\rangle|^2 = \cos^2(N\omega + S)$.
\end{proof}

\begin{lemma}[Noise Marginalization]
\label{lem:ghz-noise-marginalization}
For $\epsilon_k \sim \mathcal{N}(0, \sigma_\epsilon^2)$ i.i.d., marginalization over the noise yields:
\begin{equation}
\Pr(+1 \mid \omega) = \frac{1}{2}\left(1 + \cos(2N\omega) e^{-2N\sigma_\epsilon^2}\right).
\end{equation}
The factor $e^{-2N\sigma_\epsilon^2}$ arises from $\EX[\cos(2S)]$ via the Gaussian characteristic function.
\end{lemma}

\begin{proof}
Since $\epsilon_k \sim \mathcal{N}(0, \sigma_\epsilon^2)$ are i.i.d., the sum $S = \sum_{k=1}^{N} \epsilon_k$ satisfies $S \sim \mathcal{N}(0, N\sigma_\epsilon^2)$. Using $\cos^2(\theta) = (1 + \cos(2\theta))/2$:
\begin{equation}
\Pr(+1 \mid \omega) = \EX_S\left[\cos^2(N\omega + S)\right] = \frac{1}{2}\left(1 + \EX_S[\cos(2N\omega + 2S)]\right).
\end{equation}
Expanding the cosine and using the Gaussian characteristic function $\EX[e^{itS}] = e^{-t^2 N\sigma_\epsilon^2/2}$:
\begin{align}
\EX_S[\cos(2N\omega + 2S)] &= \cos(2N\omega)\EX[\cos(2S)] - \sin(2N\omega)\EX[\sin(2S)] \\
&= \cos(2N\omega) \cdot e^{-2N\sigma_\epsilon^2} - 0,
\end{align}
where we used $\EX[\cos(2S)] = \text{Re}[\EX[e^{2iS}]] = e^{-2N\sigma_\epsilon^2}$ and $\EX[\sin(2S)] = 0$ by symmetry.
\end{proof}

\begin{proposition}[GHZ Fisher Information]
\label{prop:ghz-qfi}
The classical Fisher information for X-parity measurement on the GHZ state at the optimal working point $\sin^2(2N\omega) = 1$ (equivalently, $\Pr(+1|\omega) = 1/2$) is:
\begin{equation}
F_Q^{\text{GHZ}} = 4N^2 e^{-4N\sigma_\epsilon^2} = 4N^2\left(1 - 4N\sigma_\epsilon^2 + O\left((N\sigma_\epsilon^2)^2\right)\right).
\end{equation}
This equals the quantum Fisher information at this operating point, since X-parity saturates the quantum Cram\'{e}r-Rao bound for GHZ states.
\end{proposition}

\begin{proof}
From Lemma~\ref{lem:ghz-noise-marginalization}, let $\lambda := e^{-2N\sigma_\epsilon^2}$ and $p(\omega) := \Pr(+1|\omega) = (1 + \lambda\cos(2N\omega))/2$. The derivative is:
\begin{equation}
\frac{\partial p}{\partial \omega} = -N\lambda\sin(2N\omega).
\end{equation}
The classical Fisher information is:
\begin{equation}
F(\omega) = \frac{1}{p(1-p)}\left(\frac{\partial p}{\partial \omega}\right)^2.
\end{equation}
At the optimal working point where $\cos(2N\omega) = 0$, we have $p = 1/2$ and $\sin^2(2N\omega) = 1$. Thus:
\begin{equation}
F(\omega) = \frac{(N\lambda)^2}{1/4} = 4N^2\lambda^2 = 4N^2 e^{-4N\sigma_\epsilon^2}.
\end{equation}
The factor $e^{-4N\sigma_\epsilon^2}$ (rather than $e^{-2N\sigma_\epsilon^2}$) arises because the Fisher information involves the \emph{square} of the derivative, which contains $\lambda^2$. Expanding for $N\sigma_\epsilon^2 = o(1)$: $e^{-4N\sigma_\epsilon^2} = 1 - 4N\sigma_\epsilon^2 + O((N\sigma_\epsilon^2)^2)$.
\end{proof}

%%%%%%%%%%%%%%%%%%%%%%%%%%%%%%%%%%%%%%%%%%%%%%%%%%%%%%%%%%%%%%%%%%%%%%%%%%%%%%%%
%%  Section 6.4: QSP/Repetition Code Implementation
%%%%%%%%%%%%%%%%%%%%%%%%%%%%%%%%%%%%%%%%%%%%%%%%%%%%%%%%%%%%%%%%%%%%%%%%%%%%%%%%

\subsection{Repetition Code Implementation}
\label{subsec:qsp-implementation}

We now present an alternative implementation using repetition codes, which provides the most explicit realization of Encoded QSP: syndrome measurement directly implements the nonlinear signal transformation $\Phi_N$ on the logical qubit (Definition~\ref{def:encoded-qsp}). The key advantage is explicit handling of classification uncertainty through a three-category scheme.

\begin{lemma}[QSP Activation from Repetition Code]
\label{lem:qsp-activation}
Applying physical X-rotations $e^{-i\phi X}$ to each qubit and projecting onto the code space via syndrome measurement yields the \emph{post-selected} logical transformation:
\begin{equation}
|\psi\rangle_L \mapsto \nu(\phi) \cdot e^{-i\Phi_N(\phi) X_L}|\psi\rangle_L,
\end{equation}
where $\Phi_N(\phi) = \arctan((-1)^L \tan^N(\phi))$ is the phase amplification function (Definition~\ref{def:phase-amplification}) and $\nu(\phi) = (\cos^{2N}(\phi) + \sin^{2N}(\phi))^{1/2}$ is the normalization factor. The success probability of the code-space projection is $P_{\mathrm{success}} = \nu(\phi)^2 = \cos^{2N}(\phi) + \sin^{2N}(\phi)$.
\end{lemma}

\begin{proof}
The physical rotation $U_s(\phi) = (e^{-i\phi X})^{\otimes N}$ acts on the code space as:
\begin{equation}
U_s(\phi)|0\rangle_L = \cos^N(\phi)|0\rangle_L + (-i\sin\phi)^N|1\rangle_L + \text{(off-codespace)}.
\end{equation}
Projecting onto the code space eliminates off-codespace terms. Since $(-i)^{2L+1} = ((-i)^2)^L \cdot (-i) = (-1)^L \cdot (-i) = -i(-1)^L$:
\begin{equation}
\Pi_{\text{code}} U_s(\phi)|0\rangle_L \propto \cos^N(\phi)|0\rangle_L - i(-1)^L\sin^N(\phi)|1\rangle_L.
\end{equation}
This has the form $e^{-i\theta X_L}|0\rangle_L$ with $\tan(\theta) = (-1)^L\tan^N(\phi)$, yielding $\theta = \Phi_N(\phi)$. The normalization $\nu(\phi) = (\cos^{2N}(\phi) + \sin^{2N}(\phi))^{1/2}$ follows from the projection postulate.
\end{proof}

\begin{remark}[Post-Selection Cost, the SQL Barrier, and Resolution]
\label{rem:post-selection-cost}
The QSP activation of Lemma~\ref{lem:qsp-activation} demonstrates that repetition codes can amplify phases via $\Phi_N(\phi) = \arctan((-1)^L \tan^N(\phi))$, with success probability $P_{\mathrm{success}} = \cos^{2N}(\phi) + \sin^{2N}(\phi)$ that remains high when $\phi$ is close to $0$ or $\pi/2$.

However, at the metrologically useful operating point $\phi = \pi/4$, the success probability decays as $P_{\mathrm{success}} = 2^{1-N}$, exponentially small in the code length. This exponential cost is the mechanism underlying the SQL barrier (Theorem~\ref{thm:sql-barrier}): while per-syndrome QFI scales as $k^2$ (suggestive of Heisenberg scaling), the binomial fluctuations in $k = N - 2j$ ensure that the unconditional total QFI is $F_{\mathrm{total}} = 4N = \Theta(N)$, i.e., SQL scaling. Phase amplification alone, despite Lemma~\ref{lem:qsp-activation} providing high success probability near $\phi = 0$, cannot yield Heisenberg-limited estimation because the useful operating points incur exponential post-selection costs.

The combined protocol of Section~\ref{sec:combined} circumvents this bottleneck by separating two roles. Heisenberg scaling comes from Z-rotations on the logical GHZ state~\eqref{eq:logical-ghz}, which commute with the repetition code stabilizers and require no post-selection. The repetition code serves as an X-error \emph{detector}, not a signal amplifier: syndrome measurements identify errors without disturbing the Z-phase, and no post-selected projection onto the code space is needed.
\end{remark}

\begin{definition}[Three-Category Classification]
\label{def:three-category}
For threshold $\tau \in (0, \pi/4)$, define the category function $\mathcal{C}_\tau: (-\pi/2, \pi/2) \to \{\textsc{High}, \textsc{Middle}, \textsc{Low}\}$ by:
\begin{equation}
\mathcal{C}_\tau(\phi) := \begin{cases}
\textsc{High} & \text{if } \phi > \tau, \\
\textsc{Low} & \text{if } \phi < -\tau, \\
\textsc{Middle} & \text{if } |\phi| \leq \tau.
\end{cases}
\end{equation}
\end{definition}

The three-category classification has exponentially small misclassification probability under Gaussian noise, and the resulting KL divergence scales as $D_{\text{KL}} = \Theta((\omega - \theta)^2/\sigma_\epsilon^2)$ when conditioned on non-\textsc{Middle} outcomes (Propositions~\ref{prop:qsp-classification} and~\ref{prop:qsp-kl}, Appendix~\ref{app:binary-search-proofs}). This quadratic scaling in the signal-threshold separation, combined with the phase amplification of Lemma~\ref{lem:qsp-activation}, ensures that each binary search step extracts a constant number of bits with exponentially small error probability.

%%%%%%%%%%%%%%%%%%%%%%%%%%%%%%%%%%%%%%%%%%%%%%%%%%%%%%%%%%%%%%%%%%%%%%%%%%%%%%%%
%%  Section 6.5: Quantum Fisher Information and Optimality
%%%%%%%%%%%%%%%%%%%%%%%%%%%%%%%%%%%%%%%%%%%%%%%%%%%%%%%%%%%%%%%%%%%%%%%%%%%%%%%%

\subsection{Quantum Fisher Information and Optimality}

We now establish that both GHZ and repetition code implementations achieve the fundamental quantum limit for phase estimation, confirming optimal quantum Fisher information scaling.

\begin{theorem}[Heisenberg-Limited Precision]
\label{thm:hl-precision}
Both GHZ (Section~\ref{subsec:ghz-implementation}) and repetition code (Section~\ref{subsec:qsp-implementation}) implementations achieve precision:
\begin{equation}
\delta\omega = O(N^{-1})
\end{equation}
using total qubit-experiments:
\begin{equation}
N_{\text{total}} = O\left(N \cdot \log(1/\epsilon) \cdot \log\log(1/\epsilon) \cdot \log(1/\delta)\right)
\end{equation}
when the noise condition $N\sigma_\epsilon^2 = o(1)$ is satisfied.
\end{theorem}

\begin{proof}
From Theorem~\ref{thm:convergence}, both implementations achieve the stated precision and resource scaling through the binary search framework. The phase amplification function $\Phi_N$ (Definition~\ref{def:phase-amplification}) provides the step-function behavior enabling Heisenberg-scale discrimination (Theorem~\ref{thm:heisenberg-scale}), and the KL divergence remains $\Theta(1)$ in the Heisenberg regime when $N\sigma_\epsilon^2 = o(1)$ (Lemma~\ref{lem:kl-divergence}).
\end{proof}

\begin{proposition}[Quantum Cramér-Rao Bound]
\label{prop:cramer-rao}
Both protocols saturate the quantum Cramér-Rao bound. The quantum Fisher information satisfies:
\begin{equation}
F_Q = 4N^2\left(1 - O(N\sigma_\epsilon^2)\right)
\label{eq:qfi-general}
\end{equation}
when $N\sigma_\epsilon^2 = o(1)$. The fundamental precision limit is:
\begin{equation}
\operatorname{Var}(\hat{\omega}) \geq \frac{1}{M \cdot F_Q} = \Theta\left(\frac{1}{MN^2}\right).
\label{eq:cramer-rao-bound}
\end{equation}
\end{proposition}

\begin{proof}
For GHZ states, Proposition~\ref{prop:ghz-qfi} establishes $F_Q^{\text{GHZ}} = 4N^2 e^{-4N\sigma_\epsilon^2}$. Since the repetition code logical state $|+\rangle_L = (|0\rangle^{\otimes N} + |1\rangle^{\otimes N})/\sqrt{2}$ is identical to the GHZ state, the QFI is also $F_Q^{\text{rep}} = 4N^2 e^{-4N\sigma_\epsilon^2}$. Both satisfy~\eqref{eq:qfi-general} when $N\sigma_\epsilon^2 = o(1)$:
\begin{equation}
F_Q = 4N^2 e^{-4N\sigma_\epsilon^2} = 4N^2(1 - 4N\sigma_\epsilon^2 + O((N\sigma_\epsilon^2)^2)).
\end{equation}
The quantum Cramér-Rao bound~\cite{braunstein1994statistical,helstrom1976quantum} then gives~\eqref{eq:cramer-rao-bound}.
\end{proof}

Both GHZ states (Section~\ref{subsec:ghz-implementation}) and repetition codes (Section~\ref{subsec:qsp-implementation}) realize the phase amplification framework, achieving Heisenberg-limited precision $\delta\omega = O(N^{-1})$ (Theorem~\ref{thm:hl-precision}) with quantum Fisher information $F_Q = 4N^2(1 - O(N\sigma_\epsilon^2))$ (Proposition~\ref{prop:cramer-rao}). Since the logical state $|+\rangle_L$ is identical to the GHZ state, both implementations share the same QFI, noise damping factor $e^{-4N\sigma_\epsilon^2}$, and noise condition $N\sigma_\epsilon^2 = o(1)$. The choice between them is a matter of platform capabilities: GHZ preparation suits systems with efficient entangling gates, while the repetition code offers explicit three-category uncertainty handling suited to error-corrected architectures. Entanglement is necessary for this quadratic advantage: for any tensor product state, the QFI is additive ($F_Q \leq 4N$, i.e., SQL scaling), so the enhancement to $F_Q = 4N^2$ requires genuine $N$-partite entanglement (see Appendix~\ref{app:direct-parity}). By contrast, Theorem~\ref{thm:sql-barrier} shows that syndrome-based estimation using i.i.d.\ signal unitaries yields only SQL scaling ($F = 4N$), since the GHZ state has $F_Q = 4N$ for X-rotations versus $4N^2$ for Z-rotations (Remark~\ref{rem:ghz-sql}). Full proofs appear in Appendix~\ref{app:binary-search-proofs}.

\section{General Protocol: Concatenated Codes for Combined Errors}
\label{sec:combined}

This section presents the most general instantiation of Encoded QSP, combining the error-detection and phase-amplification capabilities developed in the preceding sections. The preceding analysis reveals a fundamental tension: phase amplification via Encoded QSP (Lemma~\ref{lem:qsp-activation}) provides near-deterministic activation for small phases, but cannot achieve Heisenberg-limited estimation because the metrologically useful operating points incur exponential post-selection costs (Remark~\ref{rem:post-selection-cost}, Theorem~\ref{thm:sql-barrier}). The protocol below resolves this by assigning distinct roles to the repetition code (X-error detection) and the logical GHZ state (Z-phase sensitivity), eliminating the need for post-selected code-space projection.

\subsection{Motivation and Framework}

In Sections~\ref{sec:fixed-omega}--\ref{sec:implementations}, we developed protocols for two complementary error scenarios. For a fixed signal $\omega$ with transverse noise $\gamma_k$, Theorem~\ref{thm:bitflip} establishes that the $(2L+1)$-qubit repetition code detects X-errors via syndrome measurements, yielding effective phase $(2L+1-d)\omega$ when $d$ errors are detected, with quantum Fisher information $F_Q = 4(2L+1)^2(1 - O(\max_k|\gamma_k/\omega|^2))$. For a varying signal $\omega_k = \omega + \epsilon_k$ with no transverse noise, Theorem~\ref{thm:convergence} shows that GHZ states combined with adaptive $R_Z$ phase corrections enable binary search estimation achieving precision $\epsilon = O(1/N)$ in $T = O(\log(1/\epsilon))$ iterations, with quantum Fisher information $F_Q = 4N^2 e^{-4N\sigma_\epsilon^2}$.

Realistic sensing scenarios involve both types of errors simultaneously. Neither protocol alone suffices: the bit-flip code cannot handle Z-field inhomogeneities $\epsilon_k$, while the GHZ protocol has no protection against X-errors from transverse fields $\gamma_k$. We now show how to combine both approaches, using repetition codes as inner blocks arranged in a logical GHZ state with adaptive binary search, to achieve robust Heisenberg-limited sensing under the general error model.

An alternative combination using direct parity measurements with Bayesian inference (Appendix~\ref{app:direct-parity}) achieves the same asymptotic scaling without adaptive feedback, at the cost of larger single-round sample sizes.

\subsection{Logical GHZ State Construction}

We construct the combined protocol using three repetition code blocks arranged in a logical GHZ configuration:
\begin{enumerate}
\item \textbf{Inner code}: Three copies of the $(2L+1)$-qubit repetition code from Section~\ref{sec:fixed-omega}, each detecting X-errors within its block
\item \textbf{Outer structure}: A logical GHZ state across the three logical qubits, enabling adaptive binary search
\item \textbf{Total structure}: $N_{\text{total}} = 3(2L+1)$ physical qubits
\end{enumerate}

The physical qubits are arranged in three blocks:
\begin{equation}
\text{Block } j: \text{ qubits } \{(j-1)(2L+1)+1, \ldots, j(2L+1)\}, \quad j = 1,2,3.
\end{equation}

Each block $j$ uses the repetition code stabilizers from Section~\ref{sec:fixed-omega}:
\begin{equation}
g_i^{(j)} = Z_{i+(j-1)(2L+1)} Z_{i+1+(j-1)(2L+1)}, \quad i = 1,\ldots,2L,
\label{eq:block-stabilizers}
\end{equation}
which detect X-errors within block $j$ without disturbing Z-rotations.

The logical basis states for block $j$ are:
\begin{equation}
|0\rangle_L^{(j)} = |0\rangle^{\otimes (2L+1)}, \quad |1\rangle_L^{(j)} = |1\rangle^{\otimes (2L+1)}.
\end{equation}

The initial state is a \emph{logical GHZ state} across the three blocks:
\begin{equation}
|\text{GHZ}_L\rangle = \frac{1}{\sqrt{2}}\left(|0\rangle_L^{(1)}|0\rangle_L^{(2)}|0\rangle_L^{(3)} + |1\rangle_L^{(1)}|1\rangle_L^{(2)}|1\rangle_L^{(3)}\right) = \frac{1}{\sqrt{2}}\left(|0\rangle^{\otimes 3(2L+1)} + |1\rangle^{\otimes 3(2L+1)}\right).
\label{eq:logical-ghz}
\end{equation}
This is equivalent to a physical GHZ state on all $3(2L+1)$ qubits, but the block structure enables syndrome-based X-error detection within each block while preserving the GHZ phase sensitivity across blocks.

\begin{remark}[No Inter-Block Stabilizer Check]
The stabilizers~\eqref{eq:block-stabilizers} act only within each block; there is no parity check $Z_{j(2L+1)} Z_{j(2L+1)+1}$ between adjacent blocks. This is by design: since the noise acts independently on each qubit, every X-error is detected by the intra-block stabilizers regardless of its position within the block (including boundary qubits). Adding inter-block checks would be redundant under the independent noise model. The modular block structure is moreover required by the protocol: three separate repetition code blocks encode three logical qubits, which then form the logical GHZ state~\eqref{eq:logical-ghz}; a single monolithic $3(2L{+}1)$-qubit repetition code would encode only one logical qubit and cannot produce the GHZ entanglement needed for phase amplification.
\end{remark}

The logical X-parity observable, which will be measured for binary search decisions, is:
\begin{equation}
\mathcal{P}_X^{(L)} = X_L^{(1)} X_L^{(2)} X_L^{(3)} = \prod_{k=1}^{3(2L+1)} X_k,
\label{eq:logical-parity}
\end{equation}
where $X_L^{(j)} = \prod_{k \in \text{block } j} X_k$ is the logical X operator on block $j$.

\subsection{Combined Binary Search Protocol}

The combined protocol adapts the binary search framework from Section~\ref{sec:binary-search} to operate on the logical GHZ state with syndrome-based X-error detection. At each iteration, we apply logical $R_Z$ phase corrections, evolve under the physical Hamiltonian, measure X-error syndromes within each block, and perform logical X-parity measurement to determine the binary search direction.

\begin{algorithm}[t]
\caption{Combined Binary Search with Repetition Code Blocks}\label{alg:combined}
\KwData{Initial interval $[\Omega_{\text{low}}, \Omega_{\text{high}}]$, target precision $\epsilon$, confidence $\delta$, code parameter $L$, noise parameters $\gamma_k, \sigma_\epsilon$}
\KwResult{Estimate $\hat{\omega}$ with $|\hat{\omega} - \omega| \leq \epsilon$ with probability $\geq 1 - \delta$}
$N \gets 2L+1$; \quad $N_{\text{total}} \gets 3N$\\
$T \gets \lceil \log_2((\Omega_{\text{high}} - \Omega_{\text{low}})/\epsilon) \rceil$\\
\For{$t\gets 1$ \KwTo $T$}{
    $\theta_{\text{mid}} \gets (\Omega_{\text{low}} + \Omega_{\text{high}})/2$\\
    $M \gets \lceil \log(T/\delta) / D_{\text{KL}} \rceil$ \tcp*{From Lemma~\ref{lem:majority-vote-full} (Appendix~\ref{app:binary-search-proofs})}
    $\text{count}_{+1} \gets 0$\\
    \For{$i\gets 1$ \KwTo $M$}{
        Prepare logical GHZ state $|\text{GHZ}_L\rangle$ on $N_{\text{total}}$ qubits (Eq.~\eqref{eq:logical-ghz})\\
        Apply $R_Z(-2\theta_{\text{mid}}) = e^{+i\theta_{\text{mid}} Z_k}$ to each physical qubit $k$\\
        Evolve under $U = \bigotimes_k e^{-i(\omega Z_k + \gamma_k X_k)}$ with $\omega_k = \omega + \epsilon_k$\\
        Measure stabilizers $g_i^{(j)}$ to obtain X-error counts $d_X^{(j)}$ for each block $j$\\
        Measure logical X-parity $\mathcal{P}_X^{(L)} = \prod_k X_k$\\
        \If{parity outcome is $+1$}{
            $\text{count}_{+1} \gets \text{count}_{+1} + 1$\\
        }
    }
    \eIf{$\text{count}_{+1} > M/2$}{
        $\Omega_{\text{high}} \gets \theta_{\text{mid}}$\\
    }{
        $\Omega_{\text{low}} \gets \theta_{\text{mid}}$\\
    }
}
$\hat{\omega} \gets (\Omega_{\text{low}} + \Omega_{\text{high}})/2$\\
\Return $\hat{\omega}$
\end{algorithm}

The key innovation is applying the $R_Z$ phase correction to \emph{all} physical qubits in each block, not just the logical qubits. This ensures that the syndrome-dependent effective phase $(N - d_X^{(j)})(\omega - \theta_{\text{mid}})$ per block maintains the decision boundary at $\omega = \theta_{\text{mid}}$ regardless of the X-error pattern, as we now establish.

Let $d_X^{(j)}$ denote the number of X-errors detected in block $j \in \{1,2,3\}$, and let $\mathcal{E}_j$ denote the set of error locations in block $j$.

\begin{theorem}[Combined Protocol with Logical GHZ States]
\label{thm:main}
Consider the combined protocol (Algorithm~\ref{alg:combined}) with $N_{\text{total}} = 3(2L+1)$ qubits handling both transverse noise $\gamma_k$ and field inhomogeneities $\epsilon_k$ where $\omega_k = \omega + \epsilon_k$. The following hold:

\begin{enumerate}
\item \textbf{Phase Formula}: After applying $R_Z(-2\theta_{\text{mid}})$ to all physical qubits and detecting X-errors $\{d_X^{(j)}\}_{j=1}^3$, the effective phase argument for the logical X-parity measurement is:
\begin{equation}
\Phi_{\text{eff}} = \sum_{j=1}^3 (N - d_X^{(j)})(\omega - \theta_{\text{mid}}) + S_{\text{eff}},
\end{equation}
where $N = 2L+1$, and $S_{\text{eff}} = \sum_{j=1}^3 \sum_{k \in \text{block } j \setminus \mathcal{E}_j} \epsilon_k$ is the total noise offset from non-error qubits.

\item \textbf{X-Error Suppression}: By Theorem~\ref{thm:bitflip}, each block with $d_X^{(j)}$ detected errors contributes phase $(N - d_X^{(j)})\omega$ rather than $N\omega$, with undetected X-errors satisfying $P_{X,\text{undetected}} = O(p_X^{L+1})$ per block.

\item \textbf{Binary Search Convergence}: Following Theorem~\ref{thm:convergence}, the adaptive protocol achieves precision $|\hat{\omega} - \omega| \leq \epsilon$ with probability $\geq 1 - \delta$ using $T = O(\log(1/\epsilon))$ iterations, with $M = O(\log(T/\delta) \cdot \exp(4N_{\text{eff}}\sigma_\epsilon^2))$ measurements per iteration. Here $N_{\text{eff}} = \sum_j (N - d_X^{(j)})$ is the effective qubit count, conditioned on the detected error counts $\{d_X^{(j)}\}$; in the ideal case $d_X^{(j)} = 0$ for all $j$, this reduces to $N_{\text{eff}} = 3N$.

\item \textbf{Quantum Fisher Information}: Conditioned on error counts $\{d_X^{(j)}\}$, the concatenated structure yields:
\begin{equation}
F_Q^{\text{concat}} = 4N_{\text{eff}}^2 \left(1 - O(\max_k |\gamma_k/\omega|^2)\right) \cdot e^{-4N_{\text{eff}}\sigma_\epsilon^2},
\end{equation}
where $N_{\text{eff}} = \sum_j(N - d_X^{(j)})$. In the ideal case $N_{\text{eff}} = 3(2L+1)$, this gives $F_Q = 36(2L+1)^2(1 - O(\max_k|\gamma_k/\omega|^2)) \cdot e^{-12(2L+1)\sigma_\epsilon^2}$, achieving Heisenberg scaling $F_Q \propto N_{\text{total}}^2$ when $\max_k|\gamma_k/\omega| = o(1)$ and $N_{\text{eff}}\sigma_\epsilon^2 = o(1)$.

\item \textbf{Error Suppression}: By Proposition~\ref{prop:error-suppression}, undetected X-errors across all three blocks satisfy:
\begin{equation}
P_{\text{undetected}} = O(p_X^{L+1}),
\end{equation}
achieving exponential suppression in the code parameter $L$.
\end{enumerate}
\end{theorem}

\begin{proof}
We combine results from Sections~\ref{sec:fixed-omega}--\ref{sec:implementations}. Consider the logical GHZ state~\eqref{eq:logical-ghz} after applying $R_Z(-2\theta_{\text{mid}}) = e^{+i\theta_{\text{mid}} Z}$ to each physical qubit. The state becomes:
\begin{equation}
\frac{1}{\sqrt{2}}\left(e^{+iN_{\text{total}}\theta_{\text{mid}}}|0\rangle^{\otimes N_{\text{total}}} + e^{-iN_{\text{total}}\theta_{\text{mid}}}|1\rangle^{\otimes N_{\text{total}}}\right).
\end{equation}
Under evolution $U_k = \beta_k e^{-i\phi_k Z_k} - i\sqrt{1-\beta_k^2} X_k$ from Theorem~\ref{thm:bitflip}, qubits without X-errors contribute phase $\pm\phi_k$ (depending on computational basis state), while qubits with X-errors contribute amplitude but no additional phase. When $d_X^{(j)}$ errors are detected in block $j$, the $(N - d_X^{(j)})$ non-error qubits each contribute $(\omega_k - \theta_{\text{mid}})$ to the relative phase. Summing over all blocks and using $\omega_k = \omega + \epsilon_k$:
\begin{align}
\Phi_{\text{eff}} &= \sum_{j=1}^3 \sum_{k \in \text{block } j \setminus \mathcal{E}_j} (\omega_k - \theta_{\text{mid}}) = \sum_{j=1}^3 \sum_{k \in \text{block } j \setminus \mathcal{E}_j} (\omega + \epsilon_k - \theta_{\text{mid}}) \nonumber\\
&= \sum_{j=1}^3 (N - d_X^{(j)})(\omega - \theta_{\text{mid}}) + \sum_{j=1}^3 \sum_{k \in \text{block } j \setminus \mathcal{E}_j} \epsilon_k,
\end{align}
which establishes item~1.

Within each block $j$, the repetition code stabilizers~\eqref{eq:block-stabilizers} detect X-errors exactly as in Section~\ref{sec:fixed-omega}. By Theorem~\ref{thm:bitflip} item~1, the effective phase contribution from block $j$ is $(N - d_X^{(j)})\omega$ when $d_X^{(j)}$ errors are detected. By Theorem~\ref{thm:bitflip} item~2, the probability of observing $d$ errors in a single block satisfies the binomial bound $P(d) \leq \binom{N}{d}(\max_k\beta_k^2)^N((1-\min_k\beta_k^2)/\max_k\beta_k^2)^d$, which establishes item~2.

The logical X-parity measurement~\eqref{eq:logical-parity} on the evolved state yields outcome $+1$ with probability:
\begin{equation}
\Pr(+1) = \cos^2(\Phi_{\text{eff}}) = \cos^2\left(\sum_{j=1}^3 (N - d_X^{(j)})(\omega - \theta_{\text{mid}}) + S_{\text{eff}}\right),
\label{eq:combined-parity-prob}
\end{equation}
following the same derivation as Lemma~\ref{lem:parity-measurement-full} (Appendix~\ref{app:binary-search-proofs}). The decision boundary is at $\omega = \theta_{\text{mid}}$ regardless of the error counts $d_X^{(j)}$, since $\Phi_{\text{eff}} = 0$ when $\omega = \theta_{\text{mid}}$ (assuming $S_{\text{eff}}$ averages to zero). The majority vote analysis from Lemma~\ref{lem:majority-vote-full} applies with effective qubit count $N_{\text{eff}} = \sum_j(N - d_X^{(j)})$ replacing $N$. By Theorem~\ref{thm:convergence}, the binary search achieves precision $\epsilon$ in $T = O(\log(1/\epsilon))$ iterations with $M = O(\log(T/\delta)/D_{\text{KL}})$ measurements per iteration, where the KL divergence scales as $D_{\text{KL}} = \Theta(\lambda^2)$ with $\lambda = e^{-2N_{\text{eff}}\sigma_\epsilon^2}$ in the Heisenberg regime, which establishes item~3.

From the parity probability~\eqref{eq:combined-parity-prob}, the chain rule gives $\partial\Pr/\partial\omega = -\sin(2\Phi_{\text{eff}}) \cdot \partial\Phi_{\text{eff}}/\partial\omega$. Since $S_{\text{eff}}$ is independent of $\omega$, we have $\partial\Phi_{\text{eff}}/\partial\omega = \sum_j(N - d_X^{(j)}) = N_{\text{eff}}$, giving $\partial\Pr/\partial\omega = -N_{\text{eff}}\sin(2\Phi_{\text{eff}})$ at the operating point. Following the derivation in Proposition~\ref{prop:ghz-qfi}, the Fisher information at the optimal point $\Phi_{\text{eff}} = \pi/4$ (where $\Pr = 1/2$) is:
\begin{equation}
F(\omega) = \frac{(\partial_\omega \Pr)^2}{\Pr(1-\Pr)} = 4N_{\text{eff}}^2.
\end{equation}
In the absence of X-errors, $N_{\text{eff}} = 3N = 3(2L+1)$, giving $F_Q = 4 \cdot 9(2L+1)^2 = 36(2L+1)^2$. The X-error correction from Theorem~\ref{thm:bitflip} item~3 contributes factor $(1 - O(\max_k|\gamma_k/\omega|^2))$, while noise marginalization from Lemma~\ref{lem:ghz-noise-marginalization} contributes factor $e^{-4N_{\text{eff}}\sigma_\epsilon^2}$, which establishes item~4.

Finally, by Proposition~\ref{prop:error-suppression}, undetected X-errors in a single $(2L+1)$-qubit repetition code require at least $L+1$ errors, occurring with probability $O(p_X^{L+1})$. Since the three blocks operate independently for X-error detection, the total undetected error probability remains $O(p_X^{L+1})$ (dominated by the single-block failure probability), which establishes item~5.
\end{proof}

Now we are ready to put these components together and show how the combined protocol can achieve Heisenberg-limited scaling in the estimate of the frequency using both QSP and QEC.  We state this result in the following corollary.

\begin{corollary}[Resource Scaling for Combined Protocol]
\label{cor:combined-resources}
The combined protocol (Algorithm~\ref{alg:combined}) achieves precision $|\hat{\omega} - \omega| \leq \epsilon$ with probability $\geq 1 - \delta$ using total resources:
\begin{enumerate}
\item Physical qubits per measurement: $N_{\text{total}} = 3(2L+1)\in O(1/\epsilon))$
\item Total measurements: $M_{\text{total}} = O(\log(1/\epsilon) \cdot \log\log(1/\epsilon) \cdot \log(1/\delta))$ in the noiseless regime
\item Total qubit-experiments: $Q_{\text{total}} = O(\log^2(1/\epsilon) \cdot \log\log(1/\epsilon) \cdot \log(1/\delta))$
\end{enumerate}
\end{corollary}

\begin{proof}
By Theorem~\ref{thm:main} item 3 and following the resource analysis in Theorem~\ref{thm:convergence}, the binary search requires $T = O(\log(1/\epsilon))$ iterations. Each iteration uses $M = O(\log(T/\delta))$ measurements in the Heisenberg regime where $D_{\text{KL}} = \Theta(1)$. Total measurements are $M_{\text{total}} = T \times M = O(\log(1/\epsilon) \cdot \log\log(1/\epsilon) \cdot \log(1/\delta))$, following the same calculation as Theorem~\ref{thm:convergence} item 3. Multiplying by $N_{\text{total}} = 3(2L+1)$ qubits per measurement yields the total qubit-experiments. Setting $\epsilon = c/(3(2L+1))$ for constant $c$ makes $T = O(1)$, achieving Heisenberg-limited precision with resources $Q_{\text{total}} = O(N_{\text{total}} \log(1/\delta))$.
\end{proof}

%\begin{remark}[Achievability and Directions for Code Optimization]
Theorem~\ref{thm:main} and Corollary~\ref{cor:combined-resources} establish that Heisenberg-limited scaling is \emph{achievable} under combined transverse and longitudinal noise; the concatenated repetition code construction is not claimed to be the most resource-efficient encoding for this task. Optimal QEC codes designed specifically for quantum metrology~\cite{zhou2018achieving, zhou2018optimal, kessler2014quantum} could tighten the constant prefactors or extend the tolerable noise regime beyond the independent-block construction used here, while ancilla-free error correction codes~\cite{layden2019ancilla} may reduce the qubit overhead by eliminating dedicated ancilla qubits in the syndrome extraction circuit. More broadly, multi-parameter optimization frameworks~\cite{gorecki2020optimal} and QEC strategies that retain Heisenberg scaling under noise~\cite{dur2014improved} could identify codes better matched to the joint transverse-plus-longitudinal noise structure, potentially allowing beneficial inter-block stabilizer checks that exploit correlated error patterns. The results of this section therefore serve as an achievability bound; optimized code design tailored to the sensing Hamiltonian remains a natural direction for improving resource efficiency.
%\end{remark}

\section{Numerical Simulations}
\label{sec:numerical}

The theoretical analysis of the preceding sections establishes Heisenberg scaling in the asymptotic limit and proves that error-detection protocols achieve it under noise satisfying the device quality condition. This section characterizes finite-size behavior through three complementary numerical experiments. First, Bayesian phase estimation with bare GHZ probes under depolarizing noise provides a baseline, confirming the polylogarithmic overhead of Theorem~\ref{thm:resource} and quantifying how noise degrades scaling toward the SQL. Second, simulations of the bit-flip repetition code with Hamiltonian noise constitute the primary numerical contribution, showing that syndrome post-selection achieves near-Heisenberg scaling at noise levels where bare probes have degraded toward the SQL. Third, a combined protocol simulation validates Theorem~\ref{thm:main} under joint transverse noise and longitudinal inhomogeneities. Throughout this section, $\gamma$ denotes the per-qubit depolarizing rate (or the Hamiltonian transverse field strength, depending on the noise model); the correspondence between the two parameterizations is detailed in the discussion following the results tables.

\subsection{Methods}

We simulate Bayesian phase estimation for an unknown signal phase $\omega = 0.3$ (fixed for all simulations). The simulation implements a Bayesian variant of Algorithm~\ref{alg:info}, replacing the maximum likelihood step with full posterior computation over a discretized phase grid; the two approaches are asymptotically equivalent but the Bayesian formulation provides richer finite-sample information and natural uncertainty quantification. For each target precision $\epsilon$, the protocol sets a qubit budget $N = \lfloor 1/\epsilon \rfloor$ (bounded by $10 \leq N \leq 16{,}384$), discretizes the phase space into $2^m$ uniformly-spaced grid points $\omega_j \in [0, 2\pi)$ where $m$ is chosen per protocol to satisfy the Nyquist condition ($m = 16$ for bare GHZ, up to $m = 20$ for the combined protocol), initializes a uniform prior, and performs sequential Bayesian updates with log-space posterior computation for numerical stability (cf.\ Algorithm~\ref{alg:info}). The experiment budget is $K = 10{,}000$ for the bare-GHZ simulation and $K = 50{,}000$ for the error-detected and combined protocols, which require a larger budget to accommodate post-selection rejection overhead. The protocol terminates early when the circular phase error $|\hat{\omega} - \omega|_{\mathrm{circ}}$ falls below $1.2\epsilon$, consistent with the $|\hat{\phi} - \phi| \leq \epsilon$ criterion of Theorem~\ref{thm:resource}.

For the bare-GHZ simulation, each experiment selects a random probe size $n \in \{1, \ldots, N\}$ uniformly. The resulting parity outcome has likelihood
\begin{equation}
\label{eq:ghz-likelihood-sim}
P(+1 \mid \omega, n, \{\gamma_k\}) = \frac{1 + V_n \cos(2n\omega)}{2}, \qquad V_n = \prod_{k=1}^{n}(1 - \gamma_k),
\end{equation}
where $V_n$ is the depolarizing visibility and the factor $2n$ reflects GHZ phase amplification (Proposition~\ref{prop:cramer-rao}). Sampling $n$ uniformly from $\{1,\ldots,N\}$ probes the signal at multiple frequencies, resolving the periodicity ambiguity inherent to $\cos(2n\omega)$ without requiring the adaptive binary search of Algorithm~\ref{alg:bitflip}. No separate random measurement basis is needed: the full Bayesian posterior update extracts all available phase information from each probe size, unlike frequentist estimators that require randomized bases for worst-case guarantees.

Noise is modeled as a depolarizing channel with mean error rate $\gamma$ per qubit. For heterogeneous configurations, individual rates $\gamma_k \sim \mathcal{N}(\gamma, (\gamma h)^2)$ are clipped to $[0, 0.99]$, following the model of Proposition~\ref{prop:heterogeneous}.

Total resource cost is $T = \sum_{i=1}^{K_{\text{conv}}} n_i$, the cumulative number of qubits consumed across all experiments until convergence (corresponding to $N_{\text{total}}$ in Theorem~\ref{thm:resource}), where $K_{\text{conv}} \leq K$. This qubit-experiment count is the natural resource metric for comparing against the $T = O(\log(1/\epsilon)/\epsilon)$ scaling of Theorem~\ref{thm:resource}. Each configuration is evaluated over 60 logarithmically-spaced precision targets $\epsilon \in [10^{-4}, 10^{-1}]$ across 40 independent random seeds (seeds 2--41). For each seed and precision target, a deterministic per-$\epsilon$ random number generator is derived from the seed to ensure reproducibility. Convergence is checked every 100 experiments (every 100 accepted rounds for post-selection, every 10 updates for full-likelihood).

To characterize the error-detection protocols of Section~\ref{sec:fixed-omega}, we simulate Bayesian phase estimation with the $[[2L{+}1, 1]]$ bit-flip repetition code following Algorithm~\ref{alg:bitflip}. Each experiment applies the Hamiltonian noise model of Theorem~\ref{thm:bitflip} with evolution time multiplier $M$ sampled uniformly from $\{1, \ldots, M_{\max}\}$, where $M_{\max} = \lfloor 1/\epsilon \rfloor / N$ (bounded by $M_{\max} \leq 2^{m-2}/N$ to avoid phase aliasing). In \emph{post-selection mode}, runs with any detected error ($d > 0$) are discarded; the post-selected likelihood and effective phase $\phi_{\mathrm{eff}}$ are computed as specified in Theorem~\ref{thm:bitflip}. In \emph{full-likelihood mode}, all rounds contribute to the Bayesian update with syndrome-adjusted likelihoods that weight each outcome by its probability conditioned on the observed syndrome, following Remark~\ref{rem:rejection-vs-likelihood}. Full-likelihood mode extracts more information per experiment at the cost of computing per-syndrome likelihood functions. Total resource cost per experiment is $N \times M$ qubit-time units, with $M$ playing the same role as probe size $n$: varying $M$ probes the signal at multiple frequencies, resolving phase ambiguity. We test $L \in \{1, 2, 3\}$ ($N = 3, 5, 7$ qubits) at $\gamma \in \{1\%, 5\%, 10\%\}$, with 40 seeds per configuration in both inference modes.

To validate the combined protocol of Theorem~\ref{thm:main}, we extend the error-detected simulation to the concatenated code: three independent blocks of $N_{\mathrm{code}} = 2L{+}1$ repetition code qubits in a logical GHZ configuration ($N_{\mathrm{total}} = 3 N_{\mathrm{code}}$ physical qubits). Per-qubit longitudinal inhomogeneities $\epsilon_k \sim \mathcal{N}(0, \sigma_\epsilon^2)$ are drawn once per device realization and held fixed across all precision targets. Post-selection requires trivial syndrome in all three blocks; the accepted-shot likelihood and effective phase follow Theorem~\ref{thm:main}. Both inference modes are tested as above. We test $L \in \{1, 2\}$ at $\gamma \in \{1\%, 5\%, 10\%\}$ and $\sigma_\epsilon \in \{0.01, 0.05\}$, with 40 seeds and up to $K = 50{,}000$ experiments per precision target.

\subsection{Results}

\emph{Noise model note.} The bare-GHZ simulations below use a depolarizing channel (Equation~\eqref{eq:ghz-likelihood-sim}), while the error-detected and combined simulations implement the Hamiltonian noise model of Theorem~\ref{thm:bitflip}. Cross-model comparison of scaling exponents is quantitatively meaningful only at low noise ($\gamma = 1\%$); at higher noise levels, the two models are qualitatively but not quantitatively equivalent, and comparisons should be understood as demonstrating that error detection recovers near-Heisenberg scaling under both models independently. A detailed correspondence between the noise parameters appears after the results tables.

All GHZ configurations exhibit power-law scaling $T \propto \epsilon^{-\alpha}$ across three orders of magnitude in target precision. The scaling exponent $\alpha$ is obtained by fitting $\log T = -\alpha \log \epsilon + c$ via ordinary least squares (OLS) on converged data points for each of 40 independent seeds, then computing the mean and standard error of the mean (SEM) across seeds. The reported SEM captures seed-to-seed variability but not systematic uncertainty from the fitting procedure; weighted least squares (WLS), which gives more weight to smaller $\epsilon$ where data density is lower, can shift fitted $\alpha$ values by up to $0.1\text{--}0.3$ depending on the configuration. We report OLS throughout for simplicity and uniformity, noting that this systematic uncertainty should be considered alongside the reported SEM when interpreting differences between configurations. The bare-GHZ results are summarized in Table~\ref{tab:scaling-results}.

\begin{table}[t]
\centering
\caption{Scaling exponents for bare GHZ Bayesian phase estimation with depolarizing noise. The exponent $\alpha$ is obtained by fitting $\log T = -\alpha \log \epsilon + c$ via OLS on converged data, then averaging across 40 independent seeds (SEM shown). Converged (\%) is the fraction of the 60 precision targets ($\epsilon \in [10^{-4}, 10^{-1}]$) that reached the convergence criterion $|\hat{\omega} - \omega|_{\mathrm{circ}} < 1.2\epsilon$. One heterogeneous configuration ($h = 0.3$, i.e., qubit rates drawn as $\gamma_k \sim \mathcal{N}(\gamma, (0.3\gamma)^2)$) is included for comparison; remaining configurations are homogeneous.}
\label{tab:scaling-results}
\small
\begin{tabular*}{\columnwidth}{@{\extracolsep{\fill}}lccc@{}}
\toprule
\textbf{Configuration} & $\boldsymbol{\alpha}$ \textbf{(mean $\pm$ SEM)} & \textbf{Converged (\%)} & \textbf{Seeds} \\
\midrule
Noiseless & $1.15 \pm 0.02$ & 99.4 & 40 \\
$\gamma = 1\%$ & $1.66 \pm 0.02$ & 92.4 & 40 \\
$\gamma = 5\%$ & $1.87 \pm 0.02$ & 76.8 & 40 \\
$\gamma = 5\% \pm 30\%$ (het) & $1.89 \pm 0.02$ & 76.4 & 40 \\
$\gamma = 10\%$ & $1.95 \pm 0.03$ & 67.5 & 40 \\
$\gamma = 15\%$ & $1.97 \pm 0.03$ & 62.6 & 40 \\
$\gamma = 20\%$ & $1.90 \pm 0.03$ & 59.5 & 40 \\
\bottomrule
\end{tabular*}
\end{table}

\begin{table}[t]
\centering
\caption{Scaling exponents for error-detected Bayesian phase estimation using the $[[2L{+}1, 1]]$ bit-flip repetition code (Theorem~\ref{thm:bitflip}). Two inference modes are compared: \emph{post-selection} discards rounds with detected errors ($d > 0$), while \emph{full-likelihood} retains all rounds with syndrome-adjusted likelihoods (Remark~\ref{rem:rejection-vs-likelihood}). Converged (\%) is the convergence rate (see Table~\ref{tab:scaling-results}). Accept.\ (\%) is the fraction of experiments passing post-selection (not applicable to full-likelihood). Each configuration uses 40 independent seeds with 60 precision targets $\epsilon \in [10^{-4}, 10^{-1}]$ and up to 50,000 experiments per target.}
\label{tab:error-detected-results}
\small
\begin{tabular*}{\columnwidth}{@{\extracolsep{\fill}}lccccc@{}}
\toprule
& \multicolumn{2}{c}{\textbf{Post-selection}} & \multicolumn{2}{c}{\textbf{Full-likelihood}} \\
\cmidrule(lr){2-3} \cmidrule(lr){4-5}
\textbf{Configuration} & $\boldsymbol{\alpha}$ & \textbf{Converged (\%)} & $\boldsymbol{\alpha}$ & \textbf{Converged (\%)} & \textbf{Accept.\ (\%)} \\
\midrule
$L{=}1$, $\gamma = 1\%$ & $1.12 \pm 0.02$ & 59 & $1.13 \pm 0.02$ & 63 & $>$99 \\
$L{=}1$, $\gamma = 5\%$ & $0.98 \pm 0.02$ & 77 & $0.94 \pm 0.01$ & 82 & 96 \\
$L{=}1$, $\gamma = 10\%$ & $1.06 \pm 0.01$ & 90 & $1.04 \pm 0.01$ & 100 & 86 \\
\midrule
$L{=}2$, $\gamma = 1\%$ & $1.23 \pm 0.02$ & 61 & $1.19 \pm 0.02$ & 64 & $>$99 \\
$L{=}2$, $\gamma = 5\%$ & $1.08 \pm 0.01$ & 76 & $0.99 \pm 0.01$ & 79 & 94 \\
$L{=}2$, $\gamma = 10\%$ & $1.02 \pm 0.01$ & 88 & $1.00 \pm 0.01$ & 91 & 79 \\
\midrule
$L{=}3$, $\gamma = 1\%$ & $1.11 \pm 0.02$ & 69 & $1.09 \pm 0.01$ & 95 & $>$99 \\
$L{=}3$, $\gamma = 5\%$ & $1.07 \pm 0.01$ & 100 & $1.03 \pm 0.01$ & 100 & 91 \\
$L{=}3$, $\gamma = 10\%$ & $1.00 \pm 0.004$ & 100 & $1.01 \pm 0.003$ & 100 & 72 \\
\bottomrule
\end{tabular*}
\end{table}

The noiseless case achieves $\alpha = 1.15 \pm 0.02$, consistent with Theorem~\ref{thm:resource}. The gap from the theoretical $\alpha = 1$ reflects the polylogarithmic overhead $O(\log(1/\epsilon))$: restricting the fit to $\epsilon < 10^{-2}$ yields $\alpha = 0.98 \pm 0.02$, consistent with $\alpha = 1$ and confirming that the overhead vanishes as the fit range excludes finite-size effects. Depolarizing noise degrades the exponent toward the SQL (Table~\ref{tab:scaling-results}), with $\alpha$ approaching 2 for $\gamma \geq 5\%$ due to exponential visibility decay $V_n = (1-\gamma)^n$. At $\gamma = 20\%$, the fitted exponent ($\alpha = 1.90$) is slightly lower than at $\gamma = 10\%$ ($\alpha = 1.95$), despite the expectation that higher noise should push $\alpha$ closer to the SQL. This non-monotonicity is an artifact of survivorship bias: at $\gamma = 20\%$, only 59.5\% of precision targets converge within the experiment budget, and the converging runs are disproportionately those at larger $\epsilon$ (lower precision), where scaling exponents are naturally closer to Heisenberg. The fitted $\alpha$ at $\gamma \geq 15\%$ should therefore be interpreted with caution, as the selective convergence biases the sample toward easier targets. Noise heterogeneity has a small effect ($<5\%$ change in $\alpha$ at matched $\gamma$), consistent with Proposition~\ref{prop:heterogeneous}. These bare-GHZ results establish the baseline against which error detection is measured.

\begin{figure}[t]
    \centering
    \includegraphics[width=0.95\textwidth]{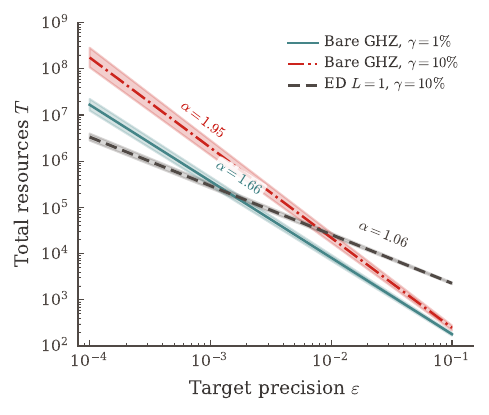}
    \caption{Error-detected protocol achieves near-Heisenberg scaling where bare GHZ degrades toward the SQL. Resource scaling $T$ (total qubits consumed) vs.\ target precision $\epsilon$ for bare GHZ (solid) at depolarizing noise rates $\gamma = 1\%$ and $10\%$, alongside error-detected post-selection (dashed, $L{=}1$, $N{=}3$ qubits) at $\gamma = 10\%$. The scaling exponent $\alpha$ is defined by $T \propto \epsilon^{-\alpha}$, where $\alpha = 1$ is the Heisenberg limit and $\alpha = 2$ is the SQL. At $\gamma = 10\%$, bare GHZ scaling approaches the SQL ($\alpha = 1.95$), while post-selection on the trivial syndrome reduces the exponent to $\alpha = 1.06$, recovering near-Heisenberg scaling. At $\gamma = 1\%$, bare GHZ achieves $\alpha = 1.66$; all noise levels tested show similar improvement under error detection (Table~\ref{tab:error-detected-results}). Lines show mean power-law fits; shaded bands indicate 95\% confidence intervals on the scaling exponent $\alpha$ (slope uncertainty only) over 40 independent seeds.}
    \label{fig:error-detection}
\end{figure}

The error-detected simulations demonstrate the central result of this paper (Table~\ref{tab:error-detected-results}, Figures~\ref{fig:error-detection} and~\ref{fig:inference-comparison}): syndrome-based inference achieves near-Heisenberg scaling at noise levels where bare GHZ probes have degraded toward the SQL. In post-selection mode, the 3-qubit code ($L = 1$) reduces the exponent from $\alpha = 1.66$ (bare GHZ, $\gamma = 1\%$) to $\alpha = 1.12 \pm 0.02$, and from $\alpha = 1.87\text{--}1.95$ (bare GHZ, $\gamma = 5\text{--}10\%$) to $\alpha = 0.98\text{--}1.06$. The 7-qubit code ($L = 3$) at $\gamma = 10\%$ reaches $\alpha = 1.00 \pm 0.004$ with 100\% convergence, the tightest scaling observed in any post-selection configuration.

Full-likelihood inference improves both scaling and convergence over post-selection by extracting phase information from every round, including those with detected errors. At $L{=}1$, $\gamma{=}10\%$, full-likelihood achieves $\alpha = 1.04$ (vs.\ $1.06$ post-selection) with 100\% convergence (vs.\ 90\%; Figure~\ref{fig:inference-comparison}). The improvement is most pronounced at moderate noise: at $L{=}2$, $\gamma{=}5\%$, full-likelihood reaches $\alpha = 0.99$ compared to $1.08$ under post-selection. At low noise ($\gamma = 1\%$), the two modes yield similar exponents, as most rounds pass post-selection and full-likelihood gains little additional information.

Several full-likelihood configurations yield fitted $\alpha$ slightly below 1 (e.g., $\alpha = 0.94 \pm 0.01$ at $L{=}1$, $\gamma = 5\%$). Since the Heisenberg limit $\alpha = 1$ is a fundamental lower bound, these values reflect systematic bias in the OLS fitting procedure rather than super-Heisenberg scaling: the reported standard errors capture seed-to-seed variance but not the systematic uncertainty inherent in fitting a power law over the finite precision range $\epsilon \in [10^{-4}, 10^{-1}]$. The bias is most pronounced in full-likelihood mode at moderate noise, where per-round information content varies across the $\epsilon$ range, introducing curvature in the $\log T$ vs.\ $\log \epsilon$ relationship. All fitted exponents remain within $0.10$ of the Heisenberg limit and should be interpreted as consistent with $\alpha = 1$.

Convergence rates for the error-detected protocol increase with noise in both modes (e.g., 59\% at $\gamma = 1\%$ to 90\% at $\gamma = 10\%$ for $L = 1$ post-selection), the opposite of the bare-GHZ trend (Table~\ref{tab:scaling-results}). This counter-intuitive pattern reflects the interplay between phase amplification and the finite experiment budget. At low noise, the effective phase $\phi_{\mathrm{eff}} \propto N_{\mathrm{code}} \cdot M \cdot \omega$ can reach thousands of radians at the smallest $\epsilon$ targets (where $M_{\max} \gg 1$), creating severe multi-modality in the Bayesian posterior that requires many experiments to resolve. At higher noise, the reduced acceptance rate consumes more of the experiment budget per posterior update, but the resulting post-selected likelihood retains high signal quality; the net effect is that the algorithm converges on more $\epsilon$ targets within the $K = 50{,}000$ budget. A detailed analysis of this convergence behavior is left to future work. Even the minimal 3-qubit code ($L = 1$) achieves $\alpha \leq 1.12$ at $\gamma = 1\text{--}10\%$ with 86--100\% acceptance rate in these simulations.

\begin{figure}[t]
    \centering
    \includegraphics[width=0.95\textwidth]{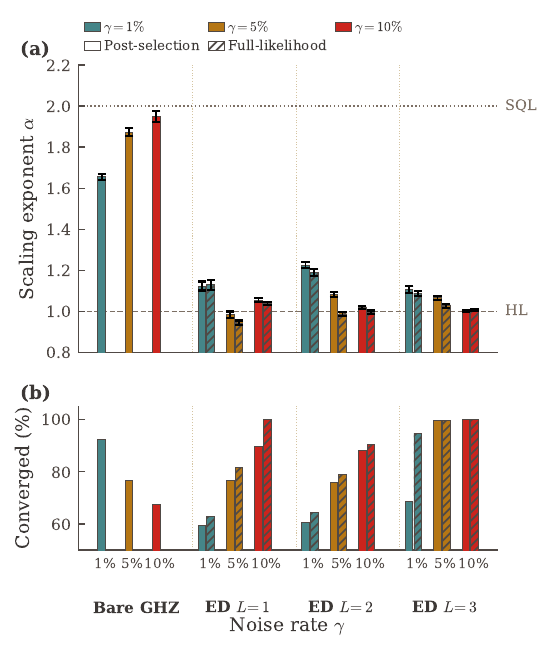}
    \caption{(a)~Scaling exponent $\alpha$ (where $T \propto \epsilon^{-\alpha}$) and (b)~convergence rate (fraction of precision targets reaching convergence within the experiment budget) across protocols and inference modes. Bar color encodes noise rate ($\gamma = 1\%, 5\%, 10\%$). Solid bars denote post-selection (retaining only error-free rounds); hatched bars denote full-likelihood (using all rounds with syndrome-adjusted likelihoods). Each group shows bare GHZ (no error detection, single bar per noise level) or the error-detected protocol at code distances $L = 1, 2, 3$. Both inference modes achieve near-Heisenberg scaling at all tested noise levels. Full-likelihood consistently improves both scaling exponents and convergence rates (e.g., $\alpha = 1.04$ at 100\% convergence vs.\ $\alpha = 1.06$ at 90\% for post-selection at $L{=}1$, $\gamma{=}10\%$) by utilizing all measurement data including error-flagged rounds. Fitted exponents slightly below the Heisenberg limit ($\alpha < 1$) for some full-likelihood configurations are finite-range OLS fitting artifacts, not violations of the theoretical bound. Error bars represent standard error of the mean over 40 seeds.}
    \label{fig:inference-comparison}
\end{figure}

The combined protocol simulations validate Theorem~\ref{thm:main} under joint transverse noise and longitudinal inhomogeneities (Table~\ref{tab:combined-results}, Figure~\ref{fig:combined-protocol}). In post-selection mode with $L = 1$ ($N_{\mathrm{total}} = 9$) and $\sigma_\epsilon = 0.01$, the combined protocol achieves $\alpha = 1.12\text{--}1.13$ across $\gamma = 1\text{--}10\%$. Acceptance rates match the theoretical prediction: at $\gamma = 5\%$, the combined acceptance of $88.6\%$ agrees with the single-block rate cubed ($96.0\%^3 \approx 88.5\%$), confirming that the three blocks operate independently under post-selection. Full-likelihood mode improves both scaling and convergence: at $L = 1$, $\gamma = 5\%$, full-likelihood reaches $\alpha = 0.92 \pm 0.01$ with 100\% convergence (vs.\ $\alpha = 1.12$ at 93\% for post-selection). The $L = 2$ code ($N_{\mathrm{total}} = 15$) achieves $\alpha = 1.13\text{--}1.19$ in post-selection and $\alpha = 0.96\text{--}1.09$ in full-likelihood across all tested noise levels, with full-likelihood at $\gamma = 5\%$ reaching $\alpha = 0.96 \pm 0.004$ with 100\% convergence. Increasing $\sigma_\epsilon$ from $0.01$ to $0.05$ has a minor effect on most configurations ($<0.05$ change in $\alpha$ at $L = 1$), though the effect is larger at $L = 2$, $\gamma = 1\%$ where low convergence rates ($\sim$62\%) increase the variance of the fitted exponent.

\begin{table}[t]
\centering
\caption{Scaling exponents for the combined protocol (Theorem~\ref{thm:main}): logical GHZ state across three $[[2L{+}1, 1]]$ repetition code blocks, under joint transverse noise ($\gamma$) and longitudinal inhomogeneities ($\sigma_\epsilon$). Both inference modes and column definitions follow Table~\ref{tab:error-detected-results}. Each configuration uses 40 independent seeds with 60 precision targets $\epsilon \in [10^{-4}, 10^{-1}]$ and up to 50,000 experiments per target. Results for $\sigma_\epsilon = 0.05$ are qualitatively similar ($<0.05$ change in $\alpha$ at $L = 1$) and discussed in the text.}
\label{tab:combined-results}
\small
\begin{tabular*}{\columnwidth}{@{\extracolsep{\fill}}lccccc@{}}
\toprule
& \multicolumn{2}{c}{\textbf{Post-selection}} & \multicolumn{2}{c}{\textbf{Full-likelihood}} \\
\cmidrule(lr){2-3} \cmidrule(lr){4-5}
\textbf{Configuration} & $\boldsymbol{\alpha}$ & \textbf{Converged (\%)} & $\boldsymbol{\alpha}$ & \textbf{Converged (\%)} & \textbf{Accept.\ (\%)} \\
\midrule
$L{=}1$, $\gamma = 1\%$, $\sigma_\epsilon = 0.01$ & $1.12 \pm 0.02$ & 70 & $1.05 \pm 0.01$ & 96 & 99.5 \\
$L{=}1$, $\gamma = 5\%$, $\sigma_\epsilon = 0.01$ & $1.12 \pm 0.01$ & 93 & $0.92 \pm 0.01$ & 100 & 88.6 \\
$L{=}1$, $\gamma = 10\%$, $\sigma_\epsilon = 0.01$ & $1.13 \pm 0.01$ & 100 & $0.98 \pm 0.002$ & 100 & 64.1 \\
\midrule
$L{=}2$, $\gamma = 1\%$, $\sigma_\epsilon = 0.01$ & $1.19 \pm 0.02$ & 62 & $1.09 \pm 0.02$ & 87 & 99.2 \\
$L{=}2$, $\gamma = 5\%$, $\sigma_\epsilon = 0.01$ & $1.15 \pm 0.01$ & 88 & $0.96 \pm 0.004$ & 100 & 81.8 \\
$L{=}2$, $\gamma = 10\%$, $\sigma_\epsilon = 0.01$ & $1.13 \pm 0.01$ & 99 & $0.99 \pm 0.001$ & 100 & 48.8 \\
\bottomrule
\end{tabular*}
\end{table}

\begin{figure}[t]
    \centering
    \includegraphics[width=0.95\textwidth]{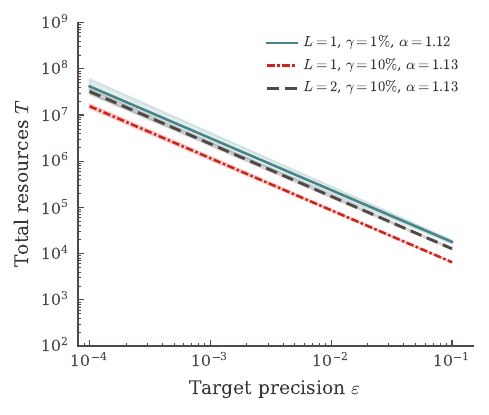}
    \caption{Combined protocol under joint transverse noise ($\gamma$) and longitudinal inhomogeneities ($\sigma_\epsilon = 0.01$). Resource scaling $T$ (total qubits consumed) vs.\ target precision $\epsilon$ for the concatenated code, which distributes a logical GHZ state across three $[[2L{+}1, 1]]$ repetition code blocks to simultaneously suppress both noise types. Results shown for $L = 1$ ($N_{\mathrm{total}} = 9$ qubits, solid) at $\gamma = 1\%$ and $10\%$, and $L = 2$ ($N_{\mathrm{total}} = 15$, dashed) at $\gamma = 10\%$. Post-selection retains only rounds where no error is detected in any of the three blocks. All configurations maintain near-Heisenberg scaling ($\alpha \leq 1.19$, Table~\ref{tab:combined-results}) even under joint noise, with the $L = 2$ code adding modest resource overhead at the same scaling exponent. Lines show mean power-law fits; shaded bands indicate 95\% confidence intervals on the scaling exponent $\alpha$ (slope uncertainty only) over 40 independent seeds.}
    \label{fig:combined-protocol}
\end{figure}

The bare-GHZ and error-detected simulations use different noise models: the bare-GHZ simulation uses a depolarizing channel (Equation~\eqref{eq:ghz-likelihood-sim}), while the error-detected and combined simulations implement the Hamiltonian noise model of Theorem~\ref{thm:bitflip}. In the weak-noise regime $\gamma_k/\omega = o(1)$, the Hamiltonian model yields per-qubit bit-flip probability $p_k = \sin^2(\omega)\gamma_k^2/\omega^2 + O(\gamma_k^4/\omega^4)$, so the depolarizing rate $\gamma$ in Table~\ref{tab:scaling-results} corresponds to Hamiltonian transverse field $\gamma_{\mathrm{Ham}} \approx \gamma^{1/2}\omega/\sin(\omega)$. At $\omega = 0.3$, this gives $\gamma_{\mathrm{Ham}} \approx 0.10$ for $\gamma = 1\%$ (where $\gamma_{\mathrm{Ham}}/\omega \approx 0.33$ and the weak-noise approximation is reasonable), $\gamma_{\mathrm{Ham}} \approx 0.23$ for $\gamma = 5\%$ ($\gamma_{\mathrm{Ham}}/\omega \approx 0.77$, marginal), and $\gamma_{\mathrm{Ham}} \approx 0.32$ for $\gamma = 10\%$ ($\gamma_{\mathrm{Ham}}/\omega > 1$, outside the weak-noise regime). Consequently, the cross-model comparison of scaling exponents between Tables~\ref{tab:scaling-results} and~\ref{tab:error-detected-results} is quantitatively meaningful only at $\gamma = 1\%$; at higher noise levels, the two noise models are qualitatively but not quantitatively equivalent, and the comparison should be understood as demonstrating that error detection recovers near-Heisenberg scaling under both models independently.

\section{Limitations and Extensions}
\label{sec:discussion}

While our protocols successfully handle transverse X-fields and longitudinal Z-field inhomogeneities, the case of Y-field components ($\chi_k \neq 0$) presents fundamental challenges that merit discussion. The difficulty arises from the algebraic structure of Pauli matrices: since $Y = iXZ$, a Y-error creates the same syndrome pattern as simultaneous X and Z errors. This syndrome ambiguity makes it impossible to distinguish between a single Y-error on qubit $k$, simultaneous X and Z errors on the same qubit, or a combination of X-error with Z-field inhomogeneity.

The mathematical origin of this challenge lies in the non-commutativity of Pauli operators. Unlike X and Z rotations which can be separated through appropriate syndrome measurements, Y-rotations fundamentally couple amplitude and phase quadratures: $e^{-i\chi_k Y_k} = e^{-i\chi_k(iXZ)} \neq e^{-i\chi_k X_k} e^{-i\chi_k Z_k}$. This coupling cannot be cleanly resolved through stabilizer measurements alone, as any syndrome pattern consistent with Y-errors admits multiple interpretations that lead to different phase corrections. This limitation connects to the Hamiltonian-not-in-Lindblad-span (HNLS) condition identified by Zhou et al.~\cite{zhou2018achieving}: when all three Pauli error channels ($X$, $Y$, $Z$) are present, the signal Hamiltonian $H_s \propto Z$ lies within the span of the effective Lindblad operators, the regime where Heisenberg scaling is fundamentally unachievable with standard QEC under Markovian noise. Our protocols circumvent this no-go by restricting to the X-noise subspace where HNLS is satisfied, i.e., where the signal generator $Z$ is linearly independent of the noise generators.

Several promising approaches could partially address this limitation. Decoherence-free subspaces offer one avenue by encoding in subspaces naturally immune to Y-rotations \cite{lidar2003decoherence,zhang2023quantum}. Dynamical decoupling provides another strategy, using carefully designed pulse sequences to average out Y-components while preserving the signal \cite{souza2011robust}. Recent theoretical advances, such as platonic sequences based on discrete symmetry groups \cite{zhen2025platonic}, show particular promise for complex multi-axis noise. Additionally, higher-distance quantum error correction codes \cite{acharya2024quantum} and continuous monitoring protocols \cite{livingston2022experimental} may offer paths toward more robust sensing in the presence of Y-fields.

A separate limitation is the device quality condition $N\sigma_\epsilon^2 = o(1)$, which requires noise to decrease with system size for Heisenberg scaling to be maintained. This condition distinguishes our results from the no-go theorem of Demkowicz-Dobrza\'{n}ski et al.~\cite{demkowicz2012elusive}, which establishes SQL as the fundamental limit under fixed-rate Markovian decoherence. Our protocols achieve Heisenberg scaling precisely because the noise structure permits it: in calibrated quantum devices, the relevant noise parameters (field inhomogeneity, transverse field strength) can be reduced through improved hardware, making $N\sigma_\epsilon^2 = o(1)$ a realistic engineering constraint rather than a fundamental obstruction. Notably, tensor product (SQL) protocols are naturally robust to inhomogeneities, with only a subdominant $\sigma_\epsilon^2/N$ correction to the estimation variance (Appendix~\ref{app:direct-parity}), whereas Heisenberg-scaling protocols require the more stringent $\sigma_\epsilon = o(1/\sqrt{N})$ condition.

On the practical side, the device requirements for our protocols are modest: GHZ state preparation ($O(N)$ CNOT gates), single-qubit $R_Z$ rotations (native on most platforms), and parity measurements (via stabilizer circuits or destructive measurement with classical post-processing). No quantum Fourier transforms or multi-qubit controlled operations beyond GHZ preparation are needed. The noise condition $N\sigma_\epsilon^2 = o(1)$ translates to concrete thresholds: for $N = 50$ qubits, one needs $\sigma_\epsilon \lesssim 0.02$, i.e., field inhomogeneities of roughly 2\% relative to the signal. Note that $\sigma_\epsilon$ parameterizes qubit-to-qubit variation in the sensed field, not gate error rates; relating this condition to specific hardware platforms requires a device-level noise model beyond the scope of the present analysis.

Several additional directions merit investigation. Adaptive shot allocation (varying $M$ per iteration rather than using a fixed budget) could reduce total resource consumption. The tangent-cotangent amplitude structure (Proposition~\ref{prop:tangent-cotangent-full}, Appendix~\ref{app:binary-search-proofs}) suggests that different syndrome subsets carry complementary information about individual qubit phases, potentially enabling multi-parameter estimation with Heisenberg scaling. Extending the analysis to time-varying signals $\omega(t)$ with bounded derivatives, and refining the asymptotic bounds for small system sizes ($N < 10$), are natural next steps toward near-term experimental demonstrations.

Despite these open questions, our results demonstrate that the Encoded QSP framework achieves metrological advantage even with restricted error models. Beyond metrology, the core principle of Encoded QSP (using syndrome measurement as a signal-processing primitive to implement nonlinear transformations on a logical qubit) may find applications in other domains where quantum signal processing intersects with error correction, including quantum simulation and quantum machine learning. Achieving full Heisenberg scaling with arbitrary Y-field components remains an important open problem at the intersection of quantum error correction and quantum signal processing.

\begin{table}[t]
\caption{Summary of protocols and baselines. In the QFI column, $N$ refers to the qubit count listed in the second column for each protocol. Exact prefactors and noise damping factors appear in the referenced theorems. ``Marg.'' indicates Z-inhomogeneity handling via Bayesian marginalization. All Heisenberg-scaling protocols require the device quality condition $N\sigma_\epsilon^2 = o(1)$ when Z-inhomogeneities are present. $^*$$M$ denotes the number of sequential signal applications per qubit; see Theorem~\ref{thm:sequential-binary-search}. The sequential approach achieves phase amplification through repeated signal application rather than entanglement.}
\label{tab:protocol-summary}
\small
\begin{tabular*}{\columnwidth}{@{\extracolsep{\fill}}lcccc@{}}
\toprule
\textbf{Protocol} & \textbf{Qubits} & \textbf{X-Suppr.} & \textbf{Z-Inhomog.} & \textbf{QFI} \\
\midrule
Bit-flip (Sec.~\ref{sec:fixed-omega}) & $2L{+}1$ & $O(p^{L+1})$ & --- & $\Theta(N^2)$ \\
Binary search (Sec.~\ref{sec:binary-search}) & $N$ & --- & Marg. & $\Theta(N^2)$ \\
Combined (Sec.~\ref{sec:combined}) & $3(2L{+}1)$ & $O(p^{L+1})$ & Marg. & $\Theta(N^2)$ \\
Sequential (Sec.~\ref{subsec:sequential-binary-search}) & $2L{+}1$ & $O(p^{L+1})$ & --- & $\Theta(NM^2)^*$ \\
\midrule
Z-syndrome (Sec.~\ref{sec:sql-barrier}) & $N$ & --- & --- & $\Theta(N)$ \\
Tensor product (App.~\ref{app:direct-parity}) & $N$ & --- & Marg. & $\Theta(N)$ \\
\bottomrule
\end{tabular*}
\end{table}

\section{Conclusion}

We have introduced encoded quantum signal processing (Encoded QSP), a framework that extends the QSP paradigm to error-detecting codes, and demonstrated its power through Heisenberg-limited quantum metrology under realistic noise. The central idea is that encoding sensor qubits into a single logical qubit via a repetition code reduces multi-qubit metrology to an effective single-qubit signal processing problem, where syndrome measurement implements nonlinear signal transformations (Definition~\ref{def:encoded-qsp}, Theorem~\ref{thm:syndrome-rotation}).

A key structural result is the SQL barrier (Theorem~\ref{thm:sql-barrier}): the simplest Encoded QSP instantiation (product-state input, single signal application) yields only SQL scaling ($F_{\mathrm{total}} = 4N$), despite per-syndrome phase amplification to $F_j = 4k^2$. The barrier arises from binomial fluctuations in the amplification factor $k = N - 2j$ at the optimal operating point. Four protocols overcome this barrier through different resources: three use entangled (GHZ) logical states with collective X-parity measurements to preserve the coherent $N$-fold phase enhancement, and a fourth (sequential binary search, Theorem~\ref{thm:sequential-binary-search}) achieves Heisenberg scaling via product states with repeated signal applications; Table~\ref{tab:protocol-summary} summarizes their key properties.

The common thread across all four protocols, and the essence of Encoded QSP, is syndrome-dependent phase extraction: error syndrome measurements provide classical side information that enables precise phase recovery despite quantum errors. Rather than applying recovery operations that would erase the signal, we extract the phase directly from syndrome-conditioned parity measurements. Combined with information-theoretic phase estimation, these protocols achieve precision $\epsilon$ with $O(1/\epsilon \cdot \operatorname{polylog}(1/\epsilon))$ total resources, establishing the Heisenberg limit $\alpha = 1$ via optimal quantum Fisher information scaling (Proposition~\ref{prop:cramer-rao}). Numerical simulations characterize the finite-size behavior: noiseless GHZ Bayesian estimation exhibits near-Heisenberg scaling consistent with $\alpha = 1$, while depolarizing noise degrades bare probes toward the SQL. Simulations of the bit-flip repetition code with Hamiltonian noise show that syndrome post-selection achieves near-Heisenberg scaling at noise levels ($\gamma = 1\text{--}10\%$) where bare probes approach the SQL, and the full concatenated protocol maintains this scaling under joint transverse noise and longitudinal inhomogeneities (Section~\ref{sec:numerical}).

As discussed in Section~\ref{sec:discussion}, these results are consistent with the no-go theorem of Demkowicz-Dobrza\'{n}ski et al.~\cite{demkowicz2012elusive}: our protocols achieve Heisenberg scaling under the device quality condition $N\sigma_\epsilon^2 = o(1)$, which falls outside the fixed-rate Markovian decoherence regime where SQL is fundamental.

The Encoded QSP framework opens several directions for future work. The broadened definition (Definition~\ref{def:encoded-qsp}) admits codes beyond the repetition code studied here; characterizing the class of achievable signal transformations for different code families, and its relationship to the polynomial universality of GQSP~\cite{motlagh2023generalized}, is a natural theoretical question. On the applications side, extending these methods to codes that handle Y-errors ($\chi_k \neq 0$), where syndrome ambiguity between signal and noise remains an open challenge, and to continuous variable systems are important next steps toward broader experimental deployment.
\begin{acknowledgments}

COM acknowledges funding from the Laboratory Directed Research and Development Program and Mathematics for Artificial Reasoning for Scientific Discovery investment at the Pacific Northwest National Laboratory, a multiprogram national laboratory operated by Battelle for the U.S. Department of Energy under Contract DE-AC05-76RLO1830. NW and COM were funded by grants from the US Department of Energy, Office of Science, National Quantum Information Science Research Centers, Co-Design Center for Quantum Advantage under contract number DE-SC0012704.  NW acknowledges funding from the NSERC Quantum Science Consortium.  NW and COM acknowledge the support of the ``Co-design for Fundamental Physics in the Fault-Tolerant Era'' workshop hosted by IQuS at the University of Washington, where much of this research was conducted.

\end{acknowledgments}

\section*{Data and Code Availability}

The simulation code and data supporting the numerical results in this paper are available at \url{https://github.com/cmortiz/heisenberg-qed-metrology}. The repository includes all simulation scripts, analysis code, and raw CSV data files for the 40-seed runs reported in Tables~\ref{tab:scaling-results}--\ref{tab:combined-results}.

\bibliographystyle{unsrt}
\bibliography{ref}

\begin{thebibliography}{10}

\bibitem{gottesman1997stabilizer}
Daniel Gottesman.
\newblock {\em Stabilizer codes and quantum error correction}.
\newblock California Institute of Technology, 1997.

\bibitem{low2017optimal}
Guang~Hao Low and Isaac~L Chuang.
\newblock Optimal hamiltonian simulation by quantum signal processing.
\newblock {\em Physical review letters}, 118(1):010501, 2017.

\bibitem{gilyen2019quantum}
Andr{\'a}s Gily{\'e}n, Yuan Su, Guang~Hao Low, and Nathan Wiebe.
\newblock Quantum singular value transformation and beyond: exponential improvements for quantum matrix arithmetics.
\newblock In {\em Proceedings of the 51st Annual ACM SIGACT Symposium on Theory of Computing}, pages 193--204, 2019.

\bibitem{giovannetti2004quantum}
Vittorio Giovannetti, Seth Lloyd, and Lorenzo Maccone.
\newblock Quantum-enhanced measurements: beating the standard quantum limit.
\newblock {\em Science}, 306(5700):1330--1336, 2004.

\bibitem{giovannetti2006quantum}
Vittorio Giovannetti, Seth Lloyd, and Lorenzo Maccone.
\newblock Quantum metrology.
\newblock {\em Physical Review Letters}, 96(1):010401, 2006.

\bibitem{kitagawa1993squeezed}
Masahiro Kitagawa and Masahito Ueda.
\newblock Squeezed spin states.
\newblock {\em Physical Review A}, 47(6):5138--5143, 1993.

\bibitem{gerry2000heisenberg}
Christopher~C Gerry.
\newblock Heisenberg-limit interferometry with four-wave mixers operating in a nonlinear regime.
\newblock {\em Physical Review A}, 61(4):043811, 2000.

\bibitem{zhou2021error}
Sisi Zhou.
\newblock {\em Error-corrected quantum metrology}.
\newblock PhD thesis, Yale University, 2021.

\bibitem{kapourniotis2019fault}
Theodoros Kapourniotis and Animesh Datta.
\newblock Fault-tolerant quantum metrology.
\newblock {\em Physical Review A}, 100(2):022335, 2019.

\bibitem{zhou2018achieving}
Sisi Zhou, Mengzhen Zhang, John Preskill, and Liang Jiang.
\newblock Achieving the heisenberg limit in quantum metrology using quantum error correction.
\newblock {\em Nature communications}, 9(1):1--11, 2018.

\bibitem{zhou2018optimal}
Sisi Zhou and Liang Jiang.
\newblock Optimal approximate quantum error correction for quantum metrology.
\newblock {\em Physical Review Research}, 2(1):013235, 2020.

\bibitem{sekatski2017quantum}
Pavel Sekatski, Michalis Skotiniotis, Jan Kolody{\'n}ski, and Wolfgang D{\"u}r.
\newblock Quantum metrology with full and fast quantum control.
\newblock {\em Quantum}, 1:27, 2017.

\bibitem{layden2019ancilla}
David Layden, Sisi Zhou, Paola Cappellaro, and Liang Jiang.
\newblock Ancilla-free quantum error correction codes for quantum metrology.
\newblock {\em Physical Review Letters}, 122(4):040502, 2019.

\bibitem{sahu2026heisenberg}
Suhas Sahu, Yifan Xu, and Sisi Zhou.
\newblock Achieving the heisenberg limit using fault-tolerant quantum error correction.
\newblock {\em arXiv preprint arXiv:2601.05457}, 2026.

\bibitem{allen2025quantum}
Richard~R. Allen, Francisco Machado, Isaac~L. Chuang, Hsin-Yuan Huang, and Soonwon Choi.
\newblock Quantum computing enhanced sensing.
\newblock {\em arXiv preprint arXiv:2501.07625}, 2025.

\bibitem{khan2025quantum}
Saeed~A. Khan, Sridhar Prabhu, Logan~G. Wright, and Peter~L. McMahon.
\newblock Quantum computational sensing using quantum signal processing, quantum neural networks, and hamiltonian engineering.
\newblock {\em arXiv preprint arXiv:2507.15845}, 2025.

\bibitem{demkowicz2012elusive}
Rafa{\l} Demkowicz-Dobrza{\'n}ski, Jan Ko{\l}ody{\'n}ski, and M{\u{a}}d{\u{a}}lin Gu{\c{t}}{\u{a}}.
\newblock Elusive heisenberg limit in quantum-enhanced metrology.
\newblock {\em Nature communications}, 3(1):1063, 2012.

\bibitem{motlagh2023generalized}
Danial Motlagh and Nathan Wiebe.
\newblock Generalized quantum signal processing.
\newblock {\em PRX Quantum}, 5:020368, 2024.

\bibitem{martyn2021grand}
John~M. Martyn, Zane~M. Rossi, Andrew~K. Tan, and Isaac~L. Chuang.
\newblock Grand unification of quantum algorithms.
\newblock {\em PRX Quantum}, 2:040203, 2021.

\bibitem{berry2014exponential}
Dominic~W Berry, Andrew~M Childs, Richard Cleve, Robin Kothari, and Rolando~D Somma.
\newblock Exponential improvement in precision for simulating sparse hamiltonians.
\newblock In {\em Proceedings of the forty-sixth annual ACM symposium on Theory of computing}, pages 283--292, 2014.

\bibitem{yoder2014fixed}
Theodore~J Yoder, Guang~Hao Low, and Isaac~L Chuang.
\newblock Fixed-point quantum search with an optimal number of queries.
\newblock {\em Physical review letters}, 113(21):210501, 2014.

\bibitem{nielsen2002quantum}
Michael~A. Nielsen and Isaac~L. Chuang.
\newblock {\em Quantum Computation and Quantum Information}.
\newblock Cambridge University Press, 10th anniversary edition, 2010.

\bibitem{dong2024optimal}
Yulong Dong, Jonathan~A. Gross, and Murphy~Yuezhen Niu.
\newblock Optimal low-depth quantum signal-processing phase estimation.
\newblock {\em Nature Communications}, 16:1504, 2025.

\bibitem{lin2026quantum}
Lin Lin and Nathan Wiebe.
\newblock {\em Quantum Algorithms for Scientific Computation}.
\newblock 2026.
\newblock Lecture notes, continuously updated.

\bibitem{svore2013faster}
Krysta~M Svore, Matthew Hastings, and Michael Freedman.
\newblock Faster phase estimation.
\newblock {\em Quantum Information and Computation}, 14:306--328, 2013.

\bibitem{kitaev1995quantum}
A~Yu Kitaev.
\newblock Quantum measurements and the abelian stabilizer problem.
\newblock {\em arXiv preprint quant-ph/9511026}, 1995.

\bibitem{granade2012robust}
Christopher~E Granade, Christopher Ferrie, Nathan Wiebe, and David~G Cory.
\newblock Robust online hamiltonian learning.
\newblock {\em New Journal of Physics}, 14(10):103013, 2012.

\bibitem{higgins2007entanglement}
Brendon~L Higgins, Dominic~W Berry, Stephen~D Bartlett, Howard~M Wiseman, and Geoff~J Pryde.
\newblock Entanglement-free heisenberg-limited phase estimation.
\newblock {\em Nature}, 450(7168):393--396, 2007.

\bibitem{wiebe2016efficient}
Nathan Wiebe and Chris Granade.
\newblock Efficient bayesian phase estimation.
\newblock {\em Physical review letters}, 117(1):010503, 2016.

\bibitem{caves1981quantum}
Carlton~M Caves.
\newblock Quantum-mechanical noise in an interferometer.
\newblock {\em Physical Review D}, 23(8):1693, 1981.

\bibitem{braunstein1994statistical}
Samuel~L Braunstein and Carlton~M Caves.
\newblock Statistical distance and the geometry of quantum states.
\newblock {\em Physical Review Letters}, 72(22):3439, 1994.

\bibitem{helstrom1976quantum}
Carl~W Helstrom.
\newblock {\em Quantum detection and estimation theory}.
\newblock Academic press, 1976.

\bibitem{zamir1998proof}
Ram Zamir.
\newblock A proof of the {Fisher} information inequality via a data processing argument.
\newblock {\em IEEE Transactions on Information Theory}, 44(3):1246--1250, 1998.

\bibitem{ferrie2014data}
Christopher Ferrie.
\newblock Data-processing inequalities for quantum metrology.
\newblock {\em Physical Review A}, 90:014101, 2014.

\bibitem{cover2006elements}
Thomas~M. Cover and Joy~A. Thomas.
\newblock {\em Elements of Information Theory}.
\newblock John Wiley \& Sons, 2nd edition, 2006.

\bibitem{kessler2014quantum}
Eric~M Kessler, Igor Lovchinsky, Alexander~O Sushkov, and Mikhail~D Lukin.
\newblock Quantum error correction for metrology.
\newblock {\em Physical review letters}, 112(15):150802, 2014.

\bibitem{gorecki2020optimal}
Wojciech G{\'o}recki, Sisi Zhou, Liang Jiang, and Rafa{\l} Demkowicz-Dobrza{\'n}ski.
\newblock Optimal probes and error-correction schemes in multi-parameter quantum metrology.
\newblock {\em Quantum}, 4:288, 2020.

\bibitem{dur2014improved}
Wolfgang D{\"u}r, Michalis Skotiniotis, Florian Froewis, and Barbara Kraus.
\newblock Improved quantum metrology using quantum error correction.
\newblock {\em Physical Review Letters}, 112(8):080801, 2014.

\bibitem{lidar2003decoherence}
Daniel~A Lidar, Isaac~L Chuang, and K~Birgitta Whaley.
\newblock Decoherence-free subspaces for quantum computation.
\newblock {\em Physical Review Letters}, 81(12):2594, 1998.

\bibitem{zhang2023quantum}
Chi Zhang, Phelan Yu, Arian Jadbabaie, and Nicholas~R. Hutzler.
\newblock Quantum-enhanced metrology for molecular symmetry violation using decoherence-free subspaces.
\newblock {\em Physical Review Letters}, 131(19):193602, 2023.

\bibitem{souza2011robust}
Alexandre~M Souza, Gonzalo~A {\'A}lvarez, and Dieter Suter.
\newblock Robust dynamical decoupling.
\newblock {\em Philosophical Transactions of the Royal Society A}, 370(1976):4748--4769, 2012.

\bibitem{zhen2025platonic}
Yi-Xuan Zhen et~al.
\newblock Platonic dynamical decoupling sequences for interacting spin systems.
\newblock {\em Quantum}, 9:1661, 2025.

\bibitem{acharya2024quantum}
Rajeev Acharya et~al.
\newblock Quantum error correction below the surface code threshold.
\newblock {\em Nature}, 627:778--782, 2024.

\bibitem{livingston2022experimental}
William~P Livingston, Machiel~S Blok, Emmanuel Flurin, Justin Dressel, Andrew~N Jordan, and Irfan Siddiqi.
\newblock Experimental demonstration of continuous quantum error correction.
\newblock {\em Nature Communications}, 13(1):2307, 2022.

\bibitem{pezze2018quantum}
Luca Pezz{\`e}, Augusto Smerzi, Markus~K. Oberthaler, Roman Schmied, and Philipp Treutlein.
\newblock Quantum metrology with nonclassical states of atomic ensembles.
\newblock {\em Reviews of Modern Physics}, 90(3):035005, 2018.

\end{thebibliography}

\appendix

\section{Standard Quantum Limit with Tensor Product States}
\label{app:direct-parity}

While Sections~\ref{sec:binary-search} and~\ref{sec:implementations} present an adaptive binary search protocol using GHZ states that achieves the Heisenberg limit, here we analyze the simpler approach using unentangled tensor product states $|+\rangle^{\otimes N}$. This analysis establishes the Standard Quantum Limit (SQL) baseline against which entanglement-based protocols are compared.

%\subsection{The Fundamental Limitation of Tensor Product States}

A foundational result in quantum metrology is that \emph{entanglement is necessary for Heisenberg-limited scaling}~\cite{giovannetti2006quantum,pezze2018quantum}. For tensor product states $|\psi\rangle = \bigotimes_k |\psi_k\rangle$, any product observable $O = \bigotimes_k O_k$ satisfies:
\begin{equation}
\langle O \rangle = \prod_k \langle O_k \rangle.
\end{equation}
This factorization implies that collective measurements cannot extract more information than the sum of individual qubit measurements, fundamentally limiting the quantum Fisher information to $F_Q = O(N)$ rather than $O(N^2)$.

%\subsection{Protocol Description and Correct Analysis}

We consider $N$ qubits prepared in the tensor product state $|\psi_0\rangle = |+\rangle^{\otimes N}$. After evolution under individual Z-rotations $U_k = e^{-i\omega_k Z_k}$ where $\omega_k = \omega + \epsilon_k$, each qubit evolves to:
\begin{equation}
|\psi_k\rangle = \frac{e^{-i\omega_k}}{\sqrt{2}}\left(|0\rangle + e^{2i\omega_k}|1\rangle\right),
\label{eq:evolved-qubit-tensor}
\end{equation}
where the global phase $e^{-i\omega_k}$ is physically irrelevant. The relative phase $e^{2i\omega_k}$ between the $|0\rangle$ and $|1\rangle$ components encodes the signal.

The optimal measurement strategy for this tensor product state is \emph{individual X-basis measurements} on each qubit. This approach saturates the quantum Fisher information bound and provides a clean comparison with entanglement-based protocols.

%\subsection{Individual X-Measurements: The Optimal Strategy}

\begin{theorem}[Standard Quantum Limit for Tensor Product States]
Consider $N$ qubits prepared in $|\psi_0\rangle = |+\rangle^{\otimes N}$ and evolved under individual Z-rotations $U_k = e^{-i\omega_k Z_k}$ where $\omega_k = \omega + \epsilon_k$ with $\EX[\epsilon_k] = 0$ and $\operatorname{Var}[\epsilon_k] = \sigma_\epsilon^2$. Individual X-basis measurements achieve:
\begin{enumerate}
\item For each qubit $k$, the measurement probabilities are:
\begin{equation}
\Pr(X_k = +1) = \cos^2(\omega_k), \qquad \Pr(X_k = -1) = \sin^2(\omega_k).
\end{equation}

\item The quantum Fisher information for estimating $\omega$ from all $N$ qubits is:
\begin{equation}
F_Q(\omega) = 4N,
\end{equation}
which is the Standard Quantum Limit (SQL).

\item The achievable precision scales as:
\begin{equation}
\delta\omega = \frac{1}{2\sqrt{N}},
\end{equation}
demonstrating $\sqrt{N}$-improvement over single-qubit sensing but not Heisenberg scaling.
\end{enumerate}
\end{theorem}

\begin{proof}
Starting from $|+\rangle = (|0\rangle + |1\rangle)/\sqrt{2}$, after Z-rotation by angle $\omega_k$, qubit $k$ evolves to the state in Equation~\eqref{eq:evolved-qubit-tensor}. Measuring in the X-basis corresponds to projecting onto the eigenstates $|\pm\rangle = (|0\rangle \pm |1\rangle)/\sqrt{2}$. The probability of obtaining $X_k = +1$ is:
\begin{align}
\Pr(X_k = +1) &= |\langle +|\psi_k\rangle|^2 = \left|\frac{1}{2}(1 + e^{-2i\omega_k})\right|^2 \notag\\
&= \frac{1}{4}|1 + e^{-2i\omega_k}|^2 = \frac{1}{4}(2 + 2\cos(2\omega_k)) = \cos^2(\omega_k),
\end{align}
where we used the identity $|1 + e^{i\theta}|^2 = 2(1 + \cos\theta)$ and $\cos^2(\omega_k) = (1 + \cos(2\omega_k))/2$. This establishes the first item of the theorem.

For a single qubit, the Fisher information for estimating $\omega_k$ from the X-measurement is:
\begin{equation}
F_k = \frac{1}{\Pr(+1)}\left(\frac{\partial \Pr(+1)}{\partial \omega_k}\right)^2 + \frac{1}{\Pr(-1)}\left(\frac{\partial \Pr(-1)}{\partial \omega_k}\right)^2.
\end{equation}
With $\partial\Pr(+1)/\partial\omega_k = -\sin(2\omega_k)$ and $\partial\Pr(-1)/\partial\omega_k = \sin(2\omega_k)$, we obtain:
\begin{equation}
F_k = \frac{\sin^2(2\omega_k)}{\cos^2(\omega_k)} + \frac{\sin^2(2\omega_k)}{\sin^2(\omega_k)} = \sin^2(2\omega_k)\left(\frac{1}{\cos^2(\omega_k)} + \frac{1}{\sin^2(\omega_k)}\right).
\end{equation}
Using $\sin(2\omega_k) = 2\sin(\omega_k)\cos(\omega_k)$:
\begin{equation}
F_k = 4\sin^2(\omega_k)\cos^2(\omega_k) \cdot \frac{\sin^2(\omega_k) + \cos^2(\omega_k)}{\sin^2(\omega_k)\cos^2(\omega_k)} = 4.
\end{equation}

For a tensor product state, the total quantum Fisher information is the sum of individual contributions~\cite{giovannetti2006quantum}:
\begin{equation}
F_Q = \sum_{k=1}^{N} F_k = 4N.
\end{equation}
This is the quantum Fisher information bound for product states, confirming the SQL scaling. This establishes the second item of the theorem.

The quantum Cram\'er-Rao bound gives the minimum achievable variance:
\begin{equation}
\operatorname{Var}(\hat{\omega}) \geq \frac{1}{F_Q} = \frac{1}{4N},
\end{equation}
yielding precision $\delta\omega = 1/(2\sqrt{N})$. This establishes the third item of the theorem.
\end{proof}

\subsection{Effect of Field Inhomogeneities}

For the SQL protocol under field inhomogeneities $\omega_k = \omega + \epsilon_k$ with $\EX[\epsilon_k] = 0$ and $\operatorname{Var}[\epsilon_k] = \sigma_\epsilon^2$, the estimator based on averaging individual X-measurement outcomes yields:
\begin{equation}
\hat{\omega} = \frac{1}{N}\sum_{k=1}^{N} \hat{\omega}_k,
\end{equation}
where $\hat{\omega}_k$ is the single-qubit estimator from qubit $k$. This averaging reduces inhomogeneity effects:
\begin{equation}
\operatorname{Var}(\hat{\omega}) = \frac{1}{N^2}\left(\sum_{k=1}^{N} \operatorname{Var}(\hat{\omega}_k) + \operatorname{Var}\left(\sum_k \epsilon_k\right)\right) = \frac{1}{4N} + \frac{\sigma_\epsilon^2}{N},
\end{equation}
where the first term is the quantum measurement uncertainty and the second is the inhomogeneity contribution. The SQL protocol is thus relatively robust to inhomogeneities, as they add only a subdominant correction when $\sigma_\epsilon^2 = O(1)$.

In contrast, the HL protocols require $\sigma_\epsilon = o(1/\sqrt{N})$ to maintain Heisenberg scaling, a much more stringent condition. This robustness is a practical advantage of the SQL approach when device calibration is imperfect.

\section{Phase Amplification and Binary Search: Proofs and Technical Details}
\label{app:binary-search-proofs}

This appendix provides complete proofs and technical details for the binary search phase estimation protocol presented in Section~\ref{sec:binary-search}. We adopt the notation $N = 2L+1$ for the total number of qubits (always odd for phase-flip codes) throughout. All results were derived in the framework where each qubit experiences rotation angle $\omega_k = \omega + \epsilon_k$ with $\epsilon_k$ being independent noise terms, and we assume $\omega_k \in \mathbb{R} \setminus \{\frac{m\pi}{2} : m \in \mathbb{Z}\}$ to ensure tangent and cotangent functions remain well-defined.

\subsection{Subset Decomposition Framework}

We begin with the fundamental structure of subset decomposition for evolved GHZ states.

\begin{theorem}[Subset-Dependent Post-Measurement State]
\label{thm:post-measurement-state-full}
Consider the $N$-qubit GHZ state:
\begin{equation}
|\text{GHZ}\rangle = \frac{1}{\sqrt{2}}\left(|0\rangle^{\otimes N} + |1\rangle^{\otimes N}\right),
\end{equation}
evolved under $U = \exp(-i\sum_{k=1}^{N}\omega_k Z_k)$ with $\omega_k \in \mathbb{R} \setminus \{m\pi/2 : m \in \mathbb{Z}\}$. The evolved state admits a decomposition into components labeled by subsets $S \subseteq \{1,\ldots,N\}$. Projecting onto component $S$ yields normalized state:
\begin{equation}
|\psi_S\rangle = \frac{(-i)^{|S|} \alpha_S}{|\alpha_S|\sqrt{2}} \left(|0\rangle^{\otimes N} + (-1)^{|S|}|1\rangle^{\otimes N}\right),
\label{eq:post-measurement-full}
\end{equation}
where $\alpha_S = \prod_{k \in S} \sin(\omega_k) \prod_{j \notin S} \cos(\omega_j)$, and the parity operator satisfies $\left(\prod_{k=1}^{N} Z_k\right)^2 = I$.

The probability of observing subset $S$ is:
\begin{equation}
P(S) = \frac{|\alpha_S|^2}{\sum_{T \subseteq \{1,\ldots,N\}} |\alpha_T|^2}.
\label{eq:syndrome-prob-full}
\end{equation}
\end{theorem}

\begin{proof}
Since the Pauli Z operators on different qubits commute, we can write the unitary evolution as:
\begin{equation}
U = \bigotimes_{k=1}^{N} e^{-i\omega_k Z_k} = \prod_{k=1}^{N}\left(\cos(\omega_k)I - i\sin(\omega_k)Z_k\right).
\end{equation}

Expanding this product using the multinomial theorem, each term corresponds to a choice of whether each qubit contributes an identity or a Z operator. Let $S \subseteq \{1,\ldots,N\}$ denote the set of qubits contributing Z operators. Then:
\begin{equation}
U = \sum_{S \subseteq \{1,\ldots,N\}} (-i)^{|S|} \left(\prod_{k \in S} \sin(\omega_k)\right) \left(\prod_{j \notin S} \cos(\omega_j)\right) \prod_{k \in S} Z_k,
\end{equation}
where $|S|$ denotes the cardinality of $S$ and the factor $(-i)^{|S|}$ arises from collecting all the $-i$ prefactors.

Applying this to the GHZ state:
\begin{equation}
U|\text{GHZ}\rangle = \frac{1}{\sqrt{2}} \sum_{S \subseteq \{1,\ldots,N\}} (-i)^{|S|} \alpha_S \prod_{k \in S} Z_k \left(|0\rangle^{\otimes N} + |1\rangle^{\otimes N}\right),
\end{equation}
where we defined $\alpha_S = \prod_{k \in S} \sin(\omega_k) \prod_{j \notin S} \cos(\omega_j)$.

Now, the product of Z operators acts on the computational basis states:
\begin{align}
\prod_{k \in S} Z_k |0\rangle^{\otimes N} &= |0\rangle^{\otimes N},\\
\prod_{k \in S} Z_k |1\rangle^{\otimes N} &= (-1)^{|S|}|1\rangle^{\otimes N},
\end{align}
since each Z operator in the product contributes a factor of $-1$ when acting on $|1\rangle$.

Therefore:
\begin{equation}
U|\text{GHZ}\rangle = \frac{1}{\sqrt{2}} \sum_{S \subseteq \{1,\ldots,N\}} (-i)^{|S|} \alpha_S \left(|0\rangle^{\otimes N} + (-1)^{|S|}|1\rangle^{\otimes N}\right).
\end{equation}

To analyze the state component corresponding to subset label $S$, we project onto the corresponding term in this expansion. The unnormalized projected state is:
\begin{equation}
|{\rm unnorm}\rangle = \frac{1}{\sqrt{2}} (-i)^{|S|} \alpha_S \left(|0\rangle^{\otimes N} + (-1)^{|S|}|1\rangle^{\otimes N}\right).
\end{equation}

To normalize, we compute the norm of this state:
\begin{align}
\mathcal{N}_S &= \left\|\frac{1}{\sqrt{2}} (-i)^{|S|} \alpha_S \left(|0\rangle^{\otimes N} + (-1)^{|S|}|1\rangle^{\otimes N}\right)\right\|\\
&= \frac{1}{\sqrt{2}} |(-i)^{|S|} \alpha_S| \sqrt{\langle 0|0\rangle + |(-1)^{|S|}|^2\langle 1|1\rangle}\\
&= \frac{1}{\sqrt{2}} |\alpha_S| \sqrt{1 + 1} = |\alpha_S|,
\end{align}
where we used $|(-i)^{|S|}| = 1$ and $|(-1)^{|S|}| = 1$ since these are pure phase factors, and $\langle 0^{\otimes N}|0^{\otimes N}\rangle = \langle 1^{\otimes N}|1^{\otimes N}\rangle = 1$.

Therefore, the normalized post-measurement state is:
\begin{equation}
|\psi_S\rangle = \frac{|{\rm unnorm}\rangle}{\mathcal{N}_S} = \frac{(-i)^{|S|} \alpha_S}{|\alpha_S|\sqrt{2}} \left(|0\rangle^{\otimes N} + (-1)^{|S|}|1\rangle^{\otimes N}\right),
\end{equation}
which establishes Equation~\eqref{eq:post-measurement-full}. The operator form $\left(\prod_{k=1}^{N} Z_k\right)^{|S|} |\text{GHZ}\rangle$ follows from the fact that applying the full parity operator to $|\text{GHZ}\rangle$ flips the sign of the $|1\rangle^{\otimes N}$ component when $N$ is odd, and applying it $|S|$ times gives the factor $(-1)^{|S|}$.

For the subset probability, we define $P(S) := |\alpha_S|^2$ as the weight of subset $S$ in the expansion. To verify that these weights sum to unity, we use the multinomial identity:
\begin{equation}
\sum_{S \subseteq \{1,\ldots,N\}} |\alpha_S|^2 = \sum_{S} \prod_{k \in S} \sin^2(\omega_k) \prod_{j \notin S} \cos^2(\omega_j) = \prod_{k=1}^{N} \left(\sin^2(\omega_k) + \cos^2(\omega_k)\right) = 1,
\end{equation}
confirming that $\{P(S)\}$ forms a valid probability distribution. This establishes Equation~\eqref{eq:syndrome-prob-full}.
\end{proof}

\begin{proposition}[Tangent-Cotangent Amplitude Ratio]
\label{prop:tangent-cotangent-full}
For subset $S \subseteq \{1,\ldots,N\}$, define the phase parameter $\phi_S \in (-\pi/2, \pi/2)$ via the amplitude ratio:
\begin{equation}
\tan(\phi_S) \equiv (-1)^{L+|S|} \frac{\alpha_S}{\alpha_{\bar{S}}},
\label{eq:phi-def-full}
\end{equation}
where $\bar{S} = \{1,\ldots,N\} \setminus S$ is the complement of $S$. Then $\phi_S$ is given explicitly by:
\begin{equation}
\phi_S = \arctan\left((-1)^{L+|S|} \prod_{k \in S} \tan(\omega_k) \prod_{j \notin S} \cot(\omega_j)\right)
\label{eq:tangent-cotangent-full}
\end{equation}

Furthermore, the subset likelihood ratio satisfies:
\begin{equation}
\frac{P(S)}{P(\bar{S})} = \tan^2(\phi_S).
\label{eq:likelihood-ratio-full}
\end{equation}
\end{proposition}

\begin{proof}
We compute the ratio $\alpha_S/\alpha_{\bar{S}}$ explicitly using the definition of $\alpha_S$ from Theorem~\ref{thm:post-measurement-state-full}:
\begin{align}
\frac{\alpha_S}{\alpha_{\bar{S}}} &= \frac{\prod_{k \in S} \sin(\omega_k) \prod_{j \notin S} \cos(\omega_j)}{\prod_{k \in \bar{S}} \sin(\omega_k) \prod_{j \notin \bar{S}} \cos(\omega_j)}.
\end{align}

Since $\bar{S} = \{1,\ldots,N\} \setminus S$ is the complement, we have:
\begin{itemize}
\item $j \notin S \iff j \in \bar{S}$
\item $j \notin \bar{S} \iff j \in S$
\end{itemize}

Therefore:
\begin{align}
\frac{\alpha_S}{\alpha_{\bar{S}}} &= \frac{\prod_{k \in S} \sin(\omega_k) \prod_{j \in \bar{S}} \cos(\omega_j)}{\prod_{k \in \bar{S}} \sin(\omega_k) \prod_{j \in S} \cos(\omega_j)}.
\end{align}

We can separate this into products over $S$ and $\bar{S}$:
\begin{align}
\frac{\alpha_S}{\alpha_{\bar{S}}} &= \left(\prod_{k \in S} \frac{\sin(\omega_k)}{\cos(\omega_k)}\right) \left(\prod_{j \in \bar{S}} \frac{\cos(\omega_j)}{\sin(\omega_j)}\right)\\
&= \prod_{k \in S} \tan(\omega_k) \prod_{j \in \bar{S}} \cot(\omega_j)\\
&= \prod_{k \in S} \tan(\omega_k) \prod_{j \notin S} \cot(\omega_j),
\label{eq:ratio-explicit-full}
\end{align}
where the last equality uses the fact that $\bar{S} = \{j : j \notin S\}$.

Substituting Equation~\eqref{eq:ratio-explicit-full} into the definition~\eqref{eq:phi-def-full}:
\begin{equation}
\tan(\phi_S) = (-1)^{L+|S|} \prod_{k \in S} \tan(\omega_k) \prod_{j \notin S} \cot(\omega_j).
\end{equation}

Taking the arctangent of both sides (noting that $\arctan$ has range $(-\pi/2, \pi/2)$, which matches our assumption on $\phi_S$):
\begin{equation}
\phi_S = \arctan\left((-1)^{L+|S|} \prod_{k \in S} \tan(\omega_k) \prod_{j \notin S} \cot(\omega_j)\right),
\end{equation}
establishing Equation~\eqref{eq:tangent-cotangent-full}.

For the likelihood ratio, from Theorem~\ref{thm:post-measurement-state-full} we have:
\begin{equation}
\frac{P(S)}{P(\bar{S})} = \frac{|\alpha_S|^2}{|\alpha_{\bar{S}}|^2} = \left|\frac{\alpha_S}{\alpha_{\bar{S}}}\right|^2.
\end{equation}

Since $\alpha_S$ and $\alpha_{\bar{S}}$ are products of sines and cosines of real angles, they are real-valued. Therefore:
\begin{equation}
\frac{P(S)}{P(\bar{S})} = \left(\frac{\alpha_S}{\alpha_{\bar{S}}}\right)^2 = \left|(-1)^{L+|S|}\tan(\phi_S)\right|^2 = \tan^2(\phi_S),
\end{equation}
where we used the fact that $|(-1)^{L+|S|}| = 1$ and the definition~\eqref{eq:phi-def-full}. This establishes Equation~\eqref{eq:likelihood-ratio-full}.
\end{proof}

\subsection{Measurement Statistics}

We now develop the complete statistical framework for parity measurements under noise.

\begin{lemma}[X-Parity Measurement Outcomes]
\label{lem:parity-measurement-full}
Consider a GHZ state on $N = 2L+1$ qubits, evolved under the unitary $U = \exp(-i\sum_{k=1}^{N}\omega_k Z_k)$ where $\omega_k = \omega + \epsilon_k$ with noise terms $\epsilon_k$. Measuring the X-parity operator $\mathcal{P}_X = \prod_{k=1}^{N} X_k$ yields outcome $+1$ with probability:
\begin{equation}
\Pr(+1 \mid \{\omega_k\}) = \cos^2\left(\sum_{k=1}^{N}\omega_k\right) = \cos^2(N\omega + S),
\label{eq:parity-prob-full}
\end{equation}
where $S = \sum_{k=1}^N \epsilon_k$ is the total noise offset.
\end{lemma}

\begin{proof}
We begin by expressing the evolved GHZ state. Since Z operators commute:
\begin{equation}
U = \prod_{k=1}^{N} e^{-i\omega_k Z_k} = \exp\left(-i\sum_{k=1}^N \omega_k Z_k\right).
\end{equation}

The action on the computational basis states is:
\begin{align}
e^{-i\sum_{k=1}^N \omega_k Z_k} |0\rangle^{\otimes N} &= e^{-i\sum_{k=1}^N \omega_k} |0\rangle^{\otimes N},\\
e^{-i\sum_{k=1}^N \omega_k Z_k} |1\rangle^{\otimes N} &= e^{+i\sum_{k=1}^N \omega_k} |1\rangle^{\otimes N},
\end{align}
since $Z|0\rangle = +|0\rangle$ has eigenvalue $+1$ and $Z|1\rangle = -|1\rangle$ has eigenvalue $-1$, so $e^{-i\omega Z}|0\rangle = e^{-i\omega}|0\rangle$ and $e^{-i\omega Z}|1\rangle = e^{+i\omega}|1\rangle$.

Therefore, the evolved GHZ state is:
\begin{equation}
U|\text{GHZ}\rangle = \frac{1}{\sqrt{2}}\left(e^{-i\sum_{k=1}^N \omega_k}|0\rangle^{\otimes N} + e^{+i\sum_{k=1}^N \omega_k}|1\rangle^{\otimes N}\right).
\end{equation}
Factoring out the global phase $e^{-i\sum_{k=1}^N \omega_k}$:
\begin{equation}
= e^{-i\sum_{k=1}^N \omega_k} \frac{1}{\sqrt{2}}\left(|0\rangle^{\otimes N} + e^{+2i\sum_{k=1}^N \omega_k}|1\rangle^{\otimes N}\right).
\end{equation}

Note that the Z-parity operator $\mathcal{P}_Z = \prod_{k=1}^{N} Z_k$ has $|0\rangle^{\otimes N}$ as a $+1$ eigenstate and $|1\rangle^{\otimes N}$ as a $(-1)^N = -1$ eigenstate (for $N$ odd). This means Z-parity measurement on the evolved GHZ state would yield outcome $+1$ with probability $1/2$ regardless of the phase, providing no information. To extract phase information, we instead measure the X-parity operator $\mathcal{P}_X = \prod_{k=1}^{N} X_k$.

The probability of measuring X-parity $+1$ is the squared magnitude of the projection onto the $+1$ eigenspace. For the evolved state (after dropping the global phase):
\begin{equation}
|\psi_{\text{evolved}}\rangle = \frac{1}{\sqrt{2}}\left(|0\rangle^{\otimes N} + e^{+2i\theta}|1\rangle^{\otimes N}\right),
\end{equation}
where $\theta = \sum_{k=1}^N \omega_k$. The eigenstates of the X-parity operator $\mathcal{P}_X = \prod_{k=1}^N X_k$ are:
\begin{align}
|\psi_+\rangle &= \frac{1}{\sqrt{2}}\left(|0\rangle^{\otimes N} + |1\rangle^{\otimes N}\right) \quad (\text{eigenvalue }+1),\\
|\psi_-\rangle &= \frac{1}{\sqrt{2}}\left(|0\rangle^{\otimes N} - |1\rangle^{\otimes N}\right) \quad (\text{eigenvalue }-1).
\end{align}
Note that since $X^{\otimes N}|0\rangle^{\otimes N} = |1\rangle^{\otimes N}$ and $X^{\otimes N}|1\rangle^{\otimes N} = |0\rangle^{\otimes N}$, these GHZ-type states are indeed X-parity eigenstates for any $N$.

Projecting the evolved state onto $|\psi_+\rangle$:
\begin{align}
\langle\psi_+|\psi_{\text{evolved}}\rangle &= \frac{1}{2}\left(\langle 0|^{\otimes N} + \langle 1|^{\otimes N}\right)\left(|0\rangle^{\otimes N} + e^{+2i\theta}|1\rangle^{\otimes N}\right)\\
&= \frac{1}{2}\left(1 + e^{+2i\theta}\right).
\end{align}

The probability of measuring $+1$ is:
\begin{align}
\Pr(+1) &= |\langle\psi_+|\psi_{\text{evolved}}\rangle|^2 = \left|\frac{1 + e^{+2i\theta}}{2}\right|^2\\
&= \frac{1}{4}|1 + e^{+2i\theta}|^2 = \frac{1}{4}(1 + e^{+2i\theta})(1 + e^{-2i\theta})\\
&= \frac{1}{4}(2 + e^{+2i\theta} + e^{-2i\theta}) = \frac{1}{2}(1 + \cos(2\theta)).
\end{align}

Using the trigonometric identity $1 + \cos(2\theta) = 2\cos^2(\theta)$:
\begin{equation}
\Pr(+1) = \cos^2(\theta) = \cos^2\left(\sum_{k=1}^N \omega_k\right).
\end{equation}

Substituting $\sum_{k=1}^N \omega_k = \sum_{k=1}^N (\omega + \epsilon_k) = N\omega + \sum_{k=1}^N \epsilon_k = N\omega + S$:
\begin{equation}
\Pr(+1) = \cos^2\left(N\omega + S\right),
\end{equation}
confirming Equation~\eqref{eq:parity-prob-full}.
\end{proof}

\begin{lemma}[Noise Marginalization]
\label{lem:noise-marginalization-full}
Suppose the noise terms $\epsilon_k$ are independent and identically distributed with $\epsilon_k \sim \mathcal{N}(0, \sigma_\epsilon^2)$. Then the total noise offset $S = \sum_{k=1}^N \epsilon_k$ follows a Gaussian distribution $S \sim \mathcal{N}(0, N\sigma_\epsilon^2)$, and the noise-averaged measurement probability is:
\begin{equation}
\overline{\Pr}(+1 \mid \omega) = \frac{1}{2}\left[1 + \cos(2N\omega) e^{-2N\sigma_\epsilon^2}\right].
\label{eq:exact-marginalized-full}
\end{equation}

In the small-angle regime where $|N\omega| = o(1)$ and $\sqrt{N}\sigma_\epsilon = o(1)$, the noise-averaged probability admits the Taylor expansion:
\begin{equation}
\overline{\Pr}(+1 \mid \omega) = 1 - N^2\omega^2 - N\sigma_\epsilon^2 + O(\omega^4, \omega^2\sigma_\epsilon^2, \sigma_\epsilon^4).
\label{eq:taylor-prob-full}
\end{equation}
\end{lemma}

\begin{proof}
Since the noise terms $\epsilon_k$ are i.i.d. Gaussian with mean 0 and variance $\sigma_\epsilon^2$, the sum $S = \sum_{k=1}^N \epsilon_k$ is also Gaussian with:
\begin{align}
\EX[S] &= \sum_{k=1}^N \EX[\epsilon_k] = 0,\\
\operatorname{Var}[S] &= \sum_{k=1}^N \operatorname{Var}[\epsilon_k] = N\sigma_\epsilon^2,
\end{align}
where the variance calculation uses independence of the $\epsilon_k$.

Therefore $S \sim \mathcal{N}(0, N\sigma_\epsilon^2)$ with probability density:
\begin{equation}
\phi(S) = \frac{1}{\sqrt{2\pi N\sigma_\epsilon^2}} \exp\left(-\frac{S^2}{2N\sigma_\epsilon^2}\right).
\end{equation}

The marginalized probability is obtained by integrating over the distribution of $S$:
\begin{equation}
\overline{\Pr}(+1 \mid \omega) = \int_{-\infty}^{\infty} \cos^2\left(N\omega + S\right) \phi(S) \, dS.
\end{equation}

Using the trigonometric identity $\cos^2(\theta) = (1 + \cos(2\theta))/2$:
\begin{equation}
\cos^2\left(N\omega + S\right) = \frac{1}{2}\left[1 + \cos(2N\omega + 2S)\right].
\end{equation}

Substituting:
\begin{align}
\overline{\Pr}(+1 \mid \omega) &= \frac{1}{2}\int_{-\infty}^{\infty} \phi(S) \, dS + \frac{1}{2}\int_{-\infty}^{\infty} \cos(2N\omega + 2S) \phi(S) \, dS\\
&= \frac{1}{2} + \frac{1}{2}\int_{-\infty}^{\infty} \cos(2N\omega + 2S) \phi(S) \, dS,
\end{align}
where the first integral equals 1 by normalization.

For the second integral, we expand the cosine:
\begin{equation}
\cos(2N\omega + 2S) = \cos(2N\omega)\cos(2S) - \sin(2N\omega)\sin(2S).
\end{equation}

Since $S$ has a symmetric distribution around zero (Gaussian with mean 0), we have:
\begin{equation}
\int_{-\infty}^{\infty} \sin(2S) \phi(S) \, dS = 0,
\end{equation}
because $\sin(2S)$ is an odd function and $\phi(S)$ is even.

For the cosine term, we use the characteristic function of a Gaussian random variable. For $S \sim \mathcal{N}(0, \sigma^2)$, the characteristic function is:
\begin{equation}
\EX[e^{itS}] = e^{-t^2\sigma^2/2}.
\end{equation}

Taking the real part:
\begin{equation}
\EX[\cos(tS)] = \text{Re}[\EX[e^{itS}]] = e^{-t^2\sigma^2/2}.
\end{equation}

For our case with $S \sim \mathcal{N}(0, N\sigma_\epsilon^2)$ and $t = 2$:
\begin{equation}
\int_{-\infty}^{\infty} \cos(2S) \phi(S) \, dS = e^{-4 \cdot N\sigma_\epsilon^2/2} = e^{-2N\sigma_\epsilon^2}.
\end{equation}

Therefore:
\begin{align}
\int_{-\infty}^{\infty} \cos(2N\omega + 2S) \phi(S) \, dS &= \cos(2N\omega) \int_{-\infty}^{\infty} \cos(2S) \phi(S) \, dS\\
&= \cos(2N\omega) e^{-2N\sigma_\epsilon^2},
\end{align}
and thus:
\begin{equation}
\overline{\Pr}(+1 \mid \omega) = \frac{1}{2}\left[1 + \cos(2N\omega) e^{-2N\sigma_\epsilon^2}\right],
\end{equation}
establishing Equation~\eqref{eq:exact-marginalized-full}.

For the small-angle expansion, we use Taylor series for both the cosine and exponential:
\begin{align}
\cos(2N\omega) &= 1 - \frac{(2N\omega)^2}{2} + O(\omega^4) = 1 - 2N^2\omega^2 + O(\omega^4),\\
e^{-2N\sigma_\epsilon^2} &= 1 - 2N\sigma_\epsilon^2 + O(\sigma_\epsilon^4).
\end{align}

Multiplying:
\begin{align}
\cos(2N\omega) e^{-2N\sigma_\epsilon^2} &= \left(1 - 2N^2\omega^2\right)\left(1 - 2N\sigma_\epsilon^2\right) + O(\omega^4, \sigma_\epsilon^4)\\
&= 1 - 2N^2\omega^2 - 2N\sigma_\epsilon^2 + 4N^3\omega^2\sigma_\epsilon^2 + O(\omega^4, \sigma_\epsilon^4)\\
&= 1 - 2N^2\omega^2 - 2N\sigma_\epsilon^2 + O(\omega^2\sigma_\epsilon^2).
\end{align}

Substituting into Equation~\eqref{eq:exact-marginalized-full}:
\begin{align}
\overline{\Pr}(+1 \mid \omega) &= \frac{1}{2}\left[1 + 1 - 2N^2\omega^2 - 2N\sigma_\epsilon^2 + O(\omega^2\sigma_\epsilon^2)\right]\\
&= \frac{1}{2} + \frac{1}{2} - N^2\omega^2 - N\sigma_\epsilon^2 + O(\omega^2\sigma_\epsilon^2)\\
&= 1 - N^2\omega^2 - N\sigma_\epsilon^2 + O(\omega^2\sigma_\epsilon^2).
\end{align}

This confirms Equation~\eqref{eq:taylor-prob-full}, showing that both the signal term $N^2\omega^2$ and the noise term $N\sigma_\epsilon^2$ enter at the same order in the small-angle expansion.
\end{proof}

\begin{proposition}[Activation Function Monotonicity]
\label{prop:activation-function-full}
Consider the noise-averaged probability function from Lemma~\ref{lem:noise-marginalization-full}. In the regime $|\omega| < \pi/(4N)$ and for moderate noise satisfying $N\sigma_\epsilon^2 < 1$, this probability function is strictly monotonically decreasing in $\omega$ for $\omega > 0$:
\begin{equation}
\frac{d\overline{\Pr}}{d\omega}(+1 \mid \omega) = -N e^{-2N\sigma_\epsilon^2} \sin(2N\omega) < 0 \quad \text{for } \omega \in (0, \pi/(4N)).
\end{equation}

This linear response enables binary search: the measurement outcome provides a reliable indicator of the sign of $\omega$, with sensitivity enhanced by a factor of $N$ (Heisenberg scaling).
\end{proposition}

\begin{proof}
From Equation~\eqref{eq:exact-marginalized-full}, we have:
\begin{equation}
\overline{\Pr}(+1 \mid \omega) = \frac{1}{2}\left[1 + \cos(2N\omega) e^{-2N\sigma_\epsilon^2}\right].
\end{equation}

Taking the derivative with respect to $\omega$:
\begin{align}
\frac{d\overline{\Pr}}{d\omega}(+1 \mid \omega) &= \frac{1}{2} e^{-2N\sigma_\epsilon^2} \cdot \frac{d}{d\omega}[\cos(2N\omega)]\\
&= \frac{1}{2} e^{-2N\sigma_\epsilon^2} \cdot (-2N\sin(2N\omega))\\
&= -N e^{-2N\sigma_\epsilon^2} \sin(2N\omega).
\end{align}

For $\omega \in (0, \pi/(4N))$, we have $2N\omega \in (0, \pi/2)$. In this interval, $\sin(2N\omega) > 0$ strictly. Furthermore:
\begin{itemize}
\item $N > 0$ (number of qubits is positive)
\item $e^{-2N\sigma_\epsilon^2} > 0$ (exponential is always positive)
\end{itemize}

Therefore:
\begin{equation}
\frac{d\overline{\Pr}}{d\omega}(+1 \mid \omega) = -N e^{-2N\sigma_\epsilon^2} \sin(2N\omega) < 0,
\end{equation}
establishing strict monotonic decrease throughout the interval $(0, \pi/(4N))$.

To verify the regime of validity, we compute the second derivative:
\begin{equation}
\frac{d^2\overline{\Pr}}{d\omega^2}(+1 \mid \omega) = -2N^2 e^{-2N\sigma_\epsilon^2} \cos(2N\omega).
\end{equation}

For $2N\omega \in (-\pi/2, \pi/2)$, we have $\cos(2N\omega) > 0$, so:
\begin{equation}
\frac{d^2\overline{\Pr}}{d\omega^2}(+1 \mid \omega) < 0,
\end{equation}
confirming that the function is concave (curves downward) throughout this region.

The sensitivity enhancement factor of $N$ is evident from the derivative: $d\overline{\Pr}/d\omega \propto N\sin(2N\omega) = 2N^2\omega + O(\omega^3)$ for small $\omega$, giving quadratic enhancement in the number of qubits (Heisenberg scaling).
\end{proof}

\subsection{Phase Correction and Quantum Fisher Information}

We now develop the phase correction protocol and establish fundamental precision limits.

\begin{lemma}[$R_Z$ Phase Shift -- Detailed Derivation]
\label{lem:rz-shift-full}
Consider $N$ qubits evolving under external field Hamiltonian $U_k = e^{-i\omega_k Z_k}$. Applying $R_Z(-2\theta) = e^{+i\theta Z}$ (using standard gate convention $R_Z(\alpha) = e^{-i\alpha Z/2}$ with argument $-2\theta$) to all qubits before evolution produces an effective phase shift:
\begin{equation}
U \cdot R_Z(-2\theta)^{\otimes N} = \exp\left(-i\sum_{k=1}^N (\omega_k - \theta) Z_k\right).
\end{equation}
Thus, the effective rotation angles become $\omega'_k = \omega_k - \theta$.
\end{lemma}

\begin{proof}
For a single qubit, applying $R_Z(-2\theta) = e^{+i\theta Z}$ before evolution $U = e^{-i\omega Z}$ gives:
\begin{equation}
U \cdot R_Z(-2\theta) = e^{-i\omega Z} \cdot e^{+i\theta Z}.
\end{equation}

Since both operators involve only the Pauli $Z$ matrix, they commute. Using the identity $e^A e^B = e^{A+B}$ for commuting operators:
\begin{equation}
e^{-i\omega Z} \cdot e^{+i\theta Z} = e^{(-i\omega + i\theta) Z} = e^{-i(\omega - \theta) Z}.
\end{equation}

For the $N$-qubit system, each qubit undergoes an independent transformation since $Z_k$ operators on different qubits commute:
\begin{align}
U \cdot R_Z(-2\theta)^{\otimes N} &= \left(\bigotimes_{k=1}^N e^{-i\omega_k Z_k}\right) \cdot \left(\bigotimes_{k=1}^N e^{+i\theta Z_k}\right)\\
&= \bigotimes_{k=1}^N \left(e^{-i\omega_k Z_k} \cdot e^{+i\theta Z_k}\right)\\
&= \bigotimes_{k=1}^N e^{-i(\omega_k - \theta) Z_k}\\
&= \exp\left(-i\sum_{k=1}^N (\omega_k - \theta) Z_k\right).
\end{align}

Thus, applying $R_Z(-2\theta)$ before evolution shifts each effective rotation angle from $\omega_k$ to $\omega_k - \theta$, as claimed.
\end{proof}

\begin{proposition}[Binary Decision Rule]
Consider $N$ qubits prepared in the GHZ state and subjected to:
\begin{enumerate}
\item Phase correction: $R_Z(-2\theta_{\text{test}})^{\otimes N}$
\item Signal evolution: $U_k = e^{-i\omega_k Z_k}$ where $\omega_k = \omega_{\text{true}} + \epsilon_k$
\item Noise: $\epsilon_k \sim \mathcal{N}(0, \sigma_\epsilon^2)$ i.i.d.
\end{enumerate}
followed by X-parity measurement. The probability of outcome $+1$ for fixed noise realization is:
\begin{equation}
\Pr(+1 | \omega_{\text{true}}, \theta_{\text{test}}, S) = \cos^2\left(N(\omega_{\text{true}} - \theta_{\text{test}}) + S\right),
\end{equation}
where $S = \sum_{k=1}^N \epsilon_k$ is the total noise offset.

In the noise-averaged regime:
\begin{equation}
\overline{\Pr}(+1) = \frac{1}{2}\left[1 + \cos(2N\Delta\omega) e^{-2N\sigma_\epsilon^2}\right],
\end{equation}
where $\Delta\omega := \omega_{\text{true}} - \theta_{\text{test}}$. This creates a reliable decision boundary for binary search with discrimination power scaling as $N^2(\Delta\omega)^2$.
\end{proposition}

\begin{proof}
By Lemma~\ref{lem:rz-shift-full}, applying the phase correction $R_Z(-2\theta_{\text{test}})$ shifts all rotation angles by $-\theta_{\text{test}}$. The effective angles become:
\begin{equation}
\omega'_k = \omega_k - \theta_{\text{test}} = (\omega_{\text{true}} + \epsilon_k) - \theta_{\text{test}} = \Delta\omega + \epsilon_k,
\end{equation}
where $\Delta\omega = \omega_{\text{true}} - \theta_{\text{test}}$ is the error in the test value.

The X-parity measurement now probes these shifted angles. Following the analysis in Lemma~\ref{lem:parity-measurement-full}, the probability of measuring $+1$ for fixed noise realization $\{\epsilon_k\}$ is:
\begin{equation}
\Pr(+1 | \Delta\omega, S) = \cos^2\left(\sum_{k=1}^N \omega'_k\right) = \cos^2\left(\sum_{k=1}^N (\Delta\omega + \epsilon_k)\right),
\end{equation}
which simplifies to:
\begin{equation}
\Pr(+1 | \Delta\omega, S) = \cos^2\left(N\Delta\omega + S\right),
\end{equation}
where $S = \sum_{k=1}^N \epsilon_k$, establishing the first claim.

For the noise-averaged probability, we marginalize over $S \sim \mathcal{N}(0, N\sigma_\epsilon^2)$. By Lemma~\ref{lem:noise-marginalization-full}, substituting $\omega \to \Delta\omega$:
\begin{equation}
\overline{\Pr}(+1) = \frac{1}{2}\left[1 + \cos(2N\Delta\omega) e^{-2N\sigma_\epsilon^2}\right],
\end{equation}
establishing the second claim.

To establish the discrimination power, we compute the derivative with respect to the true parameter:
\begin{equation}
\frac{\partial \overline{\Pr}(+1)}{\partial \omega_{\text{true}}} = \frac{\partial}{\partial \Delta\omega}\left[\frac{1}{2}\left(1 + \cos(2N\Delta\omega) e^{-2N\sigma_\epsilon^2}\right)\right] \cdot \frac{\partial \Delta\omega}{\partial \omega_{\text{true}}}.
\end{equation}

Since $\Delta\omega = \omega_{\text{true}} - \theta_{\text{test}}$ and $\theta_{\text{test}}$ is fixed during a given measurement:
\begin{equation}
\frac{\partial \Delta\omega}{\partial \omega_{\text{true}}} = 1.
\end{equation}

The derivative of the probability is:
\begin{equation}
\frac{\partial \overline{\Pr}(+1)}{\partial \omega_{\text{true}}} = -N\sin(2N\Delta\omega) e^{-2N\sigma_\epsilon^2}.
\end{equation}

For small $\Delta\omega$, using $\sin(2N\Delta\omega) = 2N\Delta\omega + O((N\Delta\omega)^3)$:
\begin{equation}
\frac{\partial \overline{\Pr}(+1)}{\partial \omega_{\text{true}}} = -2N^2\Delta\omega \cdot e^{-2N\sigma_\epsilon^2} + O((\Delta\omega)^3).
\end{equation}

The magnitude of this derivative gives the discrimination power:
\begin{equation}
\left|\frac{\partial \overline{\Pr}(+1)}{\partial \omega_{\text{true}}}\right| = \Theta(N^2|\Delta\omega|),
\end{equation}
demonstrating the $N^2$ scaling characteristic of Heisenberg-limited sensitivity. This quadratic scaling in the number of qubits enables the binary search protocol to efficiently distinguish between parameter values separated by $\Delta\omega = \Theta(1/N)$.
\end{proof}

\begin{theorem}[Quantum Fisher Information]
\label{thm:qfi-binary-search}
Consider $N$ qubits prepared in the GHZ state and subjected to phase evolution where each qubit experiences rotation angle $\omega_k = \omega + \epsilon_k$, with $\omega$ the uniform signal parameter and $\epsilon_k \sim \mathcal{N}(0, \sigma_\epsilon^2)$ i.i.d. noise terms.

\textbf{Noiseless case} ($\sigma_\epsilon = 0$): The quantum Fisher information at $\omega = 0$ is
\begin{equation}
F_Q(0) = 4N^2.
\end{equation}

\textbf{With noise, small-angle regime}: The noise-averaged quantum Fisher information is
\begin{equation}
\overline{F}_Q(0) = 4N^2 e^{-4N\sigma_\epsilon^2} = 4N^2(1 - 4N\sigma_\epsilon^2 + O((N\sigma_\epsilon^2)^2)).
\end{equation}

\textbf{Cramér-Rao bound}: For $M$ independent measurements, any unbiased estimator $\hat{\omega}$ satisfies
\begin{equation}
\operatorname{Var}(\hat{\omega}) \geq \frac{1}{M \cdot \overline{F}_Q(0)} = \Theta\left(\frac{e^{4N\sigma_\epsilon^2}}{MN^2}\right).
\end{equation}

\textbf{Heisenberg scaling}: When $N\sigma_\epsilon^2 = o(1)$, the achievable variance is $\operatorname{Var}(\hat{\omega}) = \Theta(1/(MN^2))$, demonstrating quadratic advantage over uncorrelated measurements.
\end{theorem}

\begin{proof}
We compute the quantum Fisher information using the classical Fisher information of the X-parity measurement. For the noiseless case with outcome probabilities from Lemma~\ref{lem:parity-measurement-full}:
\begin{align}
\Pr(v = 0 | \omega) &= \cos^2(N\omega) = \frac{1 + \cos(2N\omega)}{2},\\
\Pr(v = 1 | \omega) &= \sin^2(N\omega) = \frac{1 - \cos(2N\omega)}{2},
\end{align}
where $v = 0$ corresponds to X-parity outcome $+1$ and $v = 1$ to outcome $-1$.

The classical Fisher information for a two-outcome measurement is defined as:
\begin{equation}
F(\omega) = \sum_{v \in \{0,1\}} \frac{1}{\Pr(v|\omega)} \left[\frac{\partial \Pr(v|\omega)}{\partial \omega}\right]^2.
\end{equation}

Computing the derivatives:
\begin{align}
\frac{\partial \Pr(v=0|\omega)}{\partial \omega} &= \frac{\partial}{\partial \omega}\left[\frac{1 + \cos(2N\omega)}{2}\right]\\
&= \frac{1}{2} \cdot (-2N\sin(2N\omega)) = -N\sin(2N\omega),\\
\frac{\partial \Pr(v=1|\omega)}{\partial \omega} &= \frac{\partial}{\partial \omega}\left[\frac{1 - \cos(2N\omega)}{2}\right]\\
&= \frac{1}{2} \cdot 2N\sin(2N\omega) = N\sin(2N\omega).
\end{align}

Substituting into the Fisher information formula:
\begin{align}
F(\omega) &= \frac{1}{\Pr(v=0|\omega)}\left[-N\sin(2N\omega)\right]^2 + \frac{1}{\Pr(v=1|\omega)}\left[N\sin(2N\omega)\right]^2\\
&= N^2\sin^2(2N\omega)\left[\frac{1}{\Pr(v=0|\omega)} + \frac{1}{\Pr(v=1|\omega)}\right]\\
&= N^2\sin^2(2N\omega)\left[\frac{2}{1 + \cos(2N\omega)} + \frac{2}{1 - \cos(2N\omega)}\right].
\end{align}

Combining the fractions:
\begin{align}
F(\omega) &= 2N^2\sin^2(2N\omega) \cdot \frac{(1-\cos(2N\omega)) + (1+\cos(2N\omega))}{(1 + \cos(2N\omega))(1 - \cos(2N\omega))}\\
&= 2N^2\sin^2(2N\omega) \cdot \frac{2}{1 - \cos^2(2N\omega)}\\
&= 2N^2\sin^2(2N\omega) \cdot \frac{2}{\sin^2(2N\omega)}\\
&= 4N^2,
\end{align}
where we used the identity $1 - \cos^2(\theta) = \sin^2(\theta)$.

Remarkably, the classical Fisher information is constant $F(\omega) = 4N^2$ for all $\omega$ (as long as $\sin(2N\omega) \neq 0$).

For GHZ states measuring a collective phase via X-parity, the quantum Fisher information equals the classical Fisher information when the measurement is optimal:
\begin{equation}
F_Q = F = 4N^2.
\end{equation}

This achieves the Heisenberg limit where the QFI scales as $N^2$ rather than $N$. The factor of 4 arises from the doubled argument $2N\omega$ in the probability formula (the relative phase between the GHZ branches is $2\theta$ where $\theta = N\omega$). This establishes the first claim.

With noise averaging, the probabilities become (from Lemma~\ref{lem:noise-marginalization-full}):
\begin{equation}
\overline{\Pr}(v=0|\omega) = \frac{1 + \cos(2N\omega) e^{-2N\sigma_\epsilon^2}}{2}.
\end{equation}

Following similar calculations with the exponential damping factor, let $c = e^{-2N\sigma_\epsilon^2}$ for brevity:
\begin{align}
\overline{\Pr}(v=0|\omega) &= \frac{1 + c\cos(2N\omega)}{2},\\
\overline{\Pr}(v=1|\omega) &= \frac{1 - c\cos(2N\omega)}{2}.
\end{align}

The derivatives are:
\begin{align}
\frac{\partial \overline{\Pr}(v=0|\omega)}{\partial \omega} &= -cN\sin(2N\omega),\\
\frac{\partial \overline{\Pr}(v=1|\omega)}{\partial \omega} &= cN\sin(2N\omega).
\end{align}

The Fisher information becomes:
\begin{align}
\overline{F}(\omega) &= c^2N^2\sin^2(2N\omega)\left[\frac{2}{1 + c\cos(2N\omega)} + \frac{2}{1 - c\cos(2N\omega)}\right]\\
&= 2c^2N^2\sin^2(2N\omega) \cdot \frac{2}{1 - c^2\cos^2(2N\omega)}\\
&= \frac{4c^2N^2\sin^2(2N\omega)}{1 - c^2\cos^2(2N\omega)}.
\end{align}

At the optimal operating point where $\sin^2(2N\omega)$ is maximized ($\cos(2N\omega) = 0$, i.e., $2N\omega = \pi/2$), we have $\sin^2(2N\omega) = 1$ and:
\begin{equation}
\overline{F}_{\max} = \frac{4c^2N^2}{1 - 0} = 4c^2N^2 = 4N^2 e^{-4N\sigma_\epsilon^2}.
\end{equation}

For small noise $N\sigma_\epsilon^2 = o(1)$:
\begin{equation}
\overline{F}_{\max} = 4N^2(1 - 4N\sigma_\epsilon^2 + O((N\sigma_\epsilon^2)^2)).
\end{equation}

The quantum Fisher information equals the classical Fisher information at the optimal point:
\begin{equation}
\overline{F}_Q = \overline{F}_{\max} = 4N^2 e^{-4N\sigma_\epsilon^2},
\end{equation}
establishing the second claim.

The Cramér-Rao bound states that for any unbiased estimator $\hat{\omega}$ based on $M$ independent measurements:
\begin{equation}
\operatorname{Var}(\hat{\omega}) \geq \frac{1}{M \cdot F_Q} = \frac{e^{4N\sigma_\epsilon^2}}{4MN^2},
\end{equation}
establishing the third claim.

For the Heisenberg scaling regime where $N\sigma_\epsilon^2 = o(1)$, we have $e^{4N\sigma_\epsilon^2} = 1 + o(1)$, giving:
\begin{equation}
\operatorname{Var}(\hat{\omega}) \geq \frac{1}{4MN^2} = \Omega\left(\frac{1}{MN^2}\right),
\end{equation}
demonstrating the characteristic $1/N^2$ scaling of Heisenberg-limited metrology. This is a quadratic improvement over the standard quantum limit where uncorrelated measurements of $N$ qubits achieve variance $\operatorname{Var}(\hat{\omega})_{\text{SQL}} = O(1/N)$, establishing the fourth claim.
\end{proof}

\subsection{Binary Search Convergence}

We now establish the convergence properties of the adaptive binary search protocol.

\begin{lemma}[Majority Vote Correctness]
\label{lem:majority-vote-full}
Consider the binary search protocol at iteration $t$ with test phase $\omega_{\text{test}}$. Suppose the true phase $\omega_{\text{true}}$ satisfies $|\omega_{\text{true}} - \omega_{\text{test}}| \geq \Delta_{\min}$ where $\Delta_{\min}$ is the minimum reliably distinguishable phase difference.

With $M$ independent parity measurements per iteration and majority voting, the probability of an incorrect decision is bounded by
\begin{equation}
P(\text{majority vote error}) \leq \exp\left(-M \cdot D_{\mathrm{KL}}(1/2 \parallel p)\right),
\label{eq:chernoff-bound-full}
\end{equation}
where $p = \overline{\Pr}(+1 \mid \omega_{\text{true}} - \omega_{\text{test}})$ is the noise-averaged probability from Lemma~\ref{lem:noise-marginalization-full}, and
\begin{equation}
D_{\mathrm{KL}}(1/2 \parallel p) = \frac{1}{2}\ln\frac{1}{2p} + \frac{1}{2}\ln\frac{1}{2(1-p)} = -\frac{1}{2}\ln(4p(1-p))
\end{equation}
is the Kullback-Leibler divergence between the uniform distribution and Bernoulli($p$).

In the small-angle regime where $\Delta\omega = \omega_{\text{true}} - \omega_{\text{test}}$ is small, the KL divergence is
\begin{equation}
D_{\mathrm{KL}}(1/2 \parallel p) = 2 \left[\frac{N^2(\Delta\omega)^2}{4}\right]^2 e^{-4N\sigma_\epsilon^2} + O((\Delta\omega)^6) = \frac{N^4(\Delta\omega)^4}{8} e^{-4N\sigma_\epsilon^2} + O((\Delta\omega)^6).
\end{equation}

To achieve error probability $\leq \delta$ per iteration with phase separation $\Delta\omega$, it suffices to use
\begin{equation}
M \geq \frac{\ln(1/\delta)}{D_{\mathrm{KL}}(1/2 \parallel p)} = O\left(\frac{(1 + 4N\sigma_\epsilon^2)\ln(1/\delta)}{N^4(\Delta\omega)^4}\right)
\end{equation}
measurements.
\end{lemma}

\begin{proof}
Each parity measurement is an independent Bernoulli trial with success probability $p = \overline{\Pr}(+1 \mid \Delta\omega)$ where $\Delta\omega := \omega_{\text{true}} - \omega_{\text{test}}$. Let $X_1, \ldots, X_M$ be i.i.d. Bernoulli random variables with $\Pr(X_i = 1) = p$ and $\Pr(X_i = 0) = 1-p$.

The majority vote decision is based on the empirical average:
\begin{equation}
\bar{X} = \frac{1}{M}\sum_{i=1}^M X_i.
\end{equation}

If the true probability satisfies $p > 1/2$ (meaning $\omega_{\text{true}} < \omega_{\text{test}}$ based on Proposition~\ref{prop:activation-function-full}), then the majority vote should yield outcome $+1$ (i.e., $\bar{X} > 1/2$). A majority vote error occurs when $\bar{X} \leq 1/2$.

By the method of types for binary hypothesis testing~\cite[Theorem~11.1.4]{cover2006elements}, the probability that the empirical distribution of i.i.d.\ samples from distribution $p$ is closer to a different distribution $q$ than to $p$ decays exponentially:
\begin{equation}
\Pr\left(\text{empirical distribution closer to } q \text{ than to } p\right) \leq \exp\left(-M \cdot D_{\mathrm{KL}}(q \parallel p)\right).
\end{equation}

For our binary case, when $p \neq 1/2$, the probability of majority error (empirical average on the wrong side of $1/2$) is bounded by:
\begin{equation}
\Pr(\text{majority error}) \leq \exp\left(-M \cdot D_{\mathrm{KL}}(1/2 \parallel p)\right),
\end{equation}
establishing Equation~\eqref{eq:chernoff-bound-full}.

For the small-angle regime calculation, we use the Taylor expansion from Lemma~\ref{lem:noise-marginalization-full}:
\begin{equation}
p = \overline{\Pr}(+1 \mid \Delta\omega) = \frac{1}{2} - \frac{N^2(\Delta\omega)^2}{4} e^{-2N\sigma_\epsilon^2} + O(\Delta\omega^4).
\end{equation}

Define $\epsilon_p := p - 1/2$. For $|\epsilon_p| = o(1)$, we expand the KL divergence using $4p(1-p) = 1 - 4\epsilon_p^2$:
\begin{equation}
D_{\mathrm{KL}}(1/2 \parallel p) = -\frac{1}{2}\ln(1 - 4\epsilon_p^2) = 2\epsilon_p^2 + 8\epsilon_p^4 + O(\epsilon_p^6),
\end{equation}
where we used $-\ln(1-x) = x + x^2/2 + O(x^3)$ with $x = 4\epsilon_p^2$.

From Lemma~\ref{lem:noise-marginalization-full}, in the small-angle regime:
\begin{equation}
\epsilon_p = p - \frac{1}{2} = - \frac{N^2(\Delta\omega)^2}{4} e^{-2N\sigma_\epsilon^2} + O((\Delta\omega)^4).
\end{equation}

Therefore:
\begin{equation}
\epsilon_p^2 = \frac{N^4(\Delta\omega)^4}{16} e^{-4N\sigma_\epsilon^2} + O((\Delta\omega)^6).
\end{equation}

Substituting:
\begin{equation}
D_{\mathrm{KL}}(1/2 \parallel p) = 2 \cdot \frac{N^4(\Delta\omega)^4}{16} e^{-4N\sigma_\epsilon^2} + O((\Delta\omega)^6) = \frac{N^4(\Delta\omega)^4}{8} e^{-4N\sigma_\epsilon^2} + O((\Delta\omega)^6),
\end{equation}
establishing the approximation formula.

The minimum distinguishable phase $\Delta_{\min}$ is determined by requiring $D_{\mathrm{KL}} = \Theta(1)$ so that with $O(1)$ measurements, we can reliably distinguish the two cases. Setting:
\begin{equation}
\frac{N^4 \Delta_{\min}^4}{8} e^{-4N\sigma_\epsilon^2} = \Theta(1),
\end{equation}
we obtain:
\begin{equation}
N^4 \Delta_{\min}^4 = \Theta\left(8 e^{4N\sigma_\epsilon^2}\right).
\end{equation}

For the noiseless case ($\sigma_\epsilon = 0$):
\begin{equation}
\Delta_{\min} = \Theta\left(\frac{8^{1/4}}{N}\right) = \Theta(1/N),
\end{equation}
which is precisely the Heisenberg limit. The constant $1/8$ determines the KL divergence scaling in Algorithm~\ref{alg:binary-search}.

To achieve per-iteration error probability $\leq \delta$, we require:
\begin{equation}
\exp(-M \cdot D_{\mathrm{KL}}) \leq \delta \quad \Rightarrow \quad M \geq \frac{\ln(1/\delta)}{D_{\mathrm{KL}}}.
\end{equation}

Substituting the expression for $D_{\mathrm{KL}}$ and using $e^{4N\sigma_\epsilon^2} = 1 + 4N\sigma_\epsilon^2 + O((N\sigma_\epsilon^2)^2)$ for moderate noise:
\begin{equation}
M \geq 8\ln(1/\delta) \cdot \frac{e^{4N\sigma_\epsilon^2}}{N^4(\Delta\omega)^4} = O\left(\frac{(1 + 4N\sigma_\epsilon^2)\ln(1/\delta)}{N^4(\Delta\omega)^4}\right),
\end{equation}
establishing the final claim.
\end{proof}

\begin{theorem}[Binary Search Convergence and Resources]
\label{thm:convergence-full}
The binary search protocol (Algorithm~\ref{alg:binary-search}) achieves target precision $\epsilon > 0$ with success probability $\geq 1 - \delta$ using the following resources:

\textbf{1. Number of iterations:}
\begin{equation}
T = \left\lceil \log_2\left(\frac{\Omega_0}{\epsilon}\right) \right\rceil = O(\log(1/\epsilon)),
\end{equation}
where $\Omega_0 = \pi/(4N)$ is the initial search interval width.

\textbf{2. Measurements per iteration:}
\begin{equation}
M = O\left(\frac{\log(T/\delta)}{D_{\mathrm{KL}}}\right) = O(\log\log(1/\epsilon) \cdot \log(1/\delta)),
\end{equation}
where $D_{\mathrm{KL}}$ is the KL divergence at the minimum distinguishable separation from Lemma~\ref{lem:majority-vote-full}.

\textbf{3. Total parity measurements:}
\begin{equation}
M_{\text{total}} = M \cdot T = O(\log(1/\epsilon) \log\log(1/\epsilon) \log(1/\delta)).
\end{equation}

\textbf{4. Total qubit-experiments:}
\begin{equation}
N_{\text{total}} = N \cdot M \cdot T = O(N \log(1/\epsilon) \log\log(1/\epsilon) \log(1/\delta)).
\end{equation}

\textbf{5. Heisenberg scaling regime:} For target precision $\epsilon = c/(aN)$ where $a, c$ are positive constants:
\begin{equation}
T = O(1), \quad N_{\text{total}} = O(N \log(1/\delta)),
\end{equation}
matching the Heisenberg limit $\epsilon \propto 1/N$ with logarithmic overhead.

\textbf{6. Variance bound:} The final estimate $\hat{\omega}$ satisfies
\begin{equation}
\operatorname{Var}(\hat{\omega}) = O(\epsilon^2), \quad \epsilon^2 \cdot M \cdot T \cdot \overline{F}_Q = O(\log^2(1/\epsilon)\log(1/\delta)),
\end{equation}
where $\overline{F}_Q$ is the quantum Fisher information from Theorem~\ref{thm:qfi-binary-search}. The protocol achieves the quantum Cramér-Rao bound $\operatorname{Var}_{\text{CR}} = \Theta(1/(M \cdot T \cdot \overline{F}_Q))$ up to poly-logarithmic factors.
\end{theorem}

\begin{proof}
We establish the six claims through a sequence of resource and scaling arguments.

At each iteration, the binary search protocol bisects the current interval. Starting with interval width $\Omega_0$, after $t$ iterations the interval width is:
\begin{equation}
\Omega_t = \frac{\Omega_0}{2^t}.
\end{equation}

The protocol terminates when the interval width falls below the target precision:
\begin{equation}
\Omega_T \leq \epsilon.
\end{equation}

This requires:
\begin{equation}
\frac{\Omega_0}{2^T} \leq \epsilon \quad \Rightarrow \quad 2^T \geq \frac{\Omega_0}{\epsilon} \quad \Rightarrow \quad T \geq \log_2\left(\frac{\Omega_0}{\epsilon}\right).
\end{equation}

Taking the ceiling to ensure an integer number of iterations:
\begin{equation}
T = \left\lceil \log_2\left(\frac{\Omega_0}{\epsilon}\right) \right\rceil,
\end{equation}
which establishes Claim 1. Since $\Omega_0 = \pi/(4N)$ is a constant for fixed $N$, we have $T = O(\log(1/\epsilon))$.

From Lemma~\ref{lem:majority-vote-full}, each iteration requires sufficient measurements to ensure the majority vote correctly identifies which half-interval contains the true phase. To achieve per-iteration error probability $\leq \delta_{\text{iter}}$, we require:
\begin{equation}
M \cdot D_{\mathrm{KL}} \geq \ln(1/\delta_{\text{iter}}) \quad \Rightarrow \quad M \geq \frac{\ln(1/\delta_{\text{iter}})}{D_{\mathrm{KL}}}.
\end{equation}

To ensure the total error over all $T$ iterations is at most $\delta$, we use the union bound: set $\delta_{\text{iter}} = \delta/T$. Then:
\begin{equation}
M \geq \frac{\ln(T/\delta)}{D_{\mathrm{KL}}}.
\end{equation}

In the Heisenberg regime where the phase separation at each iteration is $\Delta\omega = \Theta(1/N)$, the KL divergence from Lemma~\ref{lem:majority-vote-full} is:
\begin{equation}
D_{\mathrm{KL}} = \Theta\left(\frac{N^4}{8} \cdot \frac{1}{N^4}\right) = \Theta\left(\frac{1}{8}\right) = \Theta(1).
\end{equation}

Therefore:
\begin{equation}
M = O(\ln(T/\delta)) = O(\ln T + \ln(1/\delta)) = O(\log\log(1/\epsilon) + \log(1/\delta)),
\end{equation}
where we used $\ln T = \ln(\log_2(\Omega_0/\epsilon)) = O(\log\log(1/\epsilon))$. This establishes Claim 2.

Each iteration uses $M$ measurements. Over $T$ iterations, the total number of parity measurements is:
\begin{equation}
M_{\text{total}} = T \times M = O(\log(1/\epsilon)) \times O(\log\log(1/\epsilon) \log(1/\delta)).
\end{equation}

Simplifying:
\begin{equation}
M_{\text{total}} = O(\log(1/\epsilon) \log\log(1/\epsilon) \log(1/\delta)),
\end{equation}
establishing Claim 3.

Each measurement uses $N$ qubits. The total number of qubit-experiments is:
\begin{equation}
N_{\text{total}} = N \times M_{\text{total}} = O(N \log(1/\epsilon) \log\log(1/\epsilon) \log(1/\delta)),
\end{equation}
establishing Claim 4.

We now analyze the special case where the target precision is set to match the fundamental Heisenberg limit: $\epsilon = c/(aN)$ for positive constants $a, c$.

With the initial interval $\Omega_0 = \pi/(4N)$, the number of iterations becomes:
\begin{equation}
T = \left\lceil \log_2\left(\frac{\pi/(4N)}{c/(aN)}\right) \right\rceil = \left\lceil \log_2\left(\frac{a\pi}{4c}\right) \right\rceil.
\end{equation}

Since $a, c$ are constants independent of $N$, we have:
\begin{equation}
T = O(1).
\end{equation}

The number of iterations is independent of $N$ when operating at the Heisenberg limit.

At the final iteration, the interval width is $\Omega_T = \Theta(\epsilon) = \Theta(1/N)$, giving phase separation $\Delta\omega = \Theta(1/N)$ between the midpoint and the true value. From Lemma~\ref{lem:majority-vote-full}, the KL divergence at this separation is:
\begin{equation}
D_{\mathrm{KL}} = \Theta\left(\frac{N^4}{8} \cdot \frac{1}{N^4}\right) = \Theta\left(\frac{1}{8}\right) = \Theta(1).
\end{equation}

Therefore:
\begin{equation}
M = O\left(\frac{\log(T/\delta)}{D_{\mathrm{KL}}}\right) = O(\log(1/\delta)),
\end{equation}
since $T = O(1)$ means $\log T = O(1)$.

The total qubit-experiments are:
\begin{equation}
N_{\text{total}} = N \cdot M \cdot T = N \cdot O(\log(1/\delta)) \cdot O(1) = O(N \log(1/\delta)),
\end{equation}
establishing Claim 5.

The final estimate $\hat{\omega}$ is the midpoint of the final interval $[\Omega_{\text{low}}, \Omega_{\text{high}}]$ with width $\Omega_T \leq \epsilon$. With high probability $\geq 1 - \delta$ (ensured by the union bound over all iterations), the true phase $\omega$ lies within this interval.

The maximum estimation error is half the interval width:
\begin{equation}
|\hat{\omega} - \omega| \leq \frac{\Omega_T}{2} \leq \frac{\epsilon}{2}.
\end{equation}

Therefore:
\begin{equation}
\operatorname{Var}(\hat{\omega}) \leq \EX[(\hat{\omega} - \omega)^2] \leq \left(\frac{\epsilon}{2}\right)^2 = \frac{\epsilon^2}{4}.
\end{equation}

Comparing to the Cramér-Rao bound with $M_{\text{total}} = M \cdot T$ total measurements:
\begin{equation}
\operatorname{Var}_{\text{CR}} \geq \frac{1}{M_{\text{total}} \cdot \overline{F}_Q} = \frac{1}{M \cdot T \cdot \overline{F}_Q}.
\end{equation}

From Theorem~\ref{thm:qfi-binary-search}:
\begin{equation}
\overline{F}_Q = 4N^2 e^{-4N\sigma_\epsilon^2},
\end{equation}
so:
\begin{equation}
\operatorname{Var}_{\text{CR}} \geq \frac{e^{4N\sigma_\epsilon^2}}{4M T N^2}.
\end{equation}

Our protocol achieves variance:
\begin{equation}
\operatorname{Var}(\hat{\omega}) \leq \epsilon^2 = \frac{\Omega_0^2}{4^T} = \frac{(\pi/(4N))^2}{4^T}.
\end{equation}

Since $2^T = \Theta(\Omega_0/\epsilon) = \Theta((\pi/(4N))/\epsilon)$, we have $4^T = 2^{2T} = \Theta((\Omega_0/\epsilon)^2)$, giving:
\begin{equation}
\operatorname{Var}(\hat{\omega}) = \Theta\left(\frac{\Omega_0^2}{\Omega_0^2/\epsilon^2}\right) = \Theta(\epsilon^2).
\end{equation}

The ratio to the Cramér-Rao bound is:
\begin{equation}
\frac{\operatorname{Var}(\hat{\omega})}{\operatorname{Var}_{\text{CR}}} = O\left(\epsilon^2 \cdot MTN^2 \cdot e^{4N\sigma_\epsilon^2}\right).
\end{equation}

With $M = O(\log T \log(1/\delta))$ and $T = O(\log(1/\epsilon))$:
\begin{equation}
MTN^2 = O(N^2 \log(1/\epsilon) \log\log(1/\epsilon) \log(1/\delta)).
\end{equation}

For $\epsilon = \Theta(1/N)$ (Heisenberg regime) and $N\sigma_\epsilon^2 = o(1)$ so that $e^{4N\sigma_\epsilon^2} = O(1)$:
\begin{equation}
MTN^2 = O(N^2 \log N \log\log N \log(1/\delta)),
\end{equation}
and the ratio to the Cramér-Rao bound is:
\begin{equation}
\frac{\operatorname{Var}(\hat{\omega})}{\operatorname{Var}_{\text{CR}}} = O\left(\frac{1}{N^2} \cdot N^2 \log N \log\log N \log(1/\delta)\right) = O(\log N \log\log N \log(1/\delta)).
\end{equation}
This demonstrates that the protocol achieves the quantum Cramér-Rao bound up to poly-logarithmic factors $\log^2(1/\epsilon)$, establishing Claim 6.
\end{proof}

\subsection{Three-Category Classification Analysis}

The following results establish the reliability and information content of the three-category classification scheme (Definition~\ref{def:three-category}) used in the repetition code implementation (Section~\ref{subsec:qsp-implementation}). Let $\Phi_G(x) = \int_{-\infty}^x \frac{1}{\sqrt{2\pi}} e^{-t^2/2}\, dt$ denote the standard Gaussian CDF.

\begin{proposition}[Classification Reliability]
\label{prop:qsp-classification}
Under the per-experiment noise model with $\phi = \omega + \epsilon$ where $\epsilon \sim \mathcal{N}(0, \sigma_\epsilon^2)$, the Chernoff bound for Gaussian tails $\Phi_G(-k) \leq \frac{1}{2}e^{-k^2/2}$ yields:
\begin{enumerate}
\item If $\omega > \tau + k\sigma_\epsilon$, then $\Pr(\textsc{Middle} \mid \omega) \leq \frac{1}{2}e^{-k^2/2}$.
\item If $|\omega| < \tau - k\sigma_\epsilon$, then $\Pr(\textsc{High} \mid \omega) + \Pr(\textsc{Low} \mid \omega) \leq e^{-k^2/2}$.
\end{enumerate}
\end{proposition}

\begin{proof}
For item~1: When $\omega > \tau + k\sigma_\epsilon$, the probability of falling in \textsc{Middle} (i.e., $|\phi| \leq \tau$) requires $\epsilon \leq \tau - \omega < -k\sigma_\epsilon$. By the Gaussian tail bound, $\Pr(\epsilon < -k\sigma_\epsilon) \leq \frac{1}{2}e^{-k^2/2}$, establishing item~1.

For item~2: When $|\omega| < \tau - k\sigma_\epsilon$, reaching \textsc{High} requires $\phi = \omega + \epsilon > \tau$, hence $\epsilon > \tau - \omega > k\sigma_\epsilon$. Similarly for \textsc{Low}. Both tails contribute at most $\frac{1}{2}e^{-k^2/2}$ each, establishing item~2.
\end{proof}

\begin{proposition}[QSP KL Divergence]
\label{prop:qsp-kl}
Conditioning on non-\textsc{Middle} outcomes, define:
\begin{equation}
\tilde{p} := \frac{\Pr(\textsc{High} \mid \omega)}{\Pr(\textsc{High} \mid \omega) + \Pr(\textsc{Low} \mid \omega)}.
\end{equation}
The majority rule error rate is governed by $D_{\text{KL}} = D_{\text{KL}}(1/2 \| \tilde{p})$. When $|\omega - \theta| = \Omega(\sigma_\epsilon)$ and $\tau = \Theta(\sigma_\epsilon)$:
\begin{equation}
D_{\text{KL}} = \Theta\left(\frac{(\omega - \theta)^2}{\sigma_\epsilon^2}\right).
\label{eq:qsp-kl-scaling}
\end{equation}
\end{proposition}

\begin{proof}
With phase shift $\theta$, the effective phase is $\omega - \theta + \epsilon$ where $\epsilon \sim \mathcal{N}(0, \sigma_\epsilon^2)$. The category probabilities are $\Pr(\textsc{High}) = 1 - \Phi_G((\tau - (\omega - \theta))/\sigma_\epsilon)$ and $\Pr(\textsc{Low}) = \Phi_G((-\tau - (\omega - \theta))/\sigma_\epsilon)$. Thus $\tilde{p} = \Pr(\textsc{High}) / (\Pr(\textsc{High}) + \Pr(\textsc{Low}))$. When $\omega - \theta$ is small relative to $\sigma_\epsilon$, expanding the CDF around $\pm\tau/\sigma_\epsilon$ gives:
\begin{equation}
\tilde{p} - \frac{1}{2} = \Theta\left(\frac{\omega - \theta}{\sigma_\epsilon}\right),
\end{equation}
since the linear term in $(\omega - \theta)/\sigma_\epsilon$ dominates. The KL divergence expands as $D_{\text{KL}}(1/2 \| \tilde{p}) = 2(\tilde{p} - 1/2)^2 + O((\tilde{p} - 1/2)^4)$, yielding~\eqref{eq:qsp-kl-scaling}.
\end{proof}

\section{Heterogeneous Noise Analysis}
\label{app:heterogeneous-noise}

We provide a rigorous derivation of the heterogeneous noise correction stated in Proposition~\ref{prop:heterogeneous}, establishing how qubit-to-qubit variability in transverse field strengths affects the quantum Fisher information bounds from Theorem~\ref{thm:bitflip}.

\subsection{Error Probability Under Heterogeneous Noise}

Consider the bit-flip repetition code setup from Section~\ref{sec:fixed-omega} where each qubit evolves under $U_k = e^{-i(\omega Z_k + \gamma_k X_k)}$ with fixed signal $\omega$ but heterogeneous transverse fields $\gamma_k$. From Equation~\eqref{eq:bitflip-decomp}, the probability of an X-error on qubit $k$ is
\begin{equation}
p_k = 1 - \beta_k^2 = \sin^2(\Omega_k) \cdot \frac{\gamma_k^2}{\omega^2 + \gamma_k^2},
\end{equation}
where $\beta_k = \sqrt{\cos^2(\Omega_k) + \sin^2(\Omega_k)\omega^2/(\omega^2+\gamma_k^2)}$ and $\Omega_k = \sqrt{\omega^2 + \gamma_k^2}$.

\subsection{Small Noise Expansion}

When $|\gamma_k/\omega| = o(1)$, we expand $p_k$ in powers of the dimensionless parameter $x_k = \gamma_k/\omega$. Using the binomial expansion $(1+x^2)^{1/2} = 1 + x^2/2 - x^4/8 + O(x^6)$, we have $\Omega_k = \omega\sqrt{1+x_k^2} = \omega(1 + x_k^2/2 + O(x_k^4))$. For $\sin(\omega + \delta)$ with $\delta = \omega x_k^2/2$, Taylor expansion gives $\sin(\omega + \delta) = \sin(\omega) + \delta\cos(\omega) + O(\delta^2) = \sin(\omega) + (\omega x_k^2/2)\cos(\omega) + O(x_k^4)$, yielding $\sin(\Omega_k) = \sin(\omega) + (\omega x_k^2/2)\cos(\omega) + O(x_k^4)$. Squaring: $\sin^2(\Omega_k) = \sin^2(\omega) + \omega\sin(2\omega)x_k^2/2 + O(x_k^4)$. For the fraction, $\gamma_k^2/(\omega^2+\gamma_k^2) = x_k^2/(1+x_k^2) = x_k^2(1 - x_k^2 + O(x_k^4)) = x_k^2 + O(x_k^4)$. Combining these expansions, the error probability becomes
\begin{equation}
p_k = \frac{\sin^2(\omega)}{\omega^2}\gamma_k^2 + O(\gamma_k^4/\omega^4).
\label{eq:pk-taylor}
\end{equation}
This Taylor expansion is valid when $\max_k |\gamma_k/\omega| = o(1)$ and provides the basis for analyzing heterogeneous distributions.

\subsection{Statistical Moments and Expected Error Probability}

Assume $\gamma_k \sim \mathcal{N}(\gamma, (\gamma h)^2)$ independently for each qubit $k$. For a normal distribution with mean $\gamma$ and variance $\sigma^2 = (\gamma h)^2$, the second moment is $\EX[\gamma_k^2] = \sigma^2 + \gamma^2 = \gamma^2(h^2 + 1)$. Taking expectations over the distribution in Equation~\eqref{eq:pk-taylor} and using linearity:
\begin{equation}
\EX[p_k] = \frac{\sin^2(\omega)}{\omega^2}\EX[\gamma_k^2] + O(\gamma^4) = \frac{\sin^2(\omega)}{\omega^2}\gamma^2(1+h^2) + O(\gamma^4).
\label{eq:expected-pk}
\end{equation}
The $(1+h^2)$ factor arises from $\EX[\gamma_k^2] = \gamma^2 + \operatorname{Var}[\gamma_k] = \gamma^2(1+h^2)$, showing that heterogeneity increases average error probability relative to homogeneous noise with $\gamma_k \equiv \gamma$.

\subsection{Impact on Quantum Fisher Information}

From the proof of Theorem~\ref{thm:bitflip}, the expected number of detected errors is $\EX[d] = (2L+1)\EX[p_k]$. Using Equation~\eqref{eq:expected-pk}:
\begin{equation}
\EX[d] = (2L+1)\frac{\sin^2(\omega)}{\omega^2}\gamma^2(1+h^2) + O(\gamma^4).
\end{equation}
Following the QFI calculation in the proof of Theorem~\ref{thm:bitflip}, we have $\EX[F_Q] = 4\EX[(2L+1-d)^2]$. Expanding the expectation as in the homogeneous case and using $\EX[d] = o(2L+1)$ in the small noise regime:
\begin{equation}
\EX[F_Q] = 4(2L+1)^2(1 - 2\EX[p_k] + O(\gamma^4)) = 4(2L+1)^2\left(1 - O(\gamma^2(1+h^2))\right),
\end{equation}
establishing the result of Proposition~\ref{prop:heterogeneous}. The $(1+h^2)$ correction demonstrates that qubit-to-qubit variability degrades the quantum Fisher information relative to perfectly homogeneous noise, with the degradation proportional to both the mean noise level $\gamma^2$ and the relative heterogeneity $h^2$.

This analysis establishes that heterogeneous transverse fields increase average error probability by a factor $(1+h^2)$ relative to homogeneous noise with the same mean. For realistic device parameters ($h \leq 0.5$), the penalty remains modest ($< 25\%$), consistent with the numerical observations in Section~\ref{sec:numerical} showing $<5\%$ change in scaling exponent for heterogeneous configurations. The analysis assumes $|\gamma/\omega| = o(1)$ to justify the Taylor expansion; for larger noise levels, the small-noise approximation breaks down and the observed robustness becomes algorithmic rather than fundamental, as discussed in the numerical results (Section~\ref{sec:numerical}).

\end{document}